\documentclass[11pt]{article}
\usepackage[utf8]{inputenc}
\usepackage[english]{babel}
\usepackage[active]{srcltx}
\usepackage{hyperref}
\usepackage[all]{xy}
\usepackage{amsfonts}
\usepackage{amsthm}
\usepackage{enumerate}
\usepackage{mathrsfs}
\usepackage{multicol}
\usepackage[all]{xy}
\usepackage{amsfonts}
\usepackage{latexsym,amsmath,amssymb}
\usepackage[usenames,dvipsnames]{color}
\usepackage{xfrac}
\usepackage{pstricks}
\usepackage[absolute,overlay]{textpos}
\usepackage{graphics,wrapfig,times}

\usepackage{color}
\definecolor{DarkBlue}{rgb}{0.1,0.1,0.5}
\definecolor{Red}{rgb}{0.9,0.1,0.1}
\definecolor{Green}{rgb}{0.3,0.7,0.0}
\definecolor{green2}{rgb}{0.1,0.7,0.2}
\definecolor{blue2}{rgb}{0.0,0.6,0.7}
\definecolor{pink}{rgb}{1,0.0,1}
\definecolor{orange}{rgb}{0.9,0.0,0.1}

\newtheorem{theo}{Theorem}
\newtheorem{corollary}{Corollary}

\newtheorem{lemma}{Lemma}
\newtheorem{prop}{Proposition}

\newtheorem{definition}{Definition}

\renewcommand{\d}{\mathrm{d}}
\newcommand{\derpar}[2]{\displaystyle\frac{\partial{#1}}{\partial{#2}}}

\newcommand{\Lag}{\mathcal{L}}

\newcommand{\vf}{\mathfrak{X}}
\newcommand{\df}{\Omega}

\newcommand{\Tan}{\mathrm{T}}
\newcommand{\inn}{{\mathop{i}\nolimits}}
\newcommand{\Lie}{\mathop{\mathrm{L}}\nolimits}

\newcommand{\bal}{\begin{align*}}
\newcommand{\eal}{\end{align*}}
\def\beq{\begin{equation}}
\def\eeq{\end{equation}}
\def\bea{\begin{eqnarray}}
\def\eea{\end{eqnarray}}
\def\beann{\begin{eqnarray*}}
\def\eeann{\end{eqnarray*}}
\def\ben{\begin{enumerate}}
\def\een{\end{enumerate}}
\def\bit{\begin{itemize}}
\def\eit{\end{itemize}}

\def\vf{\mathfrak X}
\def\df{{\mit\Omega}}
\def\Lag{{\cal L}}

\def\d{{\rm d}}

\def\p{{\rm p}}

\def\Tan{{\rm T}}
\def\Lie{\mathop{\rm L}\nolimits}
\def\inn{\mathop{i}\nolimits}
\def\Cinfty{{\rm C}^\infty}

\def\tabaddress#1{{\small\it\begin{tabular}[t]{c}#1
\\[1.2ex]\end{tabular}}}

\title{\sc New multisymplectic approach to the Metric-Affine
(Einstein-Palatini) action for gravity}
\author{
{\sc  Jordi Gaset\thanks{{\bf e}-{\it mail}:
   jordi.gaset@upc.edu} },
   {\sc Narciso Rom\'an-Roy\thanks{{\bf e}-{\it mail}:
   narciso.roman@upc.edu  / ORCID: 0000-0003-3663-9861.}}  \\
   \tabaddress{Department of Mathematics.
   Ed. C-3, Campus Norte UPC\\
   C/ Jordi Girona 1. 08034 Barcelona. Spain.}}
   \date{September 24, 2019}

\begin{document}

\maketitle

\pagestyle{myheadings}
\markright{\rm J. Gaset and N. Rom\'an-Roy,
   \sl Multisymplectic approach to the Einstein-Palatini action.}
\maketitle
\thispagestyle{empty}

\begin{abstract}

We present a covariant multisymplectic formulation for the {\sl Einstein-Palatini} (or {\sl Metric-Affine\/}) model of General Relativity
(without energy-matter sources).
As it is described by a first-order affine Lagrangian (in the derivatives of the fields),
it is singular and, hence, this is a gauge field theory with constraints.
These constraints are obtained after applying a {\sl constraint algorithm}
to the field equations, both in the Lagrangian and the Hamiltonian formalisms.
In order to do this, the covariant field equations must be written in a suitable geometrical way, 
using integrable distributions which are represented by multivector fields of a certain type.
We obtain and explain the geometrical and physical meaning 
of the Lagrangian constraints and we construct the 
multimomentum (covariant) Hamiltonian formalism.
The gauge symmetries of the model are discussed in both formalisms and, from them, 
the equivalence with the Einstein-Hilbert model is established.
\end{abstract}

 \bigskip
\noindent {\bf Key words}:
 \textsl{Classical field theories, jet bundles, multisymplectic forms, Einstein-Palatini action, metric-affine models, constraints, gauge symmetries.}

\bigskip
\vbox{\raggedleft AMS s.\,c.\,(2010): \null 
{\it Primary}: 53D42, 55R10, 70S05, 83C05; {\it Secondary}: 49S05, 53C15, 53C80, 53Z05.}\null

\newpage

\tableofcontents

\newpage

\section{Introduction}
\label{intro}

In recent years, there is an increasing effort in understanding the covariant
description of gravitational theories (General Relativity and other derived from it) 
using different kinds of geometric frameworks such as the 
{\sl multisymplectic} or the {\sl polysymplectic manifolds}. Thus, in \cite{BMS-2017,CVB-2006,CreTe-2016,CreTe-2016b,ESG-1995,art:GR-2017,
GMS-97, GIMMSY-mm,Krupka,KrupkaStepanova,rovelli,Sd-95,To-97}
general aspects of the theory are studied in this way, meanwhile other
papers are devoted to consider several particular problems.
For instance, in \cite{first,art:GR-2016,rosado,rosado2} the reduction and projectability 
of higher-order theories (such as the Hilbert-Einstein model) is analized,
in \cite{vey1} the vielbein models of General Relativity
are studied using the multisymplectic formulation
and in  \cite{Ka-q1,Ka-q2,Ka-2016} interesting contributions to the
problem of the precanonical quantization of gravity are done.

The multisymplectic and polysymplectic techniques have been also applied
to treat different aspects of one of the most classical approaches in General Relativity:
the {\sl Einstein-Palatini} or {\sl Metric-Affine model}
\cite{art:Capriotti,art:Capriotti2,IS-2016,MGCD-2017,art:Munoz-rosado2013}.
In particular, in \cite{art:Capriotti2} an exhaustive study of the multisymplectic 
description of the model has been done, using a unified formalism which joins
both the Lagrangian and Hamiltonian formalisms into a single one.
This unified framework had been previously stated to do a covariant multisymplectic 
formulation of the {\sl Hilbert-Einstein} model in General Relativity \cite{art:GR-2017}.

This paper is another contribution in order to complete the 
multisymplectic description of the Einstein-Palatini theory (without energy-matter sources).
In particular, we are especially interested in the following problem:
as a consequence of the degeneracy of the Lagrangian,
this is a premultisymplectic field theory and
the Lagrangian field equations are incompatible in the jet bundle
where the Lagrangian formalism takes place.
The problem of finding a submanifold where this equations have consistent solutions
(if it exists) is solved by applying a constraint algorithm adapted to this
premultisymplectic scenario (see, for instance, \cite{LMM-96,LMMMR-2005}
for a geometric description of these kinds of algorithms).
Our first aim is to implement a local-coordinate version of these algorithms.
In order to do it, the field equations are stated in a more geometrical way, 
as equations for distributions, using certain kinds of multivector fields, and in the last step studying their integrability.
The second objective is to construct the Hamiltonian formalism of the theory
and, then, apply the corresponding constraint algorithm to solve the 
incompatibility of the Hamiltonian field equations.
The constraints arising in both formalisms play a relevant role in describing 
the main features of the theory and, in the Hamiltonian formalism,
the choice of different kinds of coordinates (which have a clear geometric interpretation) 
allows us to better understanding several geometrical characteristics of the formalism.

The Metric-Affine model, as it is currently understood,
appeared first in the 1925 paper of A. Einstein \cite{Einstein}, 
where the author stated that imposing the vanishing of the trace 
of the torsion of the connection, together with the field equations, 
is enough to recover the Levi-Civita connection associated with the metric. 
Later, several authors, like \cite{pons}, pointed out that 
this property is related to the existence of a particular gauge symmetry. 
Another objective of this work is to make a geometrical analysis 
of this gauge freedom and to recover the Einstein-Hilbert model 
for General Relativity by means of a partial gauge fixing. 
A brief discussion on the classical Lagrangian symmetries of the theory 
and their associated currents is also done.

The paper is organized as follows:
Section \ref{Sec2} is devoted to present a brief review on some previous
geometric structures such as on multivector fields and distributions, as well as 
the suitable jet bundle and its corresponding multimomentum bundles 
needed for developing the Lagrangian and the Hamiltonian formalisms of the theory.
Next we describe geometrically the Einstein-Palatini model without energy-matter sources.
First, in Section \ref{Sec3}, the Lagrangian formalism of this theory is studied in detail and
the Lagrangian constraint algorithm is applied by steps,
obtaining the final constraint submanifold where the Lagrangian field 
equations have consistent solutions. 
The geometric interpretation of the different kinds of constraints 
and the gauge and natural Lagrangian symmetries are also discussed here.
Second, in Section \ref{Sec4} the Hamiltonian formalisms is stated and analysed in an analogous 
way, using two different kinds of coordinates.
Finally, the relation with the Einstein-Hilbert model 
is established in Section \ref{se:rel:he-ma},
and it is used to obtain the final constraint submanifold where the 
multivector fields solutions are integrable, 
both in the Lagrangian and the Hamiltonian formalisms.
At the end of the paper, an appendix is included, 
where we state the basic considerations and definitions on the concepts of (Noether) symmetries 
and gauge symmetries for Lagrangian field theories.

All the manifolds are real, second countable and $\Cinfty$. The maps and the structures are $\Cinfty$.  Sum over repeated indices is understood.

\section{Geometric elements}
\label{Sec2}

\subsection{Multivector fields}
\label{mvf}

(See \cite{art:Echeverria_Munoz_Roman98} for details).

\begin{definition}
Let $\tau\colon{\cal M}\to M$ be a fiber bundle.

An {\rm $m$-multivector field} in ${\cal M}$ is a skew-symmetric contravariant 
tensor of order $m$ in ${\cal M}$. The set of $m$-multivector fields 
in ${\cal M}$ is denoted $\vf^m({\cal M})$.

In general, a multivector field $\mathbf{X}\in\vf^m({\cal M})$ is said to be {\rm locally decomposable} if,
for every $p\in {\cal M}$, there is an open neighbourhood  $U_p\subset{\cal M}$
and $X_1,\ldots ,X_m\in\vf (U_p)$ such that $\mathbf{X}\vert_{U_p}=X_1\wedge\ldots\wedge X_m$.

Locally decomposable $m$-multivector fields $\mathbf{X}\in\vf^m({\cal M})$ are locally associated with $m$-dimensional
distributions $D\subset\Tan{\cal M}$, and multivector fields associated with
the same distribution make an {\rm equivalence class} $\{ {\bf X}\}$ in the set $\vf^m({\cal M})$.
Then,
$\mathbf{X}$ is {\rm integrable} if its associated distribution is integrable. 
\end{definition}

For every $\mathbf{X}\in\mathfrak{X}^m({\cal M})$, 
there exist $X_1,\ldots ,X_r\in\mathfrak{X} (U)$ such that
$$
\mathbf{X}\vert_{U}=\sum_{1\leq i_1<\ldots <i_m\leq r} f^{i_1\ldots i_m}X_{i_1}\wedge\ldots\wedge X_{i_m} \, ,
$$
with $f^{i_1\ldots i_m} \in C^\infty (U)$, $m \leqslant r\leqslant{\rm dim}\,{\cal M}$.
If two multivector fields ${\bf X},{\bf X}'$ belong to the same
equivalence class $\{ {\bf X}\}$ then, for every $U\subset {\cal M}$,
there exists a non-vanishing function $f\in\Cinfty(U)$ such that ${\bf X}'=f{\bf X}$ on $U$.

If $(x^\mu,y^i)$ are fiber coordinates in the bundle 
$\tau\colon{\cal M}\to M$, a $\tau$-transverse and locally decomposable multivector field ${\bf X} \in \mathfrak{X}^m({\cal M})$ is
$$
\mathbf{X} =\bigwedge_{\mu=1}^m\left(\derpar{}{x^\mu}+X_\mu^\alpha\derpar{}{y^i}\right) \ .
$$
A section $\psi(x^\mu) = (x^\mu,\,\psi^{\alpha}(x^\nu))$  of $\tau$
is an integral section of ${\bf X}$ if its component functions
satisfy the following system of partial differential equations
$$
\derpar{\psi^{\alpha}}{x^i}=X_i^\alpha\circ\psi \ .
$$

\begin{definition}
If $\Omega\in\df^k({\cal M})$ and $\mathbf{X}\in\mathfrak{X}^m({\cal M})$,
the {\rm contraction} between ${\bf X}$ and $\Omega$ is defined as
the natural contraction between tensor fields; in particular,
\beann
 \inn({\bf X})\Omega\mid_{U}&:=& \sum_{1\leq \mu_1<\ldots <\mu_m\leq
 r}f^{\mu_1\ldots\mu_m} \inn(X_{\mu_1}\wedge\ldots\wedge X_{\mu_m})\Omega
\\ &=&
 \sum_{1\leq\mu_1<\ldots <\mu_m\leq r}f^{\mu_1\ldots \mu_m} \inn
 (X_{\mu_1})\ldots\inn (X_{\mu_m})\Omega \ ,
\eeann
if $k\geq m$, and equal to zero if $k<m$.
The {\rm Lie derivative} of $\Omega$ with respect to ${\bf X}$ is defined as
$$
\Lie({\bf X})\Omega:=
\d\inn ({\bf X})\Omega-(-1)^m\inn ({\bf X})\d\Omega \ . 
$$\end{definition}

\begin{definition}
 A multivector field $\mathbf{X}\in\mathfrak{X}^m({\cal M})$ is 
 \textsl{$\tau$-transverse} if, for every $\beta\in\Omega^m(M)$ with
$\beta (\tau(p))\not= 0$, at every point
$p\in{\cal M}$, we have that  $(\inn(\mathbf{X})(\tau^*\beta))_p\not= 0$. 
If $\mathbf{X}\in\mathfrak{X}^m({\cal M})$ is
integrable, then it is       
$\tau$-transverse if, and only if, its integral manifolds are local sections of $\tau$.
In this case, if $\psi\colon U\subset M\to{\cal M}$ is a local
section and $\psi (U)$ is the integral manifold 
of $\mathbf{X}$ at $p$, then  $T_p({\rm Im}\,\psi) = \mathcal{D}_p(\mathbf{X})$
and $\psi$ is an {\rm integral section} of ${\bf X}$.
\end{definition}

\begin{definition}
Consider the case that ${\cal M}=J^1\pi$, where $J^1\pi$ is the first-order jet bundle of a bundle $E\to M$.
Then, a multivector field  $\mathbf{X}\in\mathfrak{X}^m(J^1\pi)$
is  {\rm holonomic} if it is integrable and 
its integral sections are holonomic sections of the projection 
$\overline{\pi}^1\colon J^1\pi\to M$
(and hence it is locally decomposable and $\overline{\pi}^1$-transverse).
\end{definition}

\subsection{Geometrical setting for the Einstein-Palatini action (without energy-matter sources)}

We introduce here the Metric-Affine (or Einstein-Palatini) 
action for the Einstein equations of gravity without sources (no matter-energy is present).

The configuration bundle for this system is the bundle $\pi\colon{\rm E}\rightarrow M$, 
where $M$ is a connected orientable 4-dimensional manifold representing space-time, 
whose volume form is denoted $\eta\in\df^4(M)$, and
${\rm E}=\Sigma\times_MC(LM)$, where $\Sigma$ is the manifold of Lorentzian metrics on $M$ 
and $C(LM)$ is the bundle of connections on $M$;
that is, linear connections in $\Tan M$.

Consider a natural system of coordinates $(x^\mu,v^\alpha)$ in the tangent space $\tau\colon\Tan M\rightarrow M$, such that 
$\eta=\d x^0\wedge\ldots\wedge\d x^3\equiv\d^4x$. We use adapted fiber coordinates in ${\rm E}$, 
denoted $(x^\mu,g_{\alpha\beta},\Gamma^\nu_{\lambda\gamma})$,
(with $0\leq\alpha\leq\beta\leq 3$, and $\mu,\nu,\gamma,\lambda=0,1,2,3$). 
The functions $g_{\alpha\beta}$ are the components of the metric associated to the charts in the base $(x^\mu)$,
and $\Gamma^\nu_{\lambda\gamma}$ are the Christoffel symbols of the connection 
(and then the component functions $\Gamma^\nu_\gamma$ of the linear connection are 
$\Gamma^\nu_\gamma=\tau^*(-\Gamma^\nu_{\lambda\gamma}v^\lambda)$ \cite{EMR-2018}).
Since $g$ is symmetric, $g_{\alpha\beta}=g_{\beta\alpha}$ 
and actually there are 10 independent components. 
We do not assume torsionless connections and hence
$\Gamma^\nu_{\lambda\gamma}\neq\Gamma^\nu_{\gamma\lambda}$,
in general. Thus $\dim {\rm E}=78$. 
When we sum over symmetric indices and not over all the components, 
we order the indices as $0\leq\alpha\leq\beta\leq3$.

In order to state the formalism we consider the 
first-order jet bundle $J^1\pi$, which is the manifold of the $1$-jets of local sections 
$\phi \in \Gamma(\pi)$; that is, equivalence classes of local sections of $\pi$.
Points in $J^1\pi$ are denoted by $j^1_x\phi$, where $x \in M$
and $\phi \in \Gamma(\pi)$ is a representative of the equivalence class
(here $\Gamma(\pi)$ denotes the set of sections of $\pi$).
We have the natural projections
$$
\begin{array}{rcl}
\pi^1 \colon J^1\pi & \longrightarrow & {\rm E} \\
j^1_x\phi & \longmapsto & \phi(x)
\end{array}
\quad ; \quad
\begin{array}{rcl}
\overline{\pi}^1 \colon J^1\pi & \longrightarrow & M \\
j^1_x\phi & \longmapsto & x
\end{array} \ .
$$
Induced coordinates in $J^1\pi$ are denoted
$(x^\mu,\,g_{\alpha\beta},\,\Gamma^\nu_{\lambda\gamma},\,g_{\alpha\beta,\mu},
\,\Gamma^\nu_{\lambda\gamma,\mu})$, and $\dim J^1\pi=374$.
Finally, if $\phi \in \Gamma(\pi)$, the \textsl{$1$st prolongation} or {\sl canonical lifting} of $\phi$ to $J^1\pi$
is denoted by $j^1\phi \in \Gamma(\overline{\pi}^1)$.

A special kind of vector fields are the {\sl coordinate total derivatives} 
\cite{pere,book:Saunders89},
which are locally given as
$$
\displaystyle
D_\tau=\derpar{}{x^\tau}
+\sum_{\alpha\leq\beta} \left(g_{\alpha\beta,\tau}\derpar{}{g_{\alpha\beta}}+g_{\alpha\beta,\mu\tau}\derpar{}{g_{\alpha\beta,\mu}}\right)+ 
\Gamma_{\alpha\beta,\tau}^\nu\derpar{}{\Gamma_{\alpha\beta}^\nu}
+
\Gamma_{\alpha\beta,\mu\tau}^\nu\derpar{}{\Gamma_{\alpha\beta,\mu}^\nu} \ .
$$
Observe that, if $f\in\Cinfty(J^1\pi)$, then $D_\tau f\in\Cinfty(J^2\pi)$.

Next, let 
${\cal M}\pi\equiv\Lambda_2^4(\Tan^*{\rm E})$ 
be the bundle of $4$-forms in ${\rm E}$ vanishing 
by the action of two $\pi$-vertical vector fields,
which is usually called the {\sl extended multimomentum bundle} of ${\rm E}$,
and is endowed with the canonical projections
$$
\kappa\colon{\cal M}\pi\to{\rm E} \quad ;\quad
\overline{\kappa}=\pi\circ\kappa\colon{\cal M}\pi\to M \, .
$$
Induced local coordinates in ${\cal M}\pi$ are
$(x^\mu,\,g_{\alpha\beta},\Gamma^\nu_{\lambda\gamma},\,p,\,p^{\alpha\beta,\mu},
\,p^{\lambda\gamma,\mu}_\nu)$, with $0\leq\alpha\leq\beta\leq3$.
This bundle is endowed with the
\textsl{tautological (or Liouville) $4$-form} 
$\Theta\in\df^4({\cal M}\pi)$
and the \textsl{canonical (or Liouville) $5$-form} 
$\Omega= -d\Theta_{1} \in \Omega^5({\cal M}\pi)$
which is a multisymplectic form; that is, it is closed and $1$-nondegenerate.
Their local expressions are
\beann
\Theta&=& 
p\,\d^4x+\sum_{\alpha\leq\beta}\left(\p^{\alpha\beta,\mu}\,\d g_{\alpha\beta}\wedge \d^{3}x_\mu+
p^{\lambda\gamma,\mu}_\nu\,\d\Gamma^\nu_{\lambda\gamma}\wedge \d^{3}x_{\mu}\right)
 \ , \\
\Omega&=&
-\d p\wedge \d^4x-\sum_{\alpha\leq\beta}\left(\d p^{\alpha\beta,\mu}\wedge \d g_{\alpha\beta}\wedge \d^{3}x_\mu+
\d p^{\lambda\gamma,\mu}_\nu
\wedge\d\Gamma^\nu_{\lambda\gamma}\wedge \d^{3}x_{\mu}\right) \ ;
\eeann
where $\displaystyle \d^3x_\mu=\inn\left(\derpar{}{x^\mu}\right)\d^4x$.

\section{The Metric-Affine model: Lagrangian formalism}
\label{Sec3}

\subsection{Poincar\'e-Cartan forms and field equations}

(See, for instance,\cite{AA-80,EMR-96,art:Echeverria_Munoz_Roman98,Gc-73,
art:GR-2016,GS-73,book:Saunders89} for the general setting of the 
Lagrangian formalism of field theories in jet bundles).

The {\sl Einstein-Palatini} (or {\sl Metric-Affine\/}) {\sl  Lagrangian density} is a 
$\overline{\pi}^1$-semibasic 4-form $\mathcal{L}_{\rm EP}\in\Omega^4(J^1\pi)$;
then $\mathcal{L}_{\rm EP}=L_{\rm EP}\,(\overline{\pi}^1)^*\eta$, 
where $L_{\rm EP}\in\Cinfty(J^1\pi)$ is the {\sl Einstein-Palatini Lagrangian function}
which, in the above coordinates, is given by
$$
L_{\rm EP}=\sqrt{|{\rm det}(g)|}\,g^{\alpha\beta}R_{\alpha\beta}\equiv
\varrho g^{\alpha\beta}R_{\alpha\beta}=\varrho\,R \ ,
$$
where $\varrho=\sqrt{|det(g_{\alpha\beta})|}$, 
$R=g^{\alpha\beta}R_{\alpha\beta}$ is the {\sl scalar curvature},
$R_{\alpha\beta}=
\Gamma^{\gamma}_{\beta\alpha,\gamma}-\Gamma^{\gamma}_{\gamma\alpha,\beta}+
\Gamma^{\gamma}_{\beta\alpha}\Gamma^{\sigma}_{\sigma\gamma}-
\Gamma^{\gamma}_{\beta\sigma}\Gamma^{\sigma}_{\gamma\alpha}$
are the components of the {\sl Ricci tensor}, which depend only on the connection, and
$g^{\alpha\beta}$ denotes the inverse matrix of $g$, 
namely: $g^{\alpha\beta}g_{\beta\gamma}=\delta^\alpha_\gamma$.
It is useful to consider the following auxiliary functions:
\bea \label{auxfuns1}
L^{\beta\gamma,\mu}_{\alpha}&:=&\frac{\partial L_{\rm EP}}{\partial \Gamma^{\alpha}_{\beta\gamma,\mu}}=\varrho(\delta_\alpha^\mu g^{\beta\gamma}-\delta_\alpha^\beta g^{\mu\gamma})\,,
\\ \label{auxfuns2}
H&:=&L^{\beta\gamma,\mu}_{\alpha}\Gamma^{\alpha}_{\beta\gamma,\mu}-L_{\rm EP}=\varrho g^{\alpha\beta}\left(\Gamma^{\gamma}_{\beta\sigma}\Gamma^{\sigma}_{\gamma\alpha}-\Gamma^{\gamma}_{\beta\alpha}\Gamma^{\sigma}_{\sigma\gamma}\right)\,.
\eea
The bundle $J^1\pi$ is endowed with a canonical structure
which is called the {\sl vertical endomorphism},
${\cal V}\in\df^1(J^1\pi)\otimes\Gamma (J^1\pi,{\rm V}(\pi^1 ))
\otimes\Gamma (J^1\pi,\bar\pi^{1^*}\Tan M)$
(here ${\rm V}(\pi^1)$ denotes the vertical subbundle
with respect to the projection $\pi^1$,
and $\Gamma(J^1\pi,{\rm V}(\pi^1))$ the set of sections
in the corresponding bundle) \cite{AA-80,EMR-96,Gc-73,GS-73,book:Saunders89}.
Then the {\sl Poincar\'e--Cartan forms}
associated with ${\mathcal{L}_{\rm EP}}$ are def\/ined as
\[
\Theta_{\mathcal{L}_{\rm EP}}:=\inn({\cal V}){\mathcal{L}_{\rm EP}}+{\mathcal{L}_{\rm EP}}\in\df^4(J^1\pi)
\quad ,\quad
\Omega_{\mathcal{L}_{\rm EP}}:= -\d\Theta_{\mathcal{L}_{\rm EP}}\in\df^5(J^1\pi)  \ ,
\]
and the local expression for the last one is
\beq
\label{pc:la}
\Omega_{\mathcal{L}_{\rm EP}}=\d H\wedge\d^4x -
\d L^{\beta\gamma,\mu}_{\alpha}\wedge\d \Gamma^{\alpha}_{\beta\gamma}\wedge \d^3x_{\mu}  
\ .
\eeq
Observe that it is a  $\pi^1$-projectable form.

The variational problem \cite{GPR-2017,art:Roman09} associated to the system  
$(J^1\pi,\Omega_{\mathcal{L}_{\rm EP}})$ consists in finding holonomic sections 
$\psi_\Lag=j^1\phi\in \Gamma(\overline{\pi}^1)$ (with $\phi\in\Gamma(\pi)$) 
which are solutions to the equation
$$
\psi_\Lag^*\inn(X)\Omega_{\mathcal{L}_{\rm EP}} = 0 
\quad , \quad \mbox{for every } X \in \mathfrak{X}(J^1\pi) \ ,
$$
or, what is equivalent, which are integral sections of
a multivector field $\mathbf{X}_\Lag$ contained 
in a class of holonomic multivector fields 
$\{\mathbf{X}_\Lag\}\in\mathfrak{X}^4(J^1\pi)$ such that 
\beq
 \label{fundeqs}
\inn(\mathbf{X}_\Lag)\Omega_{\mathcal{L}_{\rm EP}}=0 \quad, \quad
\forall {\bf X}_\Lag\in\{{\bf X}_\Lag\}\subset\mathfrak{X}^4(J^1\pi) .
\eeq
The $\overline{\pi}^1$-transverse multivector fields $\mathbf{X}\in\mathfrak{X}^4(J^1\pi)$ 
can be characterized by demanding that
$\inn({\bf X})(\overline{\pi}^1)^*\eta\not= 0$. Then,
for a generic locally decomposable and $\overline{\pi}^1$-transverse
multivector field in $J^1\pi$
we have the following local expression $\displaystyle\mathbf{X}=f\bigwedge_{\nu=0}^3X_\nu$, with
\beq
X_\nu=\frac{\partial}{\partial x^\nu}+
\sum_{\rho\leq\sigma}\left(f_{\rho\sigma,\nu}\frac{\partial}{\partial g_{\rho\sigma}}+
f_{\rho\sigma\mu,\nu}\frac{\partial}{\partial g_{\rho\sigma,\mu}}\right)+f^\alpha_{\beta\gamma,\nu}
\frac{\partial}{\partial \Gamma^\alpha_{\beta\gamma}}+
f^\alpha_{\beta\gamma\mu,\nu}\frac{\partial}{\partial \Gamma^\alpha_{\beta\gamma,\mu}}\ ,
\label{locexmvf}
\eeq
where the coefficients are arbitrary functions of $C^\infty(J^1\pi)$. 
If the multivector field is holonomic and we set $f=1$, then necessarily
\beq
\mathbf{X}=
\bigwedge_{\nu=0}^3\left(\frac{\partial}{\partial x^\nu}+
\sum_{\rho\leq\sigma}\left(g_{\rho\sigma,\nu}\frac{\partial}{\partial g_{\rho\sigma}}+
f_{\rho\sigma\mu,\nu}\frac{\partial}{\partial g_{\rho\sigma,\mu}}\right)+
\Gamma^\alpha_{\beta\gamma,\nu}\frac{\partial}{\partial \Gamma^\alpha_{\beta\gamma}}+
f^\alpha_{\beta\gamma\mu,\nu}\frac{\partial}{\partial \Gamma^\alpha_{\beta\gamma,\mu}}
\right) \ .
\label{sopde}
\eeq

Taking \eqref{locexmvf} and \eqref{pc:la}, the equation \eqref{fundeqs}
becomes locally
\bea
\label{eq:fun3}
0&=&\inn(X_\mu)\d H+ f^\alpha_{\beta\gamma,\mu}\inn(X_\nu)\d L_\alpha^{\beta\gamma,\nu}-
f^\alpha_{\beta\gamma,\nu}\inn(X_\mu)\d L_\alpha^{\beta\gamma,\nu} \ ,
\\ \label{eq:fun4}
0&=&\frac{\partial H}{\partial g_{\sigma\rho}}-
f^\alpha_{\beta\gamma,\mu} \frac{\partial L_{\alpha}^{\beta\gamma,\mu}}{\partial g_{\sigma\rho}} \ ,
\\ \label{eq:fun5}
0&=&\frac{\partial H}{\partial \Gamma^\alpha_{\beta\gamma} } +
\sum_{\rho\leq\sigma}\left(f_{\rho\sigma,\mu} \frac{\partial L_{\alpha}^{\beta\gamma,\mu}}{\partial g_{\rho\sigma}}\right)+
f^\tau_{\rho\sigma,\mu} \frac{\partial L^{\beta\gamma,\mu}_\alpha}{\partial \Gamma^\tau_{\rho\sigma}}-
f^\tau_{\rho\sigma,\mu} \frac{\partial L_\tau^{\rho\sigma,\mu}}{\partial \Gamma_{\beta\gamma}^\alpha} 
\nonumber \\ &=&
\frac{\partial H}{\partial \Gamma^\alpha_{\beta\gamma} } +
\sum_{\rho\leq\sigma}f_{\rho\sigma,\mu} \frac{\partial L_{\alpha}^{\beta\gamma,\mu}}{\partial g_{\rho\sigma}} \ ;
\eea
since $\displaystyle\frac{\partial L^{\beta\gamma,\mu}_\alpha}{\partial \Gamma^\tau_{\rho\sigma}}=0$.
Equations \eqref{eq:fun3} arise from the variations of the coordinates $x^\mu$ and
they hold  as a consequence of \eqref{eq:fun4} and \eqref{eq:fun5}.
The equations \eqref{eq:fun4} arise from the variations on the components of the metric, 
and contains the functions $f^\alpha_{\beta\gamma,\mu}$ related to the connection, 
thus we call them {\sl connection equations}. 
Finally, the equations \eqref{eq:fun5} 
arise from the variations on the components of the connection, 
and contain the functions $f_{\sigma\rho,\mu}$, 
thus they are called {\sl metric equations}.

The fact that a multivector field in $J^1\pi$ 
has the local expression \eqref{sopde}
(then being locally decomposable and $\overline{\pi}^1$-transverse)
is just a necessary condition to be holonomic, since it may not be integrable;
but, if it admits integral sections, 
then its integral sections are holonomic.
Locally decomposable and $\overline{\pi}^1$-transverse multivector fields
which have \eqref{sopde} as coordinate expression are said to be {\sl semiholonomic} in $J^1\pi$
(see \cite{art:Echeverria_Munoz_Roman98}
for an intrinsic definition of these kinds of multivector fields).

\subsection{Compatibility and consistency constraints}
\label{ccc}

In general,
$\overline{\pi}^1$-transverse and integrable multivector fields 
${\bf X}\in\vf^4(J^1\pi)$ which are  solutions to \eqref{fundeqs} 
could not exist. In the best of cases they exist only in some submanifold 
of $J^1\pi$ \cite{LMMMR-2005}. 
The aim in this section is to find the constraints that define this submanifold, 
using a local version of the geometric constraint algorithms \cite{LMM-96,LMMMR-2005}.

First, we introduce the following notation: as it is usual,
$$
\ker^4\Omega_{\Lag_{\rm EP}}:=\{{\bf X}\in\vf^4(J^1\pi)\,\vert\,  \inn({\bf X})\Omega_{\Lag_{\rm EP}}=0\}\ .
$$
We denote by $\ker^4_{\overline{\pi}^1}\Omega_{\Lag_{\rm EP}}$
the set of locally decomposable and 
$\overline{\pi}^1$-transverse multivector fields satisfying equations \eqref{fundeqs} 
but not being (semi)holonomic necessarily. Then,
$\ker^4_{SH}\Omega_{\Lag_{\rm EP}}$ and $\ker^4_H\Omega_{\Lag_{\rm EP}}$
denote the sets of semi-holonomic and the 
holonomic multivector fields which are solutions to the equations \eqref{fundeqs}, respectively.
Obviously we have
\beq
\ker^4_H\Omega_{\Lag_{\rm EP}}\subset
\ker^4_{SH}\Omega_{\Lag_{\rm EP}}\subset\ker^4
_{\overline{\pi}^1}\Omega_{\Lag_{\rm EP}}
 \subset\ker^4\Omega_{\Lag_{\rm EP}}\ .
\label{chain}
\eeq

We make the study in several steps, following the next procedure: 
first we consider the problem of finding locally decomposable and $\overline{\pi}^1$-transverse multivector fields
which are solution to  \eqref{fundeqs} (that is, the elements of $\ker^4_{\overline{\pi}^1}\Omega_{\Lag_{\rm EP}}$),
then we look for the semi-holonomic multivector fields belonging to $\ker^4_{SH}\Omega_{\Lag_{\rm EP}}$ 
and finally, in the next Section, 
we analyze their integrability (finding the elements of $\ker^4_H\Omega_{\Lag_{\rm EP}}$).

\subsubsection{Non-semiholonomic multivector fields (elements of 
$\ker^4_{\overline{\pi}^1}\Omega_{\Lag_{\rm EP}}$):
compatibility constraints}
\label{se:lanh}

The set $\ker^4_{\overline{\pi}^1}\Omega_{\Lag_{\rm EP}}$ consists of multivector fields of the form 
\eqref{locexmvf} whose coefficients satisfy the connection and metric 
equations (\ref{eq:fun4}) and (\ref{eq:fun5}) respectivelly. 
But the equations (\ref{eq:fun5}) are not compatible. In fact:

\begin{prop}
The necessary condition
for the existence of solutions to the metric equations (\ref{eq:fun5}) is that the following equalities hold:
\beq
A_{\alpha\beta\gamma}\equiv g_{\beta\nu}T^\nu_{\alpha\gamma}-g_{\alpha\nu}T^\nu_{\beta\gamma}+\tfrac{1}{3}g_{\beta\gamma}T^\nu_{\nu \alpha}-
\tfrac{1}{3}g_{\alpha\gamma}T^\nu_{\nu\beta}=0\ ,
\label{torscond}
\eeq
where $T^\alpha_{\beta\gamma}$ are the components of the torsion tensor
which are defined as usual,
$T^\alpha_{\beta\gamma}=\Gamma^\alpha_{\beta\gamma}-\Gamma^\alpha_{\gamma\beta}$.
\end{prop}
\begin{proof}
We introduce the following functions
\beq\label{eq:inversa}
\beth^\alpha_{\beta\gamma,\lambda\zeta\nu}=
\frac{1}{\varrho}\left(-\frac{1}{2}g_{\beta \gamma}g_{\lambda\zeta}\delta_\nu^\alpha+
\frac{1}{6}g_{\lambda\zeta}g_{\nu\gamma}\delta_\beta^\alpha-
\frac13g_{\lambda\nu}g_{\zeta\gamma}\delta_\beta^\alpha+
g_{\zeta\gamma}g_{\lambda\beta}\delta_\nu^\alpha\right) \ ,
\eeq
which satisfy that
$$
\frac{\partial L_{\alpha}^{\beta\gamma,\mu}}{\partial g_{\rho\sigma}}\beth^{\alpha}_{\beta\gamma,\lambda\zeta\nu}=
\frac{n(\rho\sigma)}{2}(\delta_\nu^\mu\delta_\zeta^\sigma\delta_\lambda^\rho+\delta_\nu^\mu\delta_\lambda^\sigma\delta_\zeta^\rho)\ ;
$$
where $n(\rho\sigma)$ is a combinatorial factor such that 
$n(\rho\sigma)=1$ for $\rho=\sigma$, and $n(\rho\sigma)=2$ for $\rho\neq\sigma$.
Then, using them in the metric equations (\ref{eq:fun5}), we obtain
$$
0=\beth^\alpha_{\beta\gamma,\lambda\zeta\nu}\left(\frac{\partial H}{\partial \Gamma^\alpha_{\beta\gamma} } +\sum_{\rho\leq\sigma}
f_{\rho\sigma,\mu} \frac{\partial L_{\alpha}^{\beta\gamma,\mu}}{\partial g_{\rho\sigma}}\right)=
\beth^\alpha_{\beta\gamma,\lambda\zeta\nu}\frac{\partial H}{\partial \Gamma^\alpha_{\beta\gamma} }+\frac12(f_{\lambda\zeta,\nu}+f_{\zeta\lambda,\nu})\ .
$$
These are equations for the functions $f_{\lambda\zeta,\nu}$ which,
as a consequence of the symmetry of the metric, $g_{\alpha\beta}=g_{\beta\alpha}$,  are also symmetric: $f_{\lambda\zeta,\nu}=f_{\zeta\lambda,\nu}$. 
Nevertheless, the equations are incompatible 
because they are not symmetric under the change 
$\lambda\leftrightarrow \zeta$.
In fact; we obtain that
$$
\beth^\alpha_{\beta\gamma,\lambda\zeta\nu}\frac{\partial H}{\partial \Gamma^\alpha_{\beta\gamma} }-
\beth^\alpha_{\beta\gamma,\zeta\lambda\nu}\frac{\partial H}{\partial \Gamma^\alpha_{\beta\gamma} }=
g_{\lambda\mu}T^\mu_{\zeta\nu}-g_{\zeta\mu}T^\mu_{\lambda\nu}+\tfrac{1}{3}g_{\lambda\nu}T^\mu_{\mu \zeta}-
\tfrac{1}{3}g_{\zeta\nu}T^\mu_{\mu \lambda}=0 \ ,
$$
and the result follows from here.
\end{proof}

Conditions \eqref{torscond} are called {\sl torsion constraints} and they define the submanifold $\mathcal{S}_T\hookrightarrow J^1E$.
These torsion constraints are essential in the following discussion, 
since they impose strong restrictions on the torsion. In fact:

\begin{prop}
\label{pr:toco} 
The torsion constraints \eqref{torscond} are equivalent to
\beq
\label{torscons}
T^\alpha_{\beta\gamma}=
\frac13\delta^\alpha_\beta T^\nu_{\nu\gamma}-\frac13\delta^\alpha_\gamma T^\nu_{\nu\beta}\ .
\eeq
\end{prop}
\begin{proof}
If \eqref{torscond} holds, then
\beann
0&=&\frac12g^{\alpha \mu}\left(A_{\beta\mu\gamma}+A_{\beta\gamma\mu}+A_{\mu\gamma\beta}\right)
\\ &=&\frac12g^{\alpha \mu}\left( g_{\mu\nu}T^\nu_{\beta\gamma}-
g_{\beta\nu}T^\nu_{\mu\gamma}+\tfrac{1}{3}g_{\gamma \mu}T^\nu_{\nu \beta}-
\tfrac{1}{3}g_{\gamma \beta}T^\nu_{\nu \mu}+g_{\gamma\nu}T^\nu_{\beta\mu}-
g_{\beta\nu}T^\nu_{\gamma\mu}\right.
\\ & &
+ \left.\tfrac{1}{3}g_{\mu \gamma}T^\nu_{\nu \beta}-
\tfrac{1}{3}g_{\mu \beta}T^\nu_{\nu \gamma}+
 g_{\gamma\nu}T^\nu_{\mu\beta}-g_{\mu\nu}T^\nu_{\gamma\beta}+
\tfrac{1}{3}g_{\beta \gamma}T^\nu_{\nu \mu}-
\tfrac{1}{3}g_{\beta \mu}T^\nu_{\nu \gamma}\right)
\\
&=&T^\alpha_{\beta\gamma}-\frac13\delta^\alpha_\beta T^\nu_{\nu\gamma}+
\frac13\delta^\alpha_\gamma T^\nu_{\nu\beta}\ .
\eeann
Conversely, if
$T^\alpha_{\beta\gamma}=
\frac13\delta^\alpha_\beta T^\nu_{\nu\gamma}-\frac13\delta^\alpha_\gamma T^\nu_{\nu\beta}$,
 then
\beann
A_{\alpha\beta\gamma}&=& 
g_{\beta\nu}T^\nu_{\alpha\gamma}-g_{\alpha\nu}T^\nu_{\beta\gamma}+\tfrac{1}{3}g_{\beta\gamma}T^\nu_{\nu \alpha}-
\tfrac{1}{3}g_{\alpha\gamma}T^\nu_{\nu\beta}
\\
&=& g_{\beta\nu}\left(\tfrac13\delta^\nu_\alpha T^\mu_{\mu\gamma}-\tfrac13\delta^\nu_\gamma T^\mu_{\mu \alpha}\right) -
g_{\alpha\nu}\left(\tfrac13\delta^\nu_\beta T^\mu_{\mu \gamma}-\tfrac13\delta^\nu_\gamma T^\mu_{\mu \beta}\right)+
\tfrac{1}{3}g_{\beta\gamma}T^\nu_{\nu\alpha}-\tfrac{1}{3}g_{\alpha\gamma}T^\nu_{\nu \beta}
\\
&=&\frac13\left(g_{\beta\alpha}T^\mu_{\mu \gamma}-g_{\beta\gamma} T^\mu_{\mu\alpha} -g_{\alpha\beta} T^\mu_{\mu\gamma}+g_{\alpha\gamma} T^\mu _{\mu \beta}+g_{\beta\gamma}T^\nu_{\nu\alpha}-g_{\alpha\gamma}T^\nu_{\nu\beta}\right)=0 \ .
\eeann
\end{proof}

As a consequence of this result, on $\mathcal{S}_T$ 
the torsion is determined by its ``trace'', $tr(T)=T^\nu_{\alpha\nu}$.

\begin{prop}
\label{pr:so1}
On the submanifold $\mathcal{S}_T$, the general solutions to the equations 
\eqref{eq:fun4} and \eqref{eq:fun5} are, respectively,
\bea
\label{eq:sol2}
f^{\alpha}_{\beta\gamma,\mu}&=&\Gamma^\lambda_{\mu\gamma}\Gamma^\alpha_{\beta\lambda}+C^\alpha_{\beta\gamma,\mu}+K^\alpha_{\beta\gamma,\mu} \ ,
\\
\label{eq:sol1}
f_{\sigma\rho,\mu}&=&g_{\sigma\lambda}\Gamma^\lambda_{\mu\rho}+g_{\rho\lambda}\Gamma^\lambda_{\mu\sigma}+\frac{2}{3}g_{\sigma\rho}T^\lambda_{\lambda\mu} \ ;
\eea
for some functions 
$C^\alpha_{\beta\gamma,\mu}, K^\alpha_{\beta\gamma,\mu}\in C^\infty(J^1\pi)$ satisfying that
$$
C^\alpha_{\beta\gamma,\mu}=C_{\beta\mu}\delta^\alpha_\gamma \quad,  \quad
K^\nu_{\nu\gamma\mu}=0 \quad , \quad
K^{\nu}_{\beta\gamma \nu}+K^\nu_{\gamma\beta \nu}=0 
\quad ; \quad \mbox{\rm (on $\mathcal{S}_T$)} \  .
$$
\end{prop}
\begin{proof} 
The metric and connection equations are independent and linear. 
Thus we look for particular and homogeneous-general solutions for each one. 

It is straightforward to check that \eqref{eq:sol1} is a particular solution to the metric equations on $\mathcal{S}_T$.
Given two solutions, $f^1$ and $f^2$, their difference 
$h_{\sigma\rho,\mu}=f^1_{\sigma\rho,\mu}-f^2_{\sigma\rho,\mu}$ 
is a solution to the homogeneous equation 
$$
\sum_{\rho\leq\sigma}h_{\rho\sigma,\mu} \frac{\partial L_{\alpha}^{\beta\gamma,\mu}}{\partial g_{\rho\sigma}}=0 
\quad ; \quad \mbox{\rm (on ${\cal S}_T$)} \ .
$$
Consider the functions $\beth^\alpha_{\beta\gamma,\lambda\zeta\nu}$ which satisfy \eqref{eq:inversa},
$$
0=\sum_{\rho\leq\sigma}h_{\rho\sigma,\mu} \frac{\partial L_{\alpha}^{\beta\gamma,\mu}}{\partial g_{\rho\sigma}}\beth^\alpha_{\beta\gamma,\lambda\zeta\nu} =h_{\rho\sigma,\mu}\frac12(\delta_\nu^\mu\delta_\lambda^\sigma\delta_\zeta^\rho+\delta_\nu^\mu\delta_\zeta^\sigma\delta_\lambda^\rho) =h_{\lambda\zeta\nu}\ .
$$
Therefore, $\displaystyle h_{\sigma\rho,\mu}\vert_{\mathcal{S}_T}=0\Rightarrow f^1(p)=f^2(p)$ on $ \mathcal{S}_T$, and the solution is unique. 
In a similar way,
$$
f_{\beta\gamma,\mu}^{\alpha}=
\Gamma^\lambda_{\mu\gamma}\Gamma^\alpha_{\beta\lambda}
\quad ; \quad \mbox{\rm (on ${\cal S}_T$)} 
$$
is a particular solution to the connection equations. 
The difference between two solutions is a solution to the homogeneous equation:
\beq\label{eq:sol:hom}
h^\alpha_{\beta\gamma,\mu} \frac{\partial L_{\alpha}^{\beta\gamma,\mu}}{\partial g_{\rho\sigma}}=0
\quad ; \quad \mbox{\rm (on ${\cal S}_T$)} \ .
\eeq
This equation is equivalent to:
$$
h^\lambda_{\lambda r,s}+h^\lambda_{\lambda s,r}-h^\lambda_{rs,\lambda}-h^\lambda_{sr,\lambda}=0
\quad ; \quad \mbox{\rm (on ${\cal S}_T$)} \ .
$$
Indeed,
\beann
\frac1{\varrho n(\rho\sigma)}(2g_{r\rho}g_{s\sigma}-g_{\rho\sigma}g_{rs})h^\alpha_{\beta\gamma,\mu}  \frac{\partial L_{\alpha}^{\beta\gamma,\mu}}{\partial g_{\rho\sigma}}=h^\lambda_{\lambda r,s}+h^\lambda_{\lambda s,r}-h^\lambda_{rs,\lambda}-h^\lambda_{sr,\lambda}
\quad ; \quad \mbox{\rm (on ${\cal S}_T$)}  \ .
\\
\frac{\varrho n(\rho\sigma)}{4}(2g^{r\rho}g^{s\sigma}-g^{\rho\sigma}g^{rs})\left(h^\lambda_{\lambda r,s}+h^\lambda_{\lambda s,r}-h^\lambda_{rs,\lambda}-h^\lambda_{sr,\lambda}\right)=h^\alpha_{\beta\gamma,\mu}  \frac{\partial L_{\alpha}^{\beta\gamma,\mu}}{\partial g_{\rho\sigma}}
\quad ; \quad \mbox{\rm (on ${\cal S}_T$)}  \ .
\eeann
Some solutions of this equation are the functions of the form 
$$
h^{\alpha}_{\beta\gamma,\mu}=C_{\beta\mu}\delta^\alpha_\gamma\quad ; \quad \mbox{\rm (on ${\cal S}_T$)}  \ ,
$$
which are called {\sl trace solutions}. For any solution $h$, consider 
$K^{\alpha}_{\beta\gamma,\mu}=h^{\alpha}_{\beta\gamma,\mu}-C_{\beta\mu}\delta^\alpha_\gamma$ with $C_{\beta\mu}=h^\lambda_{\lambda\beta\mu}$. 
It follows that $K^\lambda_{\lambda\gamma\mu}=0$. 
Since the equation is linear, these functions must also be solutions. Therefore:
$$
0=K^\lambda_{\lambda r,s}+K^\lambda_{\lambda s,r}-K^\lambda_{rs,\lambda}-K^\lambda_{sr,\lambda}=-K^\lambda_{rs,\lambda}-K^\lambda_{sr,\lambda}
\quad ; \quad \mbox{\rm (on ${\cal S}_T$)}  \ .
$$
These solutions are called {\sl torsion solutions}. 
From their definition it is clear that any homogeneous solution is a sum of a trace and a torsion solution. Furthermore, if $K^\alpha_{\beta\gamma,\mu}=C^\alpha_{\beta\gamma,\mu}=C_{\beta\mu}\delta_\gamma^\alpha$, then $0=K^\lambda_{\lambda\gamma,\mu}=C_{\gamma \mu}$;
on ${\cal S}_T$. Thus, the only homogeneous solution which is both trace and torsion is $h^\alpha_{\beta\gamma,\mu}=~0$.
\end{proof}

This proposition shows also that:

\begin{corollary}
The torsion constraints \eqref{torscond} (or their equivalent expressions \eqref{torscons})
are sufficient conditions for the existence of solutions to \eqref{eq:fun5}.
\end{corollary}

These constraints could be also obtained in an intrinsic way using the
procedure described in \cite{LMMMR-2005}. 

Now we must check the tangency (or consistency) conditions.
First, observe that, taking into account \eqref{locexmvf}, \eqref{eq:sol2}, and \eqref{eq:sol1},
the general solution to the equation \eqref{fundeqs}
(before imposing the holonomy condition) are multivector fields of the form
\bea
\nonumber
\mathbf{X}=\bigwedge_{\nu=0}^3X_\nu&=&
\bigwedge_{\nu=0}^3\left[\frac{\partial}{\partial x^\nu}+
\sum_{\sigma\leq\rho}\left((
g_{\sigma\lambda}\Gamma^\lambda_{\nu\rho}+g_{\rho\lambda}\Gamma^\lambda_{\nu\sigma}+
\frac{2}{3}g_{\sigma\rho}T^\lambda_{\lambda\nu})
\frac{\partial}{\partial g_{\sigma\rho}}+
f_{\sigma\rho\mu,\nu}\frac{\partial}{\partial g_{\sigma\rho,\mu}}\right)
\right. \\ & & \left. \qquad
+(\Gamma^\lambda_{\nu\gamma}\Gamma^\alpha_{\beta\lambda}+
C^\alpha_{\beta\gamma,\nu}+K^\alpha_{\beta\gamma,\nu})
\frac{\partial}{\partial \Gamma^\alpha_{\beta\gamma}}+
f^\alpha_{\beta\gamma\mu,\nu}\frac{\partial}{\partial \Gamma^\alpha_{\beta\gamma,\mu}}
\right]
\ ; \ \mbox{\rm (on ${\cal S}_T$)} \ .
\label{mvfsolgen}
\eea
Bearing in mind the conditions on the functions 
$C^\alpha_{\beta\gamma,\mu}, K^\alpha_{\beta\gamma,\mu}$
stated in Proposition \ref{pr:so1},
the tangency condition on the torsion constraints \eqref{torscons}
$$
\Lie(X_\nu)\left(T^\alpha_{\beta\gamma}-\frac13\delta^\alpha_\beta T^\nu_{\nu\gamma}+\frac13\delta^\alpha_\gamma T^\nu_{\nu\beta}\right)=0 
\quad ; \quad \mbox{\rm (on ${\cal S}_T$)}\ ,
$$
hold on ${\cal S}_T$ as long as
$$
K^\alpha_{[\beta\gamma],\mu}=-\frac13\delta^\alpha_{[\beta} K^\nu_{\gamma]\nu,\mu}-\Gamma^\lambda_{\mu[\gamma}\Gamma^\alpha_{\beta]\lambda}+\frac13\delta^\alpha_{[\beta}\Gamma^\lambda_{\mu\gamma]}\Gamma^\nu_{\nu\lambda}-\frac13\delta^\alpha_{[\beta}\Gamma^\lambda_{\mu\nu}\Gamma^\nu_{\gamma]\lambda}
\quad ; \quad \mbox{(\rm on ${\cal S}_T$)}\ .
$$
Nevertheless, solutions to equation \eqref{fundeqs} must be holonomic multivector fields.
Thus, first we look for semiholonomic solutions, then we analyze their tangency and, finally, we study the existence of holonomic solutions.

\subsubsection{Semi-holonomic multivector fields (elements of $\ker^4_{SH}\Omega_{\Lag_{\rm EP}}$):
semiholonomic constraints}
\label{semiholcons}

If a multivector field is semiholonomic then its local expression is \eqref{sopde};
that is, 
$$
f_{\rho\sigma,\mu}=g_{\rho\sigma,\mu}\quad , \quad
f^\alpha_{\beta\gamma,\mu}=\Gamma^\alpha_{\beta\gamma,\mu} \ .
$$ 
In this case, there are more constraints 
which arise from the equations \eqref{eq:fun4} and \eqref{eq:fun5}
and are the Euler-Lagrange equations themselves:
\bea \label{eq:constcon}
\frac{\partial H}{\partial g_{\mu\nu}}-
\frac{\partial L_\alpha^{\beta\gamma,\sigma}}{\partial g_{\mu\nu}}\Gamma^\alpha_{\beta\gamma,\sigma}=0 \ ,
\\ \label{eq:constmetric}
\frac{\partial H}{\partial \Gamma^\alpha_{\beta\gamma}}+\sum_{\mu\leq\nu}
\frac{\partial L_\alpha^{\beta\gamma,\sigma}}{\partial g_{\mu\nu}}g_{\mu\nu,\sigma}=0 \ .
\eea
(Geometrically, they are a consequence of the fact that
$\Omega_{\mathcal{L}_{Ep}}$ is  $\pi^1$-projectable \cite{first,art:GR-2016,Krupka,KrupkaStepanova,rosado,rosado2}).
In this way, the connection and metric equations become {\sl semiholonomic constraints}, which are called
{\sl connection} and {\sl metric constrains}, respectively.

In particular, notice that the metric constraints \eqref{eq:constmetric} arise from 
the equations \eqref{eq:fun5}, which lead to the torsion constraints \eqref{torscons}.
Therefore, the metric constraints split into two kinds of conditions: 
the torsion constraints \eqref{torscons} themselves
and, according to equation \eqref{eq:sol1} (or, equivalently, to \eqref{mvfsolgen}),
\beq
\label{eq:pm}
g_{\rho\sigma,\mu}=
g_{\sigma\lambda}\Gamma^\lambda_{\mu\rho}+
g_{\rho\lambda}\Gamma^\lambda_{\mu\sigma}+
\frac{2}{3}g_{\rho\sigma}T^\lambda_{\lambda\mu} \ ,
\eeq
which are called {\sl pre-metricity constraints}.
They are closely related to the metricity conditions and the trace of the torsion,
as it is proved in the following: 

\begin{prop}
\label{pr:pm}
In the points of the submanifold ${\cal S}_m\hookrightarrow J^1\pi$ defined by the metric constraints \eqref{eq:constmetric}, we have that:
$$
\nabla^{\Gamma(p)} g(p)=0\ \Longleftrightarrow\ tr(T^{\Gamma(p)})=0
\quad ; \quad p\in {\cal S}_m \ .
$$
(Here, the notation $\nabla^{\Gamma(p)}$ 
means the covariant derivative with respect to the connection $\Gamma$ in the point $p$, 
and $T^{\Gamma(p)}$ denotes the torsion tensor associated to this connection).
\end{prop}
\begin{proof}
In the coordinates of $J^1\pi$
the metricity condition $\nabla^{\Gamma(p)} g(p)=0$ is
$$
\left(\nabla^{\Gamma(p)} g(p)\right)_{\rho\sigma,\mu}=g_{\rho\sigma,\mu}-g_{\sigma\lambda}\Gamma^\lambda_{\mu\rho}-g_{\rho\lambda}\Gamma^\lambda_{\mu\sigma} \ .
$$
Therefore, the statement follows immediately since the pre-metricity constraints \eqref{eq:pm} can be written as
$$
\left(\nabla^{\Gamma(p)} g(p)\right)_{\rho\sigma,\mu}=\frac23g_{\rho\sigma}T^\lambda_{\lambda\mu} \ .
$$
\end{proof}

\subsubsection{Tangency condition: consistency constraints}

Now we check the tangency (or consistency) condition for all the above sets of constraints.
A semiholonomic multivector field $\displaystyle X=\bigwedge_{\nu=0}^3X_\nu$ 
has the local expression \eqref{sopde}. 
The tangency condition on the connection constraints \eqref{eq:constcon} reads
\beq\label{eq:cons:conec}
\Lie(X_\nu)\left(\frac{\partial H}{\partial g_{\rho\sigma}}-\frac{\partial L_\alpha^{\beta\gamma,\mu}}{\partial g_{\rho\sigma}}\Gamma^\alpha_{\beta\gamma,\mu}\right)=D_\nu\frac{\partial H}{\partial g_{\rho\sigma}}-D_\nu\frac{\partial L_\alpha^{\beta\gamma,\mu}}{\partial g_{\rho\sigma}}\Gamma^\alpha_{\beta\gamma,\mu}-\frac{\partial L_\alpha^{\beta\gamma,\mu}}{\partial g_{\rho\sigma}}f^\alpha_{\beta\gamma\mu,\nu}=0 \quad \mbox{\rm (on ${\cal S}_T$)} \ ,
\eeq
and it does not lead to new constraints because they allow to determine 
the functions $f^\alpha_{\beta\gamma\mu,\nu}$ (on ${\cal S}_T$) . 
The tangency condition on the pre-metricity constraints \eqref{eq:pm} gives
\beq\label{eq:constmetric3}
f_{\sigma\rho,\mu\nu}=D_\lambda\left(g_{\sigma\lambda}\Gamma^\lambda_{\mu\rho}+
g_{\rho\lambda}\Gamma^\lambda_{\mu\sigma}+
\frac{2}{3}g_{\sigma\rho}T^\lambda_{\lambda\mu}\right)
\quad ; \quad \mbox{\rm (on ${\cal S}_T$)}  \ ,
\eeq
and it does not lead either to new constraints. 
But the tangency condition on the torsion constraints  \eqref{torscons} does
lead to new constraints
$$
\Lie(X_\nu)\left(T^\alpha_{\beta\gamma}-\frac13\delta^\alpha_\beta T^\mu_{\mu\gamma}+
\frac13\delta^\alpha_\gamma T^\mu_{\mu\beta}\right)=T^\alpha_{\beta\gamma,\nu}-
\frac13\delta^\alpha_\beta T^\mu_{\mu\gamma,\nu}+
\frac13\delta^\alpha_\gamma T^\mu_{\mu\beta,\nu}=0
\quad ; \quad  \mbox{\rm (on ${\cal S}_T$)}\ .
$$
The tangency condition on these new constraints leads to
$$
\Lie(X_\lambda)\left(T^\alpha_{\beta\gamma,\nu}-
\frac13\delta^\alpha_\beta T^\mu_{\mu\gamma,\nu}+
\frac13\delta^\alpha_\gamma T^\mu_{\mu\beta,\nu}\right)=
f^\alpha_{\beta\gamma\nu,\tau}-\frac13\delta^\alpha_\beta f^\mu_{\mu\gamma\nu,\tau}+
\frac13\delta^\alpha_\gamma f^\mu_{\mu\beta\nu,\tau}=0
\ ; \ \mbox{\rm (on ${\cal S}_{sh}$)}  \ ,
$$
which are not new constraints, but equations for the functions $f^\alpha_{\beta\gamma\mu,\nu}$.
Therefore, in the submanifold $\mathcal{S}_{sh}\hookrightarrow{\cal S}_T$ 
defined by these constraints there are semiholonomic multivector fields
solutions to the field equations, which are tangent to 
$\mathcal{S}_{sh}$.

Summarizing, we have proved that:

\begin{theo}
\label{finalconstraints}
There exists a submanifold ${\rm j}_{sh}\colon \mathcal{S}_{sh}\hookrightarrow J^1\pi$
where there are semi-holonomic multivector fields
which are solutions to the field equations
\eqref{fundeqs} 
and are tangent to ${\cal S}_{sh}$.
This submanifold is locally defined
in $J^1\pi$ by the constraints
\beann
c^{\mu\nu}&\equiv&\frac{\partial H}{\partial g_{\mu\nu}}-
\frac{\partial L_\alpha^{\beta\gamma,\sigma}}{\partial g_{\mu\nu}}\Gamma^\alpha_{\beta\gamma,\sigma}=0 \ ,
\\
m_{\rho\sigma,\mu}&\equiv&g_{\rho\sigma,\mu}-g_{\sigma\lambda}\Gamma^\lambda_{\mu\rho}-
g_{\rho\lambda}\Gamma^\lambda_{\mu\sigma}-\frac{2}{3}g_{\rho\sigma}T^\lambda_{\lambda\mu}=0 \ ,
\\
t^\alpha_{\beta\gamma}&\equiv&T^\alpha_{\beta\gamma}-
\frac13\delta^\alpha_\beta T^\mu_{\mu\gamma}+
\frac13\delta^\alpha_\gamma T^\mu_{\mu\beta}=0 \ ,
\\
r^\alpha_{\beta\gamma,\nu}&\equiv&T^\alpha_{\beta\gamma,\nu}-
\frac13\delta^\alpha_\beta T^\mu_{\mu\gamma,\nu}+
\frac13\delta^\alpha_\gamma T^\mu_{\mu\beta,\nu}=0 \ .
\eeann
\end{theo}

These constraints are not independent all of them.
For instance, the pre-metricity constraints $m_{\rho\sigma,\mu}$ 
are symmetric in the indices $\sigma,\rho$ and the constraints 
$t^\alpha_{\beta\gamma}$ and $r^\alpha_{\beta\gamma,\nu}$ are skewsymmetric in the indices $\beta,\gamma$.

\begin{prop}
The general expression of the semi-holonomic multivector fields
which are solutions to the field equations \eqref{fundeqs} on ${\cal S}_{sh}$ are
\beq
\mathbf{X}_\Lag=
\bigwedge_{\nu=0}^3\left(\frac{\partial}{\partial x^\nu}+
\sum_{\rho\leq\sigma}\left(g_{\rho\sigma,\nu}\frac{\partial}{\partial g_{\rho\sigma}}+
f_{\rho\sigma\mu,\nu}\frac{\partial}{\partial g_{\rho\sigma,\mu}}\right)+
\Gamma^\alpha_{\beta\gamma,\nu}\frac{\partial}{\partial \Gamma^\alpha_{\beta\gamma}}+
f^\alpha_{\beta\gamma\mu,\nu}\frac{\partial}{\partial \Gamma^\alpha_{\beta\gamma,\mu}}
\right) \ ,
\label{ELmvf}
\eeq
where, on the points of ${\cal S}_{sh}$,
\beann
f_{\rho\sigma\mu,\nu}&=&
D_\nu\left(g_{\sigma\lambda}\Gamma^\lambda_{\mu\rho}+g_{\rho\lambda}\Gamma^\lambda_{\mu\sigma}+\frac{2}{3}g_{\rho\sigma}T^\lambda_{\lambda\mu}\right) \ ,
\\
f^\alpha_{\beta\gamma\mu,\nu}&=&
\Gamma^\lambda_{\mu\gamma,\nu}\Gamma^\alpha_{\beta\lambda}+\Gamma^\lambda_{\mu\gamma}\Gamma^\alpha_{\beta\lambda,\nu}+C^\alpha_{\beta\gamma,\mu\nu}+K^{\alpha}_{\beta\gamma,\mu\nu}\ ,
\eeann
for any $C_{\beta\mu\nu}\in C^\infty(J^1\pi)$ and 
$K^{\alpha}_{\beta\gamma,\mu\nu}\in C^\infty(J^1\pi)$ satisfying that,
on ${\cal S}_{sh}$,
\beann
C^\alpha_{\beta\gamma,\mu\nu}&=&C_{\beta\mu\nu}\delta^\alpha_\gamma\quad ,\quad K^\lambda_{\lambda\gamma,\mu\nu}=0\quad,\quad K^\lambda_{\beta\gamma,\lambda\nu}+K^\lambda_{\gamma\beta,\lambda\nu}=0 \ , 
\\
K^\alpha_{[\beta\gamma],\mu\nu}&=&-\frac13\delta^\alpha_{[\beta} K^\lambda_{\gamma]\lambda,\mu\nu}-\Gamma^\lambda_{\mu[\gamma,\nu}\Gamma^\alpha_{\beta]\lambda}-\Gamma^\lambda_{\mu[\gamma}\Gamma^\alpha_{\beta]\lambda,\nu}
\\
&+&\frac13\delta^\alpha_{[\beta}\Gamma^\lambda_{\mu\gamma],\nu}\Gamma^\rho_{\rho\lambda}+\frac13\delta^\alpha_{[\beta}\Gamma^\lambda_{\mu\gamma]}\Gamma^\rho_{\rho\lambda,\nu}-\frac13\delta^\alpha_{[\beta}\Gamma^\lambda_{\mu\rho,\nu}\Gamma^\rho_{\gamma]\lambda}-\frac13\delta^\alpha_{[\beta}\Gamma^\lambda_{\mu\rho}\Gamma^\rho_{\gamma]\lambda,\nu} \ .
\eeann
\end{prop}
\begin{proof}
The functions $f_{\sigma\rho\mu,\nu}$ are given by \eqref{eq:constmetric3}. Now, from \eqref{eq:constcon} we obtain that
$$
\left(\frac{\partial^2 H}{\partial g_{\rho\sigma}\partial g_{\mu\nu}}-
\frac{\partial^2 L_\alpha^{\beta\gamma,\lambda}}{\partial g_{\rho\sigma}\partial g_{\mu\nu}}\Gamma^\alpha_{\beta\gamma,\lambda}\right)=0
\quad ; \quad \mbox{\rm (on ${\cal S}_{sh}$)}\ ,
$$
and therefore \eqref{eq:cons:conec} becomes 
$$
\left(\Gamma^\alpha_{\beta\gamma,\nu}\frac{\partial^2 H}{\partial \Gamma^\alpha_{\beta\gamma}\partial g_{\rho\sigma}}-\frac{\partial L_\alpha^{\beta\gamma,\mu}}{\partial g_{\rho\sigma}}f^\alpha_{\beta\gamma\mu,\nu}\right)=0 
\quad ; \quad \mbox{\rm (on ${\cal S}_{sh}$)}\ .
$$
A particular solution to these equations is $$
f^\alpha_{\beta\gamma\mu,\nu}=
\Gamma^\lambda_{\mu\gamma,\nu}\Gamma^\alpha_{\beta\lambda}+\Gamma^\lambda_{\mu\gamma}\Gamma^\alpha_{\beta\lambda,\nu}
\quad ; \quad \mbox{\rm (on ${\cal S}_{sh}$)}\ .
$$
Now, we need to find a general solution $h^\alpha_{\beta\gamma\mu,\nu}$
to the homogeneous equation, which is just
\eqref{eq:sol:hom}, but on ${\cal S}_{sh}$. Thus, proceeding as in
the proof of Proposition \ref{pr:so1}, we obtain that
$$
h^\alpha_{\beta\gamma,\mu\nu}=C^\alpha_{\beta\gamma,\mu\nu}+K^\alpha_{\beta\gamma,\mu\nu} 
\quad ; \quad \mbox{\rm (on ${\cal S}_{sh}$)} \ , 
$$
for $C_{\beta\mu\nu}\in C^\infty(J^1\pi)$ and 
$K^{\alpha}_{\beta\gamma,\mu\nu}\in C^\infty(J^1\pi)$ satisfying that
$$
C^\alpha_{\beta\gamma,\mu\nu}=C_{\beta\mu\nu}\delta^\alpha_\gamma\quad ,\quad K^\lambda_{\lambda\gamma,\mu\nu}=0\quad ,\quad K^\lambda_{\beta\gamma,\lambda\nu}+K^\lambda_{\gamma\beta,\lambda\nu}=0
\quad ; \quad \mbox{\rm (on ${\cal S}_{sh}$)}\ . 
$$
By construction, the solutions obtained in this way satisfy all the tangent conditions
on the constraints given in Theorem \ref{finalconstraints}, except
$$
\Lie(X_\nu)r^\alpha_{\beta\gamma,\mu}=0
\quad ; \quad  \mbox{\rm (on ${\cal S}_{sh}$)}\ ;
$$ 
and these equations lead to the last conditions.
\end{proof}

\noindent {\bf Comments}:
\bit
\item
It is important to point out that, up to the torsion constraints $t^\alpha_{\beta\gamma}$,
all the other constraints appear as a consequence of demanding the semiholonomy condition
on the multivector fields solution to the field equations \eqref{fundeqs}.
\item
From the constraints 
$m_{\rho\sigma\,\mu}=0$ and $t^\alpha_{\beta\gamma}=0$
in Theorem \ref{finalconstraints}, and Proposition \ref{pr:pm}
we obtain that 
$$
T^\alpha_{\beta\alpha}=0  \ \Longleftrightarrow \ 
T^\alpha_{\beta\gamma}=0  \ \Longleftrightarrow \ 
\nabla^\Gamma g=0 \ .
$$
Thus, any of these conditions are necessary and sufficient to assure that the connection becomes
the Levi-Civita connection.
This result completes the already known fact that
the vanishing of the trace torsion
is sufficient for the connection to be the Levi-Civita connection
(see, for instance, \cite{art:Capriotti2,pons}).
\eit

\subsubsection{Holonomic multivector fields (elements of $\ker^4_H\Omega_{\Lag_{\rm EP}}$): Integrability constraints}

The last step is to look for holonomic (i.e., integrable and semiholonomic) multivector fields. 
Locally, a transverse multivector field is integrable if, and only if, $[X_\mu,X_\nu]=0$ for any $\mu,\nu=0,1,2,3$. 
In any open of $U\subset\mathcal{S}_f$ where this condition holds, 
there exist integrable sections for the multivector field defined on $\pi(U)$. 
In general, integrable multivector fields could only
exist in a submanifold $\mathcal{S}_f$ of ${\cal S}_{sh}$. In this 
Section we obtain this submanifold, giving the constraints which are 
sufficient to assure that there are an holonomic multivector field; because every point of the submanifold can be reached by a section which is a solution to the field equations. This last result is proven in Proposition \ref{par:int2}, using the equivalence between the Metric-Affine and the Hilbert-Einstein models presented in Section \ref{se:rel:he-ma}.

Consider the following general expression
$$
[X_\mu,X_\nu]=F^\epsilon\frac{\partial}{\partial x^\epsilon}+
\sum_{\alpha\leq\beta}\left(F_{\alpha\beta}\frac{\partial}{\partial g_{\alpha\beta}}+
F_{\alpha\beta,\epsilon}\frac{\partial}{\partial g_{\alpha\beta,\epsilon}}\right)+
F^\alpha_{\beta\gamma}\frac{\partial}{\partial\Gamma^\alpha_{\beta\gamma}}+
F^\alpha_{\beta\gamma,\epsilon}\frac{\partial}{\partial \Gamma^\alpha_{\beta\gamma,\epsilon}}=0 \ ; \ \mbox{\rm (on ${\cal S}_{sh}$)} \ .
$$
Next, we have to take into account \eqref{ELmvf}.
First, the coefficients $F^\epsilon\vert_{{\cal S}_{sh}}=0$,
necessarily (and this is the reason for imposing the vector field to vanish, which is a stronger condition than being inside the distribution).
From the conditions $F_{\alpha\beta}\vert_{{\cal S}_{sh}}=0$, we derive that 
$$
f_{\rho\sigma\mu,\nu}-f_{\rho\sigma\nu,\mu}=0
\quad ; \quad  \mbox{\rm (on ${\cal S}_{sh}$)} \ .
$$
which are new restrictions on the functions
$\Gamma^\alpha_{\beta\gamma,\mu}$, specifically
\bea
i_{\rho\sigma,\mu\nu}&=&g_{\rho\gamma}\Gamma^\gamma_{[\nu\lambda}\Gamma^\lambda_{\mu]\sigma}+g_{\sigma\gamma}\Gamma^\gamma_{[\nu\lambda}\Gamma^\lambda_{\mu]\rho}+g_{\rho\lambda}\Gamma^\lambda_{[\mu\sigma,\nu]}+g_{\sigma\lambda}\Gamma^\lambda_{[\mu\rho,\nu]}+\frac23g_{\rho\sigma}T^\lambda_{\lambda[\mu,\nu]}
\nonumber \\
&=&g_{\rho\lambda}K^\lambda_{[\nu\sigma\mu]}+g_{\sigma\lambda}K^\lambda_{[\nu\rho\mu]}+2g_{\rho\sigma}T^\lambda_{\mu\nu}\Gamma^\gamma_{\gamma\lambda}=0\quad ; \quad \mbox{\rm (on ${\cal S}_{sh}$)} ,
\label{intcons}
\eea
where the functions $K^\alpha_{\beta\gamma\mu}$ arise from proposition \ref{pr:so1}.
(Observe that these constraints are symmetric in the indices $\rho,\sigma$ and skewsymmetric in the indices
$\mu,\nu$).
In a similar way, from the conditions $F^\alpha_{\beta\gamma}\vert_{{\cal S}_{sh}}=0$, we obtain that 
$$
f^\alpha_{\beta\gamma\mu,\nu}-f^\alpha_{\beta\gamma\nu,\mu}=0 
\quad ; \quad \mbox{\rm (on ${\cal S}_{sh}$)}\ ,
$$
which impose some restrictions on the possible solutions, namely:
$$
C_{\beta[\mu\nu]}=\Gamma^\lambda_{[\mu\beta,\nu]}\Gamma^\sigma_{\sigma\lambda}+\Gamma^\lambda_{[\mu\beta}\Gamma^\sigma_{\sigma\lambda,\nu]}
\quad ; \quad \mbox{\rm (on ${\cal S}_{sh}$)}\ ,
$$
$$K^\alpha_{\beta\gamma,[\mu\nu]}=-
\Gamma^\lambda_{[\mu\gamma,\nu]}\Gamma^\alpha_{\beta\lambda}-\Gamma^\lambda_{[\mu\gamma}\Gamma^\alpha_{\beta\lambda,\nu]}-C_{\beta[\mu\nu]}\delta^\alpha_\gamma 
\quad ; \quad \mbox{\rm (on ${\cal S}_{sh}$)}\ .
$$
The coefficients $F_{\alpha\beta,\gamma}$ vanish automatically on ${\cal S}_{sh}$ as long as 
$(f^\alpha_{\beta\gamma\mu,\nu}-f^\alpha_{\beta\gamma\nu,\mu})\vert_{{\cal S}_{sh}}=0$. 
Finally, the conditions $F^\alpha_{\beta\gamma,\epsilon}=0$
lead to a system of PDE on the functions
$C_{\beta\mu\nu},K^{\alpha}_{\beta\gamma,\mu\nu}$
which may originate new constraints.
The tangency conditions on the constraints $i_{\rho\sigma,\mu\nu}$ give
\beann
g_{\alpha\lambda}K^\lambda_{[\nu\beta\mu],\xi}+g_{\beta\lambda}K^\lambda_{[\nu\alpha\mu],\xi}&=&
-2g_{\alpha\beta,\xi}T^\lambda_{\mu\nu}\Gamma^\sigma_{\sigma\lambda}-2g_{\alpha\beta}T^\lambda_{\mu\nu,\xi}\Gamma^\sigma_{\sigma\lambda}-2g_{\alpha\beta}T^\lambda_{\mu\nu}\Gamma^\sigma_{\sigma\lambda,\xi}
\\ & & -g_{\alpha\lambda,\xi}K^\lambda_{[\nu\beta\mu]}-g_{\beta\lambda,\xi}K^\lambda_{[\nu\alpha\mu]} \quad ; \quad \mbox{\rm (on ${\cal S}_{sh}$)} \ .
\eeann

In what follows, we will denote
${\rm j}_f\colon\mathcal{S}_f\hookrightarrow J^1\pi$
the constraint submanifold defined by all the constraints
$c^{\mu\nu}$, $m_{\sigma\rho,\mu}$, $t^\alpha_{\beta\gamma}$,
$r^\alpha_{\beta\gamma,\nu}$ and $i_{\rho\sigma,\mu\nu}$. 
This is the submanifold where there exist  holonomic multivector fields solution to the field equations 
which are tangent to $\mathcal{S}_f$, as it is shown in Proposition \ref{par:int2}.
Notice that ${\cal S}_f$ is a subbundle of $J^1\pi$ over ${\rm E}$ 
and $M$ and, thus, we have the natural submersions
$$
\pi^1_f=\pi^1\circ {\rm j}_f\colon{\cal S}_f\to{\rm E}
\quad , \quad
\overline{\pi}^1_f=\overline{\pi}^1\circ {\rm j}_f\colon{\cal S}_f\to M \ .
$$

\subsection{Symmetries and gauge symmetries}

(See the Appendix \ref{appends} for reviewing
the basic definitions and considerations about symmetries and gauge symmetries
for singular Lagrangian field theories).

\subsubsection{Gauge symmetries of the Einstein-Palatini model}

\begin{prop} 
\label{pr:gv:la}
The natural gauge vector fields for the Einstein-Palatini model are the vector fields 
$X\in\mathfrak{X}(J^1\pi)$ whose local expressions are
$$
X=C_\beta\delta^\alpha_\gamma\frac{\partial}{\partial \Gamma^\alpha_{\beta\gamma}}
+D_\mu C_{\beta}\delta^\alpha_\gamma\frac{\partial}{\partial \Gamma^\alpha_{\beta\gamma,\mu}}
\quad , \quad
C_\beta\in C^\infty(J^1\pi)
\quad ; \quad \mbox{\rm (on ${\cal S}_f$)} \ .
$$
\end{prop}
\begin{proof}
Consider a vector field
$$
X=f^\mu\frac{\partial}{\partial x^\mu}+\sum_{\rho\leq\sigma}\left(f_{\rho\sigma}\frac{\partial}{\partial g_{\rho\sigma}}+f_{\rho\sigma,\mu}\frac{\partial}{\partial g_{\rho\sigma,\mu}}\right)+f^\alpha_{\beta\gamma}\frac{\partial}{\partial \Gamma^\alpha_{\beta\gamma}}+f^\alpha_{\beta\gamma,\mu}\frac{\partial}{\partial \Gamma^\alpha_{\beta\gamma,\mu}}
\in\mathfrak{X}(J^1\pi)\ .
$$
As ${\cal S}_f$ is a bundle over $M$,
clearly $X$ is $\overline{\pi}^1_f$-vertical if, and only if, it is $\overline{\pi}^1$-vertical. 
Therefore $\overline{\pi}^1_*X=0$
if, and only if, $f^\mu=0$. Furthermore
\beann
\inn(X)\Omega_{\Lag_{\rm EP}}&=&
\left(\sum_{\rho\leq\sigma}\frac{\partial H}{\partial g_{\rho\sigma}}f_{\rho\sigma}+\frac{\partial H}{\partial \Gamma^\alpha_{\beta\gamma}}f^\alpha_{\beta\gamma}\right)\d^4x-\sum_{\rho\leq\sigma}\frac{\partial L^{\beta\gamma,\mu}_{\alpha}}{\partial g_{\rho\sigma}}f_{\rho\sigma}\d\Gamma^\alpha_{\beta\gamma}\wedge \d^3x_{\mu}
\\ &&
-\sum_{\rho\leq\sigma}\frac{\partial L^{\beta\gamma,\mu}_{\alpha}}{\partial g_{\rho\sigma}}f^\alpha_{\beta\gamma}\d g_{\rho\sigma}\wedge \d^3x_{\mu}=0\ .
\eeann
After doing the pullback ${\rm j}_f^*\inn(X)\Omega_{\Lag_{\rm EP}}$,
we obtain the terms 
$$
{\rm j}_f^*\d \Gamma^\alpha_{\beta\gamma}=\frac12\d \Gamma^\alpha_{(\beta\gamma)}+\frac16\delta^\alpha_\beta \d T^r_{r\gamma}-\frac16\delta^\alpha_\gamma \d T^r_{r\beta} \ .
$$
As every coefficient must vanish, taking
in particular the corresponding to the factor 
$\d\Gamma^\alpha_{(\beta\gamma)}$, we obtain that 
$f_{\rho\sigma}\vert_{{\cal S}_f}=0$. Indeed:
$$
0=\delta_\beta^\alpha(\frac13g_{\mu \nu}g_{\gamma \lambda}-
\frac16g_{\mu\gamma}g_{\nu\lambda})\sum_{\rho\leq\sigma}f_{\rho\sigma} \frac{\partial L_{\alpha}^{(\beta\gamma),\mu}}{\partial g_{\rho\sigma}}=
\sum_{\nu \leq \lambda}(f_{\nu\lambda}+f_{\lambda\nu})\Rightarrow f_{\rho\sigma}=0
\quad ; \quad \mbox{\rm (on ${\cal S}_f$)}\ .
$$
Using these results, the problem is reduced to find $f^\alpha_{\beta\gamma}\in C^\infty(J^1\pi)$
such that
\bea
\label{eq:peg5}
f^\alpha_{\beta\gamma} \frac{\partial L_{\alpha}^{\beta\gamma,\mu}}{\partial g_{\rho\sigma}}&=&0 \quad ; \quad \mbox{\rm (on ${\cal S}_f$)}\ ,
\\ \label{eq:peg2}
f^\alpha_{\beta\gamma} \frac{\partial H}{\partial \Gamma^\alpha_{\beta\gamma} }&=&0
\quad ; \quad \mbox{\rm (on ${\cal S}_f$)}\ .
\eea
Multiplying \eqref{eq:peg5} by $g_{\mu\rho}g_{\nu\sigma}$ we obtain:
$$
f^\alpha_{\beta\gamma}+f^\alpha_{\gamma\beta}=f^r_{r\beta}\delta^\alpha_\gamma+f^r_{r\beta}\delta^\alpha_\gamma+(f^\alpha_{rs}g^{rs}-f^r_{rs}g^{\alpha s})g_{\beta\gamma} 
\quad ; \quad \mbox{\rm (on ${\cal S}_f$)}\ .
$$
This system has two kinds of solutions. First, there are the {\sl trace solutions}, given by 
$f^\alpha_{\beta\gamma}=C^\alpha_{\beta\gamma}=C_\beta\delta^\alpha_\gamma$, 
for any arbitrary function $C_\beta\in C^\infty(J^1\pi)$ \cite{pons}. 
Second, for other solutions $f^\alpha_{\beta\gamma}$, we have that $K^\alpha_{\beta\gamma}=f^\alpha_{\beta\gamma}-C^\alpha_{\beta\gamma}$, 
with $C_\gamma=f^\nu_{\nu\gamma}$. 
Contracting indices $\alpha,\beta$ we obtain $K^\alpha_{\alpha\gamma}=0$. 
Since \eqref{eq:peg5} are linear, $K^\alpha_{\beta\gamma}$ are also solutions, therefore
\beann
K^\alpha_{\beta \gamma}+K^\alpha_{\gamma\beta} = K^\alpha_{\rho\sigma}g^{\rho\sigma}g_{\beta \gamma}\ \Rightarrow \ 
K^\alpha_{\beta \gamma}+K^\alpha_{\gamma\beta}=\frac12(K^\alpha_{\rho\sigma}+K^\alpha_{\sigma\rho})g^{\rho\sigma}g_{\beta \gamma}\quad \Rightarrow & & 
\\
 g^{\beta\gamma}(K^\alpha_{\beta \gamma}+K^\alpha_{\gamma\beta})=2(K^\alpha_{\rho\sigma}+K^\alpha_{\sigma\rho})g^{\rho\sigma}\ \Rightarrow\
-g^{\beta\gamma}(K^\alpha_{\beta \gamma}+K^\alpha_{\gamma\beta})=0
\quad ; & & \mbox{\rm (on ${\cal S}_f$)}\ ,
\eeann
which implies $K^\alpha_{\rho\sigma}g^{\rho\sigma}=0$, thus $K^\alpha_{\beta \gamma}+K^\alpha_{\gamma\beta}=0$.
These are called the {\sl torsion solutions}. 
Both kinds of solutions fulfil \eqref{eq:peg2}; in fact,
\beann
C_\beta\delta^\alpha_{\gamma} \frac{\partial H}{\partial \Gamma^\alpha_{\beta\gamma} }&=&\varrho C_\beta\left(g^{\mu \beta}\Gamma^r_{r\mu}+g^{\mu r}\Gamma^\beta_{\mu r}-g^{r \beta}\Gamma^\mu_{\mu r}-g^{\mu \nu}\Gamma^\beta_{\mu\nu}\right)=0 
\quad ; \quad \mbox{\rm (on ${\cal S}_f$)}\ ;
\\
K^\alpha_{\beta\gamma} \frac{\partial H}{\partial \Gamma^\alpha_{\beta\gamma} }&=&\varrho \left(K^\alpha_{\beta\gamma}(g^{\mu \gamma}\Gamma^\beta_{\alpha\mu}+g^{\mu \beta}\Gamma^\gamma_{\mu\alpha })-K^\alpha_{\beta\gamma}g^{\gamma\beta}\Gamma^\mu_{\mu \alpha}-K^\lambda_{\lambda\gamma}g^{\mu \nu}\Gamma^\gamma_{\mu\nu}\right)
\\
&=&\varrho K^\alpha_{\beta\gamma}g^{\mu \gamma}T^\beta_{\alpha\mu}=\varrho K^\alpha_{\beta\gamma}g^{\mu \gamma}(\frac13\delta^\beta_\alpha T^r_{[r\mu]}-\frac13\delta_\mu^\beta T^r_{[r\alpha]})=0 
\quad ; \quad \mbox{\rm (on ${\cal S}_f$)}\ .
\eeann
Now we impose the tangency condition on the torsion constraints 
$$
0=\Lie(X)t^\alpha_{\beta\gamma}=f^\alpha_{[\beta\gamma]}-\frac13\delta^\alpha_\beta f^r_{[r\gamma]}+\frac13\delta^\alpha_\gamma f^r_{[r\beta]}=2K^\alpha_{\beta\gamma}=0 
\quad ; \quad \mbox{\rm (on ${\cal S}_f$)}\ .
$$
The trace solutions are tangent, but the torsion are not. Before checking the other constraints, let us impose the condition of being natural. 
The local conditions for a $\overline{\pi}^1$-vertical vector field to be natural are that
$f_{\rho\sigma},f^\alpha_{\beta\gamma}$ are $\overline{\pi}^1$-projectable,
that $f_{\rho\sigma,\mu}=D_\mu f_{\rho\sigma}$, and that $f^\alpha_{\beta\gamma,\sigma}=D_\sigma f^\alpha_{\beta\gamma}$. 
In our case, these conditions imply that 
$C_\beta\in C^\infty(J^1\pi)$ are $\overline{\pi}^1$-projectable, 
that $f^\alpha_{\beta\gamma,\mu}\vert_{\mathcal{S}_f}=\delta^\alpha_\gamma D_{\mu}C_\beta$, 
and that $f_{\rho\sigma,\mu}\vert_{\mathcal{S}_f}=~0$.
The tangency condition on the pre-metricity constraints is
\beann{}
0&=&\Lie(X)m_{\rho\sigma,\mu}=L(X)\left(g_{\rho\sigma,\mu}-g_{\sigma\lambda}\Gamma^\lambda_{\mu\rho}-g_{\rho\lambda}\Gamma^\lambda_{\mu\sigma}-\frac{2}{3}g_{\rho\sigma}T^\lambda_{\lambda\mu}\right)
\\
&=&f_{\rho\sigma,\mu}-g_{\sigma\lambda}\delta^\lambda_{\rho}C_\mu-g_{\rho\lambda}\delta^\lambda_{\sigma}C_\mu-\frac{2}{3}g_{\rho\sigma}(C_\lambda\delta^\lambda_\mu-C_\mu\delta^\lambda_\lambda)=0 
\quad ; \quad \mbox{\rm (on ${\cal S}_f$)}\ .
\eeann{}
As $f^\alpha_{\beta\gamma}\vert_{{\cal S}_f}=C_\beta\delta^\alpha_\gamma$, then
$\displaystyle \frac{\partial L_\alpha^{\beta\gamma,\sigma}}{\partial g_{\mu\nu}}f^\alpha_{\beta\gamma,\sigma}=0$
(see  Proposition \ref{pr:so1}), and hence
$$
\Lie(X)c^{\mu\nu}=\frac{\partial \varrho g^{\alpha\beta}}{\partial g_{\mu\nu}}\left(C_{\beta}\Gamma^{\sigma}_{\sigma\alpha}+\Gamma^{\gamma}_{\beta\alpha}C_{\gamma}-C_{\beta}\Gamma^{\sigma}_{\sigma\alpha}-\Gamma^{\gamma}_{\beta\alpha}C_{\gamma}\right)-\frac{\partial L_\alpha^{\beta\gamma,\sigma}}{\partial g_{\mu\nu}}f^\alpha_{\beta\gamma,\sigma}=0
\quad ; \quad \mbox{\rm (on ${\cal S}_f$)}\ .
$$ 
The tangency condition on $r^\alpha_{\beta\gamma,\nu}$ 
involves only the functions $f^\alpha_{\beta\gamma,\nu}$:
$$
0=\Lie(X)r^\alpha_{\beta\gamma,\nu}=f^\alpha_{[\beta\gamma],\nu}-\frac13\delta^\alpha_\beta f^r_{[r\gamma],\nu}+\frac13\delta^\alpha_\gamma f^r_{[r\beta],\nu}
\quad ; \quad \mbox{\rm (on ${\cal S}_f$)}\ .
$$
The trace solutions fulfil this condition automatically. Finally, the tangency condition for the integrability constraints \eqref{intcons} holds:
\beann{}
\Lie(X)i_{\rho\sigma,\mu\nu}
&=&g_{\rho\gamma}C_{[\nu}\Gamma^\lambda_{\mu]\sigma}+g_{\rho\gamma}\Gamma^\lambda_{[\nu\sigma}C_{\mu]}+g_{\sigma\gamma}C_{[\nu}\Gamma^\lambda_{\mu]\rho}+g_{\sigma\gamma}\Gamma^\lambda_{[\nu\rho}C_{\mu]}
\\
&+&g_{\rho\sigma}C_{[\mu\nu]}+g_{\rho\sigma}C_{[\mu\nu]}-2g_{\rho\sigma}C_{[\mu\nu]}=0 
\quad ; \quad \mbox{\rm (on ${\cal S}_f$)}\ .
\eeann{}
\end{proof}

\subsubsection{Lagrangian  symmetries of the Einstein-Palatini model}\label{sec:lag:sym}

Let $F$ be a diffeomorphism in $M$. For every $x\in M$,
if $g_x$ is a metric in $\Tan_xM$,  then
$F_*g_x=(F^{-1})^*(g_x)$ is also a metric with the same signature as $g_x$.
In the same way, as a connection $\Gamma_x$ is a $(1,1)$-tensor in $\Tan_xM$ \cite{EMR-2018}, denoting also by $F_*$
the induced action of $F$ on the tensor algebra, we define:

\begin{definition}
Let $F\colon M\to M$ be a diffeomorphism.
The {\rm canonical lift of $F$ to the bundle ${\rm E}$}
is the diffeomorphism ${\mathcal{F}}\colon E\to {\rm E}$ 
defined as follows: for every $(x,g_x,\Gamma_x)\in E$, then
${\mathcal{F}}(x,g_x,\Gamma_x):=(F(x),F_*g_x,F_*\Gamma_x)$
(Thus $\pi\circ{\mathcal{F}}=F\circ\pi$).

Let $Z\in\vf (M)$.
The {\rm canonical lift of $Z$ to the bundle ${\rm E}$}
is the vector field $Y_Z\in\vf({\rm E})$ whose associated
local one-parameter groups of diffeomorphisms  ${\cal F}_t$
are the canonical lifts to the bundle ${\rm E}$
of the local one-parameter groups of diffeomorphisms $F_t$ of $Z$.
\end{definition}

In coordinates, if
$\displaystyle Z=f^\mu(x)\frac{\partial}{\partial x^\mu}\in\mathfrak{X}(M)$,
the canonical lift of $Z$ to the bundle $E\to M$ is
\beann
Y_Z&=&f^\mu\frac{\partial}{\partial x^\mu}-\sum_{\alpha\leq \beta}\left(\frac{\partial f^\lambda}{\partial x^\alpha}g_{\lambda\beta}+\frac{\partial f^\lambda}{\partial x^\beta}g_{\lambda\alpha}\right)\frac{\partial}{\partial g_{\alpha\beta}}
\\ & &
+\left(\frac{\partial f^\alpha}{\partial x^\lambda}\Gamma^\lambda_{\beta\gamma}-\frac{\partial f^\lambda}{\partial x^\beta}\Gamma^\alpha_{\lambda\gamma}-\frac{\partial f^\lambda}{\partial x^\gamma}\Gamma^\alpha_{\beta\lambda}-\frac{\partial^2f^\alpha}{\partial x^\beta\partial x^\gamma}\right)\frac{\partial}{\partial \Gamma^\alpha_{\beta\gamma}}
\in\vf({\rm E})\  .
\eeann
Furthermore, every diffeomorphism in ${\rm E}$ induces a diffeomorphism in $J^1\pi$. The vector fields generating these transformations are canonical liftings
$X=j^1Y$, for $Y\in\vf(E)$.
Hence, for the above ones we have
\beann
j^1Y_Z&=&
f^\mu\frac{\partial}{\partial x^\mu}-\sum_{\alpha\leq \beta}\left(\frac{\partial f^\lambda}{\partial x^\alpha}g_{\lambda\beta}+\frac{\partial f^\lambda}{\partial x^\beta}g_{\lambda\alpha}\right)\frac{\partial}{\partial g_{\alpha\beta}}
\\ & &
-\sum_{\alpha\leq\beta}\left(\frac{\partial^2f^\nu}{\partial x^\alpha\partial x^\mu}g_{\nu\beta}+\frac{\partial^2f^\nu}{\partial x^\beta\partial x^\mu}g_{\alpha\nu}+\frac{\partial f^\nu}{\partial x^\alpha}g_{\nu\beta,\mu}+\frac{\partial f^\nu}{\partial x^\beta}g_{\alpha\nu,\mu}+\frac{\partial f^\nu}{\partial x^\mu}g_{\alpha\beta,\nu}\right)\frac{\partial}{\partial g_{\alpha\beta,\mu}}
\\ & &
+\left(\frac{\partial f^\alpha}{\partial x^\lambda}\Gamma^\lambda_{\beta\gamma}-\frac{\partial f^\lambda}{\partial x^\beta}\Gamma^\alpha_{\lambda\gamma}-\frac{\partial f^\lambda}{\partial x^\gamma}\Gamma^\alpha_{\beta\lambda}-\frac{\partial^2f^\alpha}{\partial x^\beta\partial x^\gamma}\right)\frac{\partial}{\partial \Gamma^\alpha_{\beta\gamma}}
\\ & &
+\left(\frac{\partial f^\alpha}{\partial x^\lambda}\Gamma^\lambda_{\beta\gamma,\mu}-\frac{\partial f^\lambda}{\partial x^\beta}\Gamma^\alpha_{\lambda\gamma,\mu}-\frac{\partial f^\lambda}{\partial x^\gamma}\Gamma^\alpha_{\beta\lambda,\mu}-\frac{\partial f^\lambda}{\partial x^\mu}\Gamma^\alpha_{\beta\gamma,\lambda}\right.
\\& &
+\left.\frac{\partial^2 f^\alpha}{\partial x^\lambda\partial x^\mu}\Gamma^\lambda_{\beta\gamma}-\frac{\partial^2 f^\lambda}{\partial x^\beta\partial x^\mu}\Gamma^\alpha_{\lambda\gamma}-\frac{\partial^2 f^\lambda}{\partial x^\gamma\partial x^\mu}\Gamma^\alpha_{\beta\lambda}
-\frac{\partial^3f^\alpha}{\partial x^\beta\partial x^\gamma\partial x^\mu}\right)\frac{\partial}{\partial \Gamma^\alpha_{\beta\gamma,\mu}}
\\ &\equiv&
f^\mu\frac{\partial}{\partial x^\mu}+\sum_{\alpha\leq \beta}Y_{\alpha\beta}\frac{\partial}{\partial g_{\alpha\beta}}+\sum_{\alpha\leq\beta}Y_{\alpha\beta\mu}\frac{\partial}{\partial g_{\alpha\beta,\mu}}
+ Y^\alpha_{\beta\gamma}\frac{\partial}{\partial\Gamma^\alpha_{\beta\gamma}}+Y^\alpha_{\beta\gamma\mu}\frac{\partial}{\partial \Gamma^\alpha_{\beta\gamma,\mu}}
\in\vf(J^1\pi)\ .
\eeann

We have that $\mathcal{L}_{\rm EP}$ 
is invariant under diffeomorphisms (using the constraints $c^{\mu\nu}$).
Then, for every $Z\in\mathfrak{X}(M)$, we have that
$\Lie (j^1Y_Z)\mathcal{L}_{\rm EP}\vert_{{\mathcal S}_f}=0$. 
In addition, $j^1Y_Z$ are tangent to ${\cal S}_f$. In fact, as they are natural vector fields
that leave the Einstein-Palatini Lagrangian invariant, 
then the corresponding Euler-Lagrange equations are also invariant,
and hence for the constraints
$c^{\mu\nu}$
we have that
$$
\Lie(j^1Y_Z)c^{\mu\nu}=-\left(\frac{\partial f^\mu}{\partial x^\rho}\delta^\nu_{\sigma}+\frac{\partial f^\nu}{\partial x^\sigma}\delta^\mu_{\rho}\right)\left(\frac{\partial H}{\partial g_{\rho\sigma}}-\frac{\partial L^{\beta\gamma,\lambda}_\alpha}{\partial g_{\rho\sigma}}\Gamma^{\alpha}_{\beta\gamma,\lambda}\right)=0 \quad ; \quad \mbox{\rm (on ${\cal S}_f$)}\ ;
$$
while for the other constraints, after a long calculation, we obtain
\beann{}
\Lie(j^1Y_Z)m_{\rho\sigma,\mu}= 
\left(-\frac{\partial f^\alpha}{\partial x^\rho}\delta^\beta_{\sigma}\delta^\nu_{\mu}-\frac{\partial f^\beta}{\partial x^\sigma}\delta^\alpha_{\rho}\delta^\nu_{\mu}-\frac{\partial f^\nu}{\partial x^\mu}\delta^\alpha_{\rho}\delta^\beta_{\sigma}\right)m_{\alpha\beta,\nu} =0;
\ \mbox{\rm (on ${\cal S}_f$)} \ ,
\\
\Lie(j^1Y_Z)t^\alpha_{\beta\gamma}= \left(\frac{\partial f^\alpha}{\partial x^\lambda}\delta^\rho_{\beta}\delta^\sigma_{\gamma}-\frac{\partial f^\rho}{\partial x^\beta}\delta^\alpha_{\lambda}\delta^\sigma_{\gamma}-\frac{\partial f^\sigma}{\partial x^\gamma}\delta^\alpha_{\lambda}\delta^\rho_{\beta}\right)t^\lambda_{\rho\sigma}=0;
\ \mbox{\rm (on ${\cal S}_f$)} \ ,
\\
\Lie(j^1Y_Z)r^\alpha_{\beta\gamma,\nu}= \left(\frac{\partial f^\alpha}{\partial x^\lambda}\delta^\rho_{\beta}\delta^\sigma_{\gamma}\delta^\tau_{\nu}-\frac{\partial f^\rho}{\partial x^\beta}\delta^\alpha_{\lambda}\delta^\sigma_{\gamma}\delta^\tau_{\nu}-\frac{\partial f^\sigma}{\partial x^\gamma}\delta^\alpha_{\lambda}\delta^\rho_{\beta}\delta^\tau_{\nu}-\frac{\partial f^\tau}{\partial x^\nu}\delta^\alpha_{\lambda}\delta^\rho_{\beta}\delta^\sigma_{\gamma}\right)r^\lambda_{\rho\sigma,\tau}=0;
\ \mbox{\rm (on ${\cal S}_f$)} \ ,
\\
\Lie(j^1Y_Z)i_{\rho\sigma,\mu\nu}= 
\left(-\frac{\partial f^\alpha}{\partial x^\rho}\delta^\beta_{\sigma}\delta^\lambda_{\mu}\delta^\gamma_\nu-\frac{\partial f^\beta}{\partial x^\sigma}\delta^\alpha_{\rho}\delta^\lambda_{\mu}\delta^\gamma_\nu-\frac{\partial f^\lambda}{\partial x^\mu}\delta^\alpha_{\rho}\delta^\beta_{\sigma}\delta^\gamma_\nu-\frac{\partial f^\gamma}{\partial x^\nu}\delta^\alpha_{\rho}\delta^\beta_{\sigma}\delta^\lambda_\mu\right)i_{\alpha\beta,\lambda\gamma} =0;
\ \mbox{\rm (on ${\cal S}_f$)} \ .
\eeann{}
Thus, these vector fields are natural infinitesimal Lagrangian symmetries 
and, hence, natural infinitesimal Noether symmetries. 
Then an associated conserved quantity to each $j^1Y_Z$ is
$\xi_{Y_Z}=\inn(j^1Y_Z)\Theta_{\mathcal{L}_{\rm EP}}$ 
(see the Appendix \ref{appends}), which has the local expression:
$$
\xi_{Y_Z}=\inn(j^1Y)\Theta_{\mathcal{L}_{\rm EP}}=
(L^{\beta\gamma,\mu}_\alpha Y^\alpha_{\beta\gamma}-Hf^\mu)\d^3x_\mu+f^\mu L_\alpha^{\beta\gamma,\nu}\d \Gamma^\alpha_{\beta\gamma}\wedge\d^2 x_{\mu\nu} \ .
$$
Finally, given a section $\psi_\Lag$ solution the field equations, 
the Noether current associated with $j^1Y_Z$ is
$$
\psi_\Lag^*\xi_{Y_Z}=
\psi_\Lag^*(L^{\beta\gamma,\mu}_\alpha(Y^\alpha_{\beta\gamma}-\Gamma^\alpha_{\beta\gamma,\lambda}f^\lambda)-f^\mu L_{\rm EP})\d^3x_\mu \ .
$$

\noindent {\bf Comment}:
The term ``gauge'' is also used in physics to refer 
the invariance of the equations 
with respect to changes of variables in the base manifold $M$.
Nevertheless, in our geometric formalism, 
these are really the natural symmetries that we have studied in this Section,
and they are mathematically different from the geometric gauge symmetries 
that we have analysed in the previous Section.

\section{The Metric-Affine model: Hamiltonian formalism}
\label{Sec4}

\subsection{Canonical Hamiltonian formalism}

(See, for instance, \cite{CCI-91,LMM-96,EMR-99b,art:Roman09} 
for the general setting of the multisymplectic Hamiltonian formalism for first-order field theories).

First, let ${\cal M}\pi\equiv\Lambda_2^m\Tan^*E$,
be the bundle of $m$-forms on
$E$ vanishing by the action of two $\pi$-vertical vector fields,
which is called the {\sl extended multimomentum bundle},
and has local coordinates $(x^\mu,\,g_{\alpha\beta},\Gamma^\nu_{\lambda\gamma},\,p^{\alpha\beta,\mu},
\,p^{\lambda\gamma,\mu}_\nu,p)$, $(0\leq\alpha\leq\beta\leq 3)$.
Consider the quotient bundle $J^1\pi^* = {\cal M}\pi / \Lambda^4_1(\Tan^*{\rm E})$
(where $\Lambda_1^4(\Tan^*{\rm E})$ is the bundle of $\pi$-semibasic $4$-forms in ${\rm E}$),
which is the \textsl{restricted multimomentum bundle} of ${\rm E}$, 
and is endowed with the natural projections
 \[
  \tau\colon \ J^1\pi^*\to{\rm E}\qquad , \quad
 \overline{\tau}=\pi\circ\tau\colon \ J^1\pi^*\to M\qquad , \quad
\mu\colon{\cal M}\pi\to J^1\pi^*   .
  \]
Induced local coordinates in $J^1\pi^* $ are
$(x^\mu,\,g_{\alpha\beta},\Gamma^\nu_{\lambda\gamma},\,p^{\alpha\beta,\mu},
\,p^{\lambda\gamma,\mu}_\nu)$, $(0\leq\alpha\leq\beta\leq 3)$.

The {\sl Legendre map} $\mathcal{FL}_{\rm EP}\colon J^1\pi\longrightarrow J^1\pi^*$
(see \cite{EMR-00a} for the definition) is given, for the Einstein-Palatini Lagrangian, by
\bea
\mathcal{FL}_{\rm EP}^{\quad *}\,x^\mu=x^\mu &,&
\mathcal{FL}_{\rm EP}^{\quad *}\,g_{\alpha\beta}=g_{\alpha\beta}\quad ,\quad 
\mathcal{FL}_{\rm EP}^{\quad *}\,\Gamma^\alpha_{\beta\gamma}=\Gamma^\alpha_{\beta\gamma}
\nonumber \\
\ \mathcal{FL}_{\rm EP}^{\quad *}\,p^{\alpha\beta,\mu}=\frac{\partial L_{\rm EP}}{\partial g_{\alpha\beta,\mu}}=0 &,& 
\mathcal{FL}_{\rm EP}^{\quad *}\,p_\alpha^{\beta\gamma,\mu}=
\frac{\partial L_{\rm EP}}{\partial \Gamma^\alpha_{\beta\gamma,\mu}}=
L^{\beta\gamma,\mu}_{\alpha}=\varrho(\delta_\alpha^\mu g^{\beta\gamma}-\delta_\alpha^\beta g^{\mu\gamma}) \ ,
\label{Legtrans}
\eea
and $p^{\alpha\beta,\mu}$ and $p_\alpha^{\beta\gamma,\mu}$ are called
the {\sl momentum coordinates of the metric} and {\sl the connection}, respectively.
We have that, for every $j^1_x\phi\in J^1\pi$,
$$
\Tan_{j^1_x\phi}\mathcal{FL}_{\rm EP}=\left( \begin{array}{ccccc}
1 & 0 & 0 & 0 & 0 \\
0 & 1 & 0 & 0 & 0 \\
0 & 0 & 1 & 0 & 0  \\
0 & 0 & 0 & 0 & 0  \\
0 &
\displaystyle\frac{\partial^2 L_{\rm EP}}{\partial g_{\nu\lambda}\partial \Gamma^\alpha_{\beta\gamma,\mu}} & 
0 &
0  &
0
\end{array} \right) \ .
$$
Locally we have that
\beq
\ker\,(\mathcal{FL}_{\rm EP})_*={\left<\derpar{}{g_{\alpha\beta,\mu}},
\derpar{}{\Gamma_{\lambda\gamma,\mu}^\nu}\right>}_{0\leq\alpha\leq\beta\leq3}\ .
\label{kerfl}
\eeq

\begin{prop} 
\label{PdifE}
$\mathcal{P}\equiv\mathcal{FL}_{\rm EP}(J^1\pi)$ is a closed submanifold of $J^1\pi^*$, 
which is diffeomorphic to ${\rm E}$.
\end{prop}
\begin{proof}
From \eqref{kerfl} we have that
$\mathcal{P}$ is locally defined by the constraints
\beq
p^{\alpha\beta,\mu}=0\quad , \quad 
p^{\beta\gamma,\mu}_\alpha=\varrho(\delta_\alpha^\mu g^{\beta\gamma}-\delta_\alpha^\beta g^{\mu\gamma}) \ ,
\label{hamcons}
\eeq
which remove the degrees of freedom in the fibers of the projection
$\tau$.
\end{proof}

If $\jmath\colon{\cal P}{\hookrightarrow} J^1\pi^*$
is the natural embedding, we denote by 
$$
\tau_\mathcal{P}=\tau\circ\jmath\colon\mathcal{P}\to E
\quad , \quad
\overline{\tau}_\mathcal{P}=\overline{\tau}\circ\jmath\colon\mathcal{P}\to M
$$
the restrictions to ${\cal P}$ of the natural projections
$\tau$ and $\overline{\tau}$.
Then, this Proposition states that $\tau_{\cal P}$ is a diffeomorphism.

\begin{prop} 
\label{PdifE2}
$\mathcal{L}_{\rm EP}$ is an {\sl almost-regular Lagrangian density}. 
\end{prop}
\begin{proof}
We prove the three conditions that define this concept:
First, as we have seen, $\mathcal{P}$ is a closed submanifold of $J^1\pi^*$.
Second, as $\dim\,\mathcal{P}={\rm rank}(\Tan_{j^1_x\phi}\mathcal{FL}_{\rm EP})=78$, 
for every $j^1_x\phi\in J^1\pi$,  then $\mathcal{FL}_{\rm EP}$ is a submersion
onto its image.
Finally, taking into account Proposition \ref{PdifE}, 
we conclude that the fibers of the Legendre map,
$(\mathcal{FL}_{\rm EP})^{-1}(\mathcal{FL}(j^1_x\phi))$,
are just the fibers of the projection $\pi^1$, and they 
are connected submanifolds of $J^1\pi$ (recall that $J^1\pi$ 
is connected because we are considering metrics with fixed signature).
\end{proof}

As a consequence of this Proposition, 
the existence of the Hamiltonian formalism for this system is assured.
In fact; consider the so-called {\sl extended Legendre map}
 \cite{CCI-91,art:Roman09},
$\widetilde{\mathcal{FL}}_{\rm EP}\colon J^1\pi\longrightarrow {\cal M}\pi$, which is locally given by 
\bea
\widetilde{\mathcal{FL}}_{\rm EP}^{\quad *}p &=&
L_{\rm EP}-g_{\alpha\beta,\mu}\,\frac{\partial L_{\rm EP}}{\partial g_{\alpha\beta,\mu}}-
\Gamma^\alpha_{\beta\gamma,\mu}\,\frac{\partial L_{\rm EP}}{\partial \Gamma^\alpha_{\beta\gamma,\mu}} 
\nonumber\\ &=&
L_{\rm EP}-\Gamma^{\alpha}_{\beta\gamma,\mu}L^{\beta\gamma,\mu}_{\alpha}=
-H=
\varrho g^{\alpha\beta}\left(\Gamma^{\gamma}_{\beta\alpha}\Gamma^{\sigma}_{\sigma\gamma}-\Gamma^{\gamma}_{\beta\sigma}\Gamma^{\sigma}_{\gamma\alpha}\right)\ ,
\label{consp}
\eea
and the same expressions as in \eqref{Legtrans} for the other coordinates.
Let
$\widetilde{\cal P}:=\widetilde{\mathcal{FL}}_{\rm EP}(J^1\pi)$ and
$\tilde\jmath\colon\widetilde{\cal P}\hookrightarrow{\cal M}\pi$
the natural imbedding, and denote by 
$\widetilde{\mathcal{FL}}_{\rm EP}^o$ and $\mathcal{FL}_{\rm EP}^o$
the restrictions of $\widetilde{\mathcal{FL}}_{\rm EP}$ and ${\cal F}\Lag_{\rm EP}$
to their images; that is, the maps defined by
$\widetilde{\mathcal{FL}}_{\rm EP} = \tilde\jmath \circ\widetilde{\mathcal{FL}}_{\rm EP}^o$ and
$\mathcal{FL}_{\rm EP} = \jmath \circ\mathcal{FL}_{\rm EP}^o$,
respectively.
It can be proved \cite{LMM-96} that the $\mu$-transverse submanifold
$\tilde{\cal P}$ is dif\/feomorphic to ${\cal P}$
(observe that \eqref{consp} is really a constraint in ${\cal M}\pi$).
 This diffeomorphism is denoted
 $\tilde\mu\colon\widetilde{\cal P}\to{\cal P}$,
 and it is just the restriction of the projection $\mu$
 to $\widetilde{\cal P}$.
 Then, taking $H_{\cal P}:=\tilde\mu^{-1}$, we have the diagram
 \[
 \begin{array}{cccc}
\begin{picture}(20,52)(0,0)
\put(0,0){\mbox{$J^1\pi$}}
\end{picture}
&
\begin{picture}(65,52)(0,0)
 \put(7,28){\mbox{$\widetilde{\mathcal{FL}}_{\rm EP}^o$}}
 \put(24,7){\mbox{$\mathcal{FL}_{\rm EP}^o$}}
 \put(0,7){\vector(2,1){65}}
 \put(0,4){\vector(1,0){65}}
\end{picture}
&
\begin{picture}(90,52)(0,0)
 \put(5,0){\mbox{${\mathcal P}$}}
 \put(5,42){\mbox{$\widetilde{\mathcal P}$}}
 \put(5,13){\vector(0,1){25}}
 \put(10,38){\vector(0,-1){25}}
 \put(-12,22){\mbox{$H_{\mathcal P}$}}
 \put(12,22){\mbox{$\tilde\mu$}}
 \put(30,45){\vector(1,0){55}}
 \put(30,4){\vector(1,0){55}}
 \put(48,12){\mbox{$\jmath$}}
 \put(48,33){\mbox{$\tilde\jmath$}}
 \end{picture}
&
\begin{picture}(15,52)(0,0)
 \put(0,0){\mbox{$J^1\pi^*$}}
 \put(0,41){\mbox{${\mathcal M}\pi$}}
 \put(10,38){\vector(0,-1){25}}
 \put(0,22){\mbox{$\mu$}}
\end{picture}
 \end{array}
 \]
and $\tilde\jmath\circ H_{\mathcal P}$ is called a {\sl Hamiltonian section}.
As ${\cal M}\pi$ is a subbundle of $\Lambda^m\Tan^*E$,
it is endowed with a canonical form
$\Theta\in\df^4({\cal M}\pi)$ (the ``tautological form''),
and a canonical multisymplectic form
 $\Omega:=-\d\Theta\in\df^5({\cal M}\pi)$,
which are known as the {\sl multimomentum Liouville forms}.
Then we def\/ine the Hamilton--Cartan forms
\[
\Theta_H=(\tilde\jmath\circ H_{\cal P})^*\Theta\in\df^4({\cal P}) 
\quad ,\quad
\Omega_H=
 -\d\Theta_H=(\tilde\jmath\circ H_{\cal P})^*\Omega\in\df^5({\cal P})  .
 \]
In general, $\Omega_h^0$ is a pre-multisymplectic form.
The Poincar\'e-Cartan forms are $\mathcal{FL}_{\rm EP}^o$-projectable and, in particular,
$\Theta_{\Lag_{\rm EP}}={\mathcal{FL}^o_{\rm EP}}^*\, \Theta_H$ and
$\Omega_{\Lag_{\rm EP}}={\mathcal{FL}^o_{\rm EP}}^*\, \Omega_H$,

In this way we have constructed the Hamiltonian system 
$({\cal P},\Omega_H)$, which is associated with the
almost-regular Lagrangian system .
Then, the variational problem associated with this system  \cite{GPR-2017,art:Roman09}
consists in finding sections $\psi_H\colon M\rightarrow\mathcal{P}$
which are solutions to the equation
$$
\psi_H^*i(X)\Omega_H=0 \quad, \quad 
\mbox{\rm for every $X \in\mathfrak{X}(\mathcal{P})$}\ .
$$
or, what is equivalent, which are integral sections of a multivector field contained 
in a class of $\overline{\tau}_{\cal P}$-transverse integrable multivector fields 
$\{{\bf X}_H\}\subset\mathfrak{X}^4(\mathcal{P})$ such that
\beq{}
\label{eq:hvf}
\inn{({\bf X}_H)}\Omega_H=0 \quad ,\quad
\forall {\bf X}_H\in\{{\bf X}_H\}\subset\mathfrak{X}^4(\mathcal{P}) \ .
\eeq

In order to do a local analysis of the Hamiltonian formalism for this system,
we can use two kinds of coordinates on $\mathcal{P}$:
the so-called {\sl non-momenta} and {\sl pure connection coordinates}.

\subsection{Non-momenta coordinates}
\label{nmc}

Bearing in mind Proposition \ref{PdifE}, we can take
$(x^\lambda,g_{\rho\sigma},\Gamma^\alpha_{\beta\gamma})$ as local coordinates in ${\cal P}$, 
with $0\leq\rho\leq\sigma\leq  3$.
These are the {\sl non-momenta coordinates} of $\mathcal{P}$. 
Using them, the local expression of $\Omega_H$ is the same as that of $\Omega_{\Lag_{\rm EP}}$ 
(see \eqref{pc:la}). 
As a consequence, the Hamiltonian analysis of the system is similar to that in the Lagrangian formalism 
(up to the analysis of the holonomy). 

Note that the functions $L^{\beta\gamma,\mu}_{\alpha}$ and $H$
introduced in \eqref{auxfuns1} and \eqref{auxfuns2} are also
 $\mathcal{FL}_{\rm EP}^o$-projectable and, hence, 
we commit an abuse of notation denoting the corresponding functions of $\Cinfty({\cal P})$ with the same simbols.
Then, for a $\overline{\tau}_{\cal P}$-transverse multivector field ${\bf X}\in\mathfrak{X}^4(\mathcal{P})$,
whose local expression in these coordinates is
$$
\mathbf{X}=\bigwedge_{\nu=0}^3 X_\nu=\bigwedge_{\nu=0}^3\left(
\frac{\partial}{\partial x^\nu}+\sum_{\rho\leq\sigma}f_{\rho\sigma,\nu}\frac{\partial}{\partial g_{\rho\sigma}}+f^\alpha_{\beta\gamma,\nu}\frac{\partial}{\partial \Gamma^\alpha_{\beta\gamma}}
\right) \ ,
$$
the local expression of equation \eqref{eq:hvf} is
\bea
\label{eq:ha1}
\frac{\partial H}{\partial g_{\rho\sigma}}-f^\alpha_{\beta\gamma,\mu} \frac{\partial L_{\alpha}^{\beta\gamma,\mu}}{\partial g_{\rho\sigma}}&=&0,
\\ \label{eq:ha2}
\frac{\partial H}{\partial \Gamma^\alpha_{\beta\gamma} } +\sum_{\rho\leq\sigma}f_{\rho\sigma,\mu} \frac{\partial L_{\alpha}^{\beta\gamma,\mu}}{\partial g_{\rho\sigma}}&=&0,
\eea
together with other equalities which are consequence of these two sets of equations. 
This system of equations is the same as \eqref{eq:fun4} and \eqref{eq:fun5} and,
therefore, the analysis made in Section \ref{se:lanh} is valid here. 

\begin{prop}
[Constraints]
A necessary condition
for the existence of solutions to the system of equations \eqref{eq:ha1} and \eqref{eq:ha2}
(and, in particular, \eqref{eq:ha1})
is that the following equalities hold
$$
T^\alpha_{\beta\gamma}=
\frac13\delta^\alpha_\beta T^\nu_{\nu\gamma}-\frac13\delta^\alpha_\gamma T^\nu_{\nu\beta}\ .
$$
These constraints define the submanifold $\jmath_f\colon\mathcal{P}_f\hookrightarrow \mathcal{P}$.
\label{subham}
\end{prop}
\begin{proof}
The proof is the same than for Propositions 
\ref{torscond} and \ref{pr:toco}. 
They are also the projections of the torsion constraints by the Legendre map.
\end{proof}

Finally, the tangency conditions of ${\bf X}$ for these constraints on ${\cal P}_f$ are
$$
\Lie(X_\nu)(T^\alpha_{\beta\gamma}-\frac13\delta^\alpha_\beta T^\nu_{\nu\gamma}+\frac13\delta^\alpha_\gamma T^\nu_{\nu\beta})=
(f^\alpha_{\beta\gamma,\nu}-\frac13\delta^\alpha_\beta f^\nu_{\nu\gamma,\nu}+\frac13\delta^\alpha_\gamma f^\nu_{\nu\beta,\nu})=0
\quad ; \quad \mbox{(\rm on ${\cal P}_f$)}\ ,
$$
which does not lead to new constraints. Notice that these results about the Hamiltonian constraints are coherent 
with the comment in Section \ref{ccc}
about the fact that, up to the torsion constraints $t^\alpha_{\beta\gamma}$,
all the other Lagrangian constraints appear as a consequence of demanding 
the semiholonomy condition for the solutions to the Lagrangian field equations
and, hence, they cannot be projectable functions under the Legendre map \cite{LMMMR-2002}.
In fact, a simple computation shows that 
$$
\Lie(X)c^{\mu\nu}\not=0 \ , \
\Lie(X)m_{\sigma\rho,\mu}\not= 0 \ , \
\Lie(X)r^\alpha_{\beta\gamma,\nu}\not= 0 \quad ; \quad
\mbox{\rm for some $X\in\ker({\cal FL}_{\rm EP}^o)_*=\ker({\cal FL}_{\rm EP})_*$} \ ,
$$
which are the necessary and sufficient conditions for these functions
not to be ${\cal FL}_{\rm EP}^o$-projectable.
In the same way, the integrability Lagrangian constraints are not
 ${\cal FL}_{\rm EP}^o$-projectable either.

\begin{prop}
[Solutions]
The solutions to the Hamiltonian field equations \eqref{eq:ha1} and \eqref{eq:ha2} are
\bea
\mathbf{X}_H=\bigwedge_{\nu=0}^3 X_\nu&=&
\bigwedge_{\nu=0}^3\left(\frac{\partial}{\partial x^\nu}+(\Gamma^\lambda_{\nu\gamma}\Gamma^\alpha_{\beta\lambda}+
C_{\beta,\nu}\delta^\alpha_\gamma+
K^\alpha_{\beta\gamma,\nu})\frac{\partial}{\partial \Gamma^\alpha_{\beta\gamma}}
\right. \nonumber
\\ & & +\left.\sum_{\rho\leq\sigma}(g_{\sigma\lambda}\Gamma^\lambda_{\mu\rho}+g_{\rho\lambda}\Gamma^\lambda_{\mu\sigma}+
\frac{2}{3}g_{\rho\sigma}T^\lambda_{\lambda\mu})
\frac{\partial}{\partial g_{\rho\sigma}}
\right) 
\quad ; \quad \mbox{(\rm on ${\cal P}_f$)} \ ;
\label{eq:sol:Ha}
\eea
with $C_{\beta,\nu}$, $K^\alpha_{\beta\gamma,\nu}\in C^\infty(\mathcal{P}_f)$ such that,
on the points of ${\cal P}_f$, they satisfy 
\bea
K^\mu_{\mu\gamma,\nu}&=&0 \quad , \quad
K^\mu_{\beta\gamma,\mu}+K^\mu_{\gamma\beta,\mu}\ =\ 0 \quad , 
\label{K1}\\
K^\alpha_{[\beta\gamma],\mu}&=&-\frac13\delta^\alpha_{[\beta} K^\nu_{\gamma]\nu,\mu}-\Gamma^\lambda_{\mu[\gamma}\Gamma^\alpha_{\beta]\lambda}+\frac13\delta^\alpha_{[\beta}\Gamma^\lambda_{\mu\gamma]}\Gamma^\nu_{\nu\lambda}-\frac13\delta^\alpha_{[\beta}\Gamma^\lambda_{\mu\nu}\Gamma^\nu_{\gamma]\lambda}\ .
\label{K2}
\eea
\end{prop}
\begin{proof}
From Proposition \ref{pr:so1} and \eqref{torscond},
we obtain \eqref{eq:sol:Ha} and \eqref{K1},
and the tangency conditions on the torsion constraints lead to obtain \eqref{K2}.
\end{proof}

Finally, the integrability condition is $[X_\mu,X_\nu]\vert_{{\cal P}_f}=0$. 
The vanishing of the coefficients of $\displaystyle\frac{\partial}{\partial g_{\sigma\rho}}$ 
do not lead to new constraints, but they do impose 
new restrictions for the possible solutions:
$$
g_{\alpha\lambda}K^\lambda_{[\nu\beta\mu]}+g_{\beta\lambda}K^\lambda_{[\nu\alpha\mu]}+2g_{\alpha\beta}T^\lambda_{\mu\nu}\Gamma^\sigma_{\sigma\lambda}=0
\quad ; \quad \mbox{(\rm on ${\cal P}_f$)}\ .
$$
The vanishing of the coefficients of $\displaystyle\frac{\partial}{\partial \Gamma^\alpha_{\beta\gamma}}$ 
lead to a system of first order PDE on the functions $C^\alpha_{\beta\gamma\mu}$ and $K^\alpha_{\beta\gamma\mu}$. This system of PDE has solutions everywhere on ${\cal P}_f$, as it is shown in Proposition \ref{par:int2}. 

The following diagram summarizes this situation:
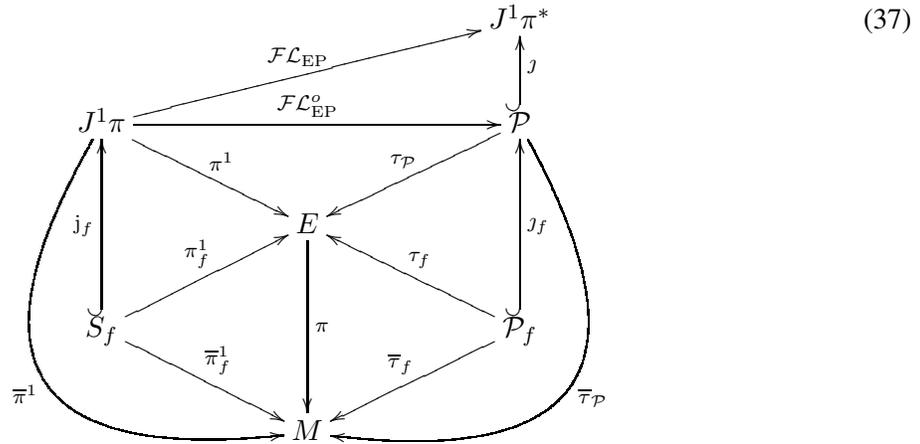
\beq
\xymatrix{
\ & \ &  & \ & J^1\pi^* \\
J^1\pi \ar[urrrr]^<(0.45){\mathcal{FL}_{\rm EP}}
\ar[rrrr]^<(0.45){\mathcal{FL}^o_{\rm EP}} \ar[drr]^<(0.45){\pi^1}
\ar@/_6.0pc/[dddrr]_{\overline{\pi}^1}
 & \ & \ & \ & {\cal P} \ar[dll]_<(0.45){\tau_{\cal P}} \ar@{^{(}->}[u]_<(0.45){\jmath} \ar@/^6.0pc/[dddll]^{\overline{\tau}_{\cal P}}   \\
\ & \ & E \ar[dd]^<(0.45){\pi}
 & \ & \ \\
S_f \ar@{^{(}->}[uu]^<(0.45){{\rm j}_f} \ar[urr]^<(0.45){\pi^1_f} \ar[drr]^<(0.45){\overline{\pi}^1_f}
 & \ & \ & \ & {\cal P}_f \ar@{^{(}->}[uu]_<(0.45){\jmath_f}  \ar[ull]_<(0.45){\tau_f}
\ar[dll]_<(0.45){\overline{\tau}_f}  \\
 & \ & M & \ &
}
\label{diagLH}
\eeq

The study of the gauge vector fields in the Hamiltonian formalism is simpler than in the Lagrangian one. In fact:

\begin{prop}
[Gauge symmetries]
\label{gauge1}
The gauge vector fields of the system are
$$
X=C_\beta\delta^\alpha_\gamma\frac{\partial}{\partial \Gamma^\alpha_{\beta\gamma}}
\quad , \quad 
C_\beta\in C^\infty(\mathcal{P})
\quad ; \quad \mbox{(\rm on ${\cal P}_f$)}\ .
$$
\end{prop}
\begin{proof} 
A $\overline{\tau}$-vertical vector field has the local expression:
$$
X=\sum_{\rho\leq\sigma}f_{\rho\sigma}\frac{\partial}{\partial g_{\rho\sigma}}+f^\alpha_{\beta\gamma}\frac{\partial}{\partial \Gamma^\alpha_{\beta\gamma}}\ .
$$
The analysis of the equation $\inn(X)\Omega_H=0$ is analogous as in Proposition \ref{pr:gv:la}. 
We find that $f_{\rho\sigma}=0$ and $f^\alpha_{\beta\gamma}=C_\beta\delta_\gamma^\alpha+K^\alpha_{\beta \gamma}$, on the points of ${\cal P}_f$; 
that is, they are a combination of a trace  and a torsion solution;
but the torsion solutions are not tangent to $\mathcal{P}_f$.
\end{proof}

The multiple solutions of the system are given by the functions $C^\alpha_{\beta\gamma,\nu}$ and $K^\alpha_{\beta\gamma,\nu}$ (see \eqref{eq:sol:Ha}). The functions $C^\alpha_{\beta\gamma,\nu}$ are related with the gauge freedom, but the former ones $K^\alpha_{\beta\gamma,\nu}$ are not.

\subsection{Pure-connection coordinates}

The non-momenta coordinates arise in a natural way from the structure of the manifolds,
but their use turn out to be very similar to the analysis made in the Lagrangian formalism, 
thus providing little extra understanding about the theory. 
A more interesting coordinates can be obtained from the second set of constraints in
\eqref{hamcons}
\beq\label{eq:ltmc}
 p_\alpha^{\beta\gamma,\mu}=\varrho\left( \delta_\alpha^\mu g^{\beta\gamma}-\delta_\alpha^\beta g^{\mu\gamma}\right) \ ;
\eeq
that is, the momenta of the connection can be obtained from the metric. 
The converse is also true; in fact:

\begin{lemma}\label{le:mc}
Denoting ${\cal T}:=\sqrt{|\det (p_\mu^{\mu\alpha,\beta})|}$, we have that
$$
g^{\alpha\beta}=-\frac1{3\varrho}p_\mu^{\mu\alpha,\beta}=-\frac{3}{{\cal T}}p_\mu^{\mu\alpha,\beta}\ .
$$
\end{lemma}
\begin{proof}
Contracting the indices $\alpha$ and $\beta$ on \eqref{eq:ltmc} we obtain
$$
 p_\nu^{\nu\gamma,\mu}=-3\varrho  g^{\gamma\mu} \ ,
$$
which is the first equality. Now, computing the determinant, 
as $\varrho=\sqrt{|\det (g_{\gamma\mu})|}$, we obtain
that the second equality holds:
$$
|\det (p_\nu^{\nu\gamma,\mu})|=3^4\varrho^4|\det (g_{\gamma\mu})|^{-1}
\ \Longleftrightarrow\ {\cal T}=9\varrho \ ,
$$
\end{proof}

It is interesting to point out that
all the results can be extended to an arbitrary dimension $m>2$; 
but ${\cal T}$ is proportional to $\varrho$ only for $m=4$.

Since the degrees of freedom of $g_{\alpha\beta}$ and 
$p_\alpha^{\beta\gamma,\mu}$ are not equal,  
equation \eqref{eq:ltmc} has several implicit restrictions.
In fact, using Lemma \ref{le:mc} to substitute the metric for momenta in 
\eqref{eq:ltmc} we obtain the constraints
$$
 p_\alpha^{\beta\gamma,\mu}=
\frac13\delta_\alpha^\beta p_\nu^{\nu\mu,\gamma}-
\frac13\delta_\alpha^\mu p_\nu^{\nu\beta,\gamma} \ ,
$$
which are very similar to the torsion constraints. 
Moreover, as $g_{\alpha\beta}=g_{\beta\alpha}$, from Lemma \ref{le:mc}
we have that  $p_\mu^{\mu\alpha,\beta}=p_\mu^{\mu\beta,\alpha}$. 
Therefore, the only degrees of freedom for the momenta of the connection 
are the symmetric part of $p_\mu^{\mu\beta,\alpha}$, 
which equals the degrees of freedom of the metric.

Denoting $p^{\alpha\beta}:=p_r^{r \alpha,\beta}$, 
we can consider the set of coordinates 
$(x^\mu,\Gamma^\alpha_{\beta\gamma},p^{\rho\sigma})$
in $\mathcal{P}$, with $0\leq\rho\leq\sigma\leq 3$,
which are called {\sl pure-connection coordinates}.
The relation between these coordinates  and the non-momenta ones is given
by the following map
$$
\Psi(x^\lambda,g_{\rho\sigma},\Gamma^\alpha_{\beta\gamma})=
(x^\mu,\Gamma^\alpha_{\beta\gamma},p^{\rho\sigma}=
-3\varrho g^{\rho\sigma}) \ ,
$$
which is invertible, and hence a local diffeomorphism, by Lemma \ref{le:mc}.

In pure-connection coordinates the Hamiltonian function has the local expression
$$
H=-\frac13p^{\alpha\beta}\left(\Gamma^{\gamma}_{\beta\sigma}\Gamma^{\sigma}_{\gamma\alpha}-\Gamma^{\gamma}_{\beta\alpha}\Gamma^{\sigma}_{\sigma\gamma}\right),
$$
and the Hamilton-Cartan form $\Omega_H$ is
\beann
\Omega_{H}&=&
\d H\wedge\d^4x +\frac16\delta_\alpha^\mu\d p^{\beta \gamma}\wedge\d \Gamma^{\alpha}_{(\beta\gamma)}\wedge \d^3x_{\mu}
\\ & &-\frac16\delta_\alpha^\beta\d p^{\mu \gamma}\wedge\d \Gamma^{\alpha}_{\beta\mu}\wedge \d^3x_{\gamma} -
\frac16\delta_\alpha^\beta\d p^{\mu \gamma}\wedge\d \Gamma^{\alpha}_{\beta\gamma}\wedge \d^3x_{\mu}\ .
\eeann

A general transverse locally decomposable multivector field in $\mathcal{P}$ has the local expression in pure-connection coordinates:
\beann
\mathbf{X}_H=\bigwedge_{\nu=0}^3 X_\nu&=&
\bigwedge_{\nu=0}^3\left(\frac{\partial}{\partial x^\nu}+f^\alpha_{\beta\gamma,\nu}\frac{\partial}{\partial \Gamma^\alpha_{\beta\gamma}}+\sum_{\alpha\leq\beta}G^{\alpha\beta}_\nu
\frac{\partial}{\partial p^{\alpha\beta}}
\right) \ .
\eeann
Then the field equations \eqref{eq:hvf} are locally
\bea \label{eq:hfun1}
\frac{1}{n(\alpha\beta)}\frac{\partial H}{\partial p^{\alpha\beta}}+\frac16 f^\mu_{(\alpha\beta),\mu}-\frac16f^\mu_{\mu(\alpha,\beta)}=0 \ ,
\\ \label{eq:hfun2}
\frac{\partial H}{\partial \Gamma^\alpha_{\beta\gamma}}-\frac13 G^{\beta\gamma}_{\alpha}+\frac13\delta^\beta_\alpha G^{\mu\gamma}_{\mu}=0 \ .
\eea

Next the results previously described in the above Section \ref{nmc} are recovered and extended:

The constraints and gauge variations are related to the connection, where both the non-momenta and pure-connection coordinates have the same expression. Therefore:

\begin{prop}
[Constraints]
A necessary condition
for the existence of solutions to the system
of equations \eqref{eq:hfun1} and \eqref{eq:hfun2}
(and, in particular, \eqref{eq:hfun2})
is that the following equalities hold
$$
T^\alpha_{\beta\gamma}=
\frac13\delta^\alpha_\beta T^\nu_{\nu\gamma}-\frac13\delta^\alpha_\gamma T^\nu_{\nu\beta}\ .
$$
These constraints define the submanifold $\jmath_f\colon\mathcal{P}_f\hookrightarrow \mathcal{P}$.
\label{subham2}
\end{prop}
\begin{proof}
They are the projections of the torsion constraints by the Legendre map. 
Alternatively, they can be deduced from \eqref{eq:hfun2} imposing that $G_\alpha^{\beta\gamma}-G_\alpha^{\gamma\beta}=0$.
\end{proof}

Taking into account the results presented in the above Section \ref{nmc}, we have:

\begin{prop}
[Solutions]
\label{solham}
The solutions to the Hamiltonian field equation \eqref{eq:hvf}
in the pure-connection coordinates are:
\beann
\mathbf{X}_H=\bigwedge_{\nu=0}^3 X_\nu&=&
\bigwedge_{\nu=0}^3\left(\frac{\partial}{\partial x^\nu}+(\Gamma^\lambda_{\nu\gamma}\Gamma^\alpha_{\beta\lambda}+
C_{\beta,\nu}\delta^\alpha_\gamma+
K^\alpha_{\beta\gamma,\nu})\frac{\partial}{\partial \Gamma^\alpha_{\beta\gamma}}
\right. 
\\ & & +\left.\sum_{\alpha\leq\beta}(-p^{\alpha\mu}\Gamma^\beta_{\nu\mu}-p^{\beta\mu}\Gamma^\alpha_{\nu\mu}-\tfrac13p^{\alpha\beta}T^{\mu}_{\mu\nu}+p^{\alpha\beta}\Gamma^{\mu}_{\mu\nu})
\frac{\partial}{\partial p^{\alpha\beta}}
\right) 
\quad ; \quad \mbox{(\rm on ${\cal P}_f$)} \ ;
\eeann
with $C_{\beta,\nu}$, $K^\alpha_{\beta\gamma,\nu}\in C^\infty(\mathcal{P}_f)$ such that,
on the points of ${\cal P}_f$, they satisfy 
\beann
K^\mu_{\mu\gamma,\nu}&=&0 \quad , \quad
K^\mu_{\beta\gamma,\mu}+K^\mu_{\gamma\beta,\mu}\ =\ 0 \quad , \\
K^\alpha_{[\beta\gamma],\mu}&=&-\frac13\delta^\alpha_{[\beta} K^\nu_{\gamma]\nu,\mu}-\Gamma^\lambda_{\mu[\gamma}\Gamma^\alpha_{\beta]\lambda}+\frac13\delta^\alpha_{[\beta}\Gamma^\lambda_{\mu\gamma]}\Gamma^\nu_{\nu\lambda}-\frac13\delta^\alpha_{[\beta}\Gamma^\lambda_{\mu\nu}\Gamma^\nu_{\gamma]\lambda}\ .
\eeann 
\end{prop}

The integrability condition is
$$
0=[X_\nu,X_\mu]=F^\epsilon\frac{\partial}{\partial x^\epsilon}+F^\alpha_{\beta\gamma}\frac{\partial}{\partial \Gamma^\alpha_{\beta\gamma}}+\sum_{\alpha\leq\beta}F_{\alpha\beta}
\frac{\partial}{\partial p^{\alpha\beta}}
\quad ; \quad \mbox{(\rm on ${\cal P}_f$)}\ .
$$
We have that $F^\epsilon\vert_{{\cal P}_f}=0$, and imposing $F_{\alpha\beta}\vert_{{\cal P}_f}=0$, we derive the following condition on the possible solutions
$$
p^{\alpha\sigma}K^\beta_{[\mu\sigma\nu]}+p^{\beta\sigma}K^\alpha_{[\mu\sigma\nu]}-\frac13p^{\alpha\beta}K^\sigma_{[\mu\sigma\nu]}=\frac23p^{\alpha\beta}T^\lambda_{\nu\mu}\Gamma^\sigma_{\sigma\lambda}
\quad ; \quad \mbox{(\rm on ${\cal P}_f$)}\ .
$$
The conditions $F^\alpha_{\beta\gamma}\vert_{{\cal P}_f}=0$ lead to a system of PDE on the functions $C_{\alpha,\beta}$ and $K^\alpha_{\beta\gamma,\mu}$ which has solutions everywhere on ${\cal P}_f$, as it is shown in Proposition \ref{par:int2}. 

\begin{prop}
[Gauge symmetries]
The gauge variations of the system are:
$$
X=C_\beta\delta^\alpha_\gamma\frac{\partial}{\partial \Gamma^\alpha_{\beta\gamma}}
\quad , \quad 
C_\beta\in C^\infty(\mathcal{P}_f)
\quad ; \quad \mbox{(\rm on ${\cal P}_f$)}\ .
$$
\label{gauge2}
\end{prop}
\begin{proof} For a generic vertical vector field
$$
X=f^\alpha_{\beta\gamma}\frac{\partial}{\partial \Gamma^\alpha_{\beta\gamma}}+\sum_{\alpha\leq\beta}G^{\alpha\beta}\frac{\partial}{\partial p^{\alpha\beta}} \ ,
$$
we have that
\beann
\inn(X)\Omega_H&=&\left(\sum_{\alpha\leq\beta}\frac{\partial H}{\partial p^{\alpha\beta}}G^{\alpha\beta}+\frac{\partial H}{\partial \Gamma^\alpha_{\beta\gamma}}f^\alpha_{\beta\gamma}\right)\d^4x-\left(\frac13\delta_\alpha^\mu G^{\beta\gamma}-\frac13\delta_\alpha^\beta G^{\mu\gamma}\right)\d\Gamma^\alpha_{\beta\gamma}\wedge \d^3x_{\mu}
\\ &&
+\left(\frac16f^\mu_{\alpha\beta}+\frac16f^\mu_{\beta\alpha}-\frac16\delta^\mu_\beta f^\nu_{\nu\alpha}-\frac16\delta^\mu_\alpha f^\nu_{\nu\beta}\right)\d g_{\rho\sigma}\wedge \d^3x_{\mu}=0\ .
\eeann
Doing the pullback to $\mathcal{P}_f$, we have that
$j^*\d \Gamma^\alpha_{\beta\gamma}=\frac12\d \Gamma^\alpha_{(\beta\gamma)}+\frac16\delta^\alpha_\beta \d T^r_{r\gamma}-\frac16\delta^\alpha_\gamma \d T^r_{r\beta}$. As every coefficient must vanish, taking
in particular the corresponding to the factor 
$\d\Gamma^\alpha_{(\beta\gamma)}$ and contracting with $\delta_\mu^\alpha$, we obtain that 
$G^{\beta\gamma}=0$. Therefore we have
\beann
\frac{\partial H}{\partial \Gamma^\alpha_{\beta\gamma}}f^\alpha_{\beta\gamma}&=&0 
\quad ; \quad \mbox{(\rm on ${\cal P}_f$)} \ ,
\\
-\frac16f^\mu_{\alpha\beta}+\frac16f^\mu_{\beta\alpha}+\frac16\delta^\mu_\beta f^\nu_{\nu\alpha}+\frac16\delta^\mu_\alpha f^\nu_{\nu\beta}  &=&  0 
\quad ; \quad \mbox{(\rm on ${\cal P}_f$)} \ .
\eeann
Following the same argument as in Proposition \ref{pr:gv:la}, 
these equations have two kinds of solutions on ${\cal P}_f$: 
trace solutions, $f^\alpha_{\beta\gamma}=C_\beta\delta_\gamma^\alpha$, and torsion solutions,
$f^\alpha_{\beta\gamma}=k^\alpha_{\beta\gamma}$; with $k^\alpha_{\beta\gamma}+k^\alpha_{\gamma\beta}=0$ and $k^\mu_{\mu\gamma}=0$. 
Likewise, only the trace solutions are tangent to $\mathcal{P}_f$.
\end{proof}

\subsection{Intrinsic interpretation of the pure-connection coordinates}

Now we present a fibered manifold and a Hamiltonian function which involve 
only the connection and we prove that this system is equivalent 
to the Hamiltonian formalism for the Metric-Affine action.

The configuration bundle for this pure-connection system is the bundle $\pi_\Gamma\colon E_\Gamma\rightarrow M$, 
where $M$ is the connected orientable 4-dimensional manifold representing space-time, as above,
and $E_\Gamma=C(LM)$, the bundle of connections on $M$;
that is, linear connections in $\Tan M$. Then, consider the bundles 
${\cal M}\pi_\Gamma\equiv\Lambda_2^4(\Tan^*E_\Gamma)$  
and $J^1\pi^*_\Gamma\equiv {\cal M}\pi_\Gamma/\Lambda_1^4(\Tan^*E_\Gamma)$, 
with local coordinates 
$(x^\mu,\Gamma^\alpha_{\beta\gamma},p, p_\alpha^{\beta\gamma,\mu})$ 
and $(x^\mu,\Gamma^\alpha_{\beta\gamma},p_\alpha^{\beta\gamma,\mu})$ respectively.

 Consider a {\sl Hamiltonian section} $h_\Gamma\colon J^1\pi_\Gamma^*\to{\cal M}\pi_\Gamma$ of the projection
 $\mu_\Gamma\colon{\cal M}\pi_\Gamma\to J^1\pi_\Gamma^*$.
 In a local chart of natural coordinates, $U\subset J^1\pi_\Gamma^*$,
 this Hamiltonian section is specif\/ied by a {\sl local Hamiltonian function}
 $H_\Gamma\in\Cinfty (U)$ such that
 $h_\Gamma(x^\mu,\Gamma^\alpha_{\beta\gamma},p_\alpha^{\beta\gamma,\mu})=
(x^\mu,\Gamma^\alpha_{\beta\gamma},p=
-H_\Gamma(x^\nu,\Gamma^\delta_{\rho\sigma},p_\delta^{\rho\sigma,\nu}),p_\alpha^{\beta\gamma,\mu})$
(see \cite{CCI-91,art:Roman09}).
This Hamiltonian function is
$$
H_\Gamma= -\frac13p^{\alpha\beta}\left(\Gamma^{\gamma}_{\beta\sigma}\Gamma^{\sigma}_{\gamma\alpha}-\Gamma^{\gamma}_{\beta\alpha}\Gamma^{\sigma}_{\sigma\gamma}\right) \ .
$$
The bundle ${\cal M}\pi_\Gamma$ is canonically endowed with the
 corresponding multisymplectic Liouville $5$-form $\Omega_\Gamma\in\df^5({\cal M}\pi_\Gamma)$. 
Then, the Hamilton-Cartan form is
$$
\Omega_{H_\Gamma}\equiv h_\Gamma^*\Omega_\Gamma=
\d H\wedge\d^4x -\d p_\alpha^{\beta\gamma,\mu}\wedge\d \Gamma^{\alpha}_{\beta\gamma}\wedge \d^3x_{\mu}\in\df^5(J^1\pi_\Gamma^*).
$$

Furthermore, we introduce the following constraints on $J^1\pi^*_\Gamma$:
$$
p_\alpha^{\beta\gamma,\mu}=\frac13\delta_\alpha^\beta p_\nu^{\nu\mu,\gamma}- \frac13\delta_\alpha^\mu p_\nu^{\nu\beta,\gamma}
\quad , \quad
p_\mu^{\mu\alpha,\beta}=p_\mu^{\mu\beta,\alpha}.
$$
Let $\jmath_\Gamma\colon{\cal P}_{\Gamma}\hookrightarrow J^1\pi^*_\Gamma$ 
be  the submanifold locally defined by these constraints. 
Then we can construct the premultisymplectic form
\beann
\Omega^0_{H_\Gamma}=\jmath_\Gamma^*\Omega_{H_\Gamma}&=&
\d H\wedge\d^4x +\frac16\delta_\alpha^\mu\d p_\nu^{\nu\beta, \gamma}\wedge\d \Gamma^{\alpha}_{(\beta\gamma)}\wedge \d^3x_{\mu}
\\& &-\frac16\delta_\alpha^\beta\d p_\nu^{\nu\mu, \gamma}\wedge\d \Gamma^{\alpha}_{\beta\mu}\wedge \d^3x_{\gamma}  -\frac16\delta_\alpha^\beta\d p_\nu^{\nu\mu ,\gamma}\wedge\d \Gamma^{\alpha}_{\beta\gamma}\wedge \d^3x_{\mu}\ .
\eeann

\begin{prop}
There exists a diffeomorphism 
$\zeta\colon{\cal P}_\Gamma\rightarrow \mathcal{P}$ such that 
$\Omega_{H_\Gamma}^0=\zeta^*\Omega_H$
and hence
the Hamiltonian systems $({\cal P}_\Gamma,\Omega_{H_\Gamma})$ and 
$(\mathcal{P},\Omega_H)$ are equivalents.
\end{prop}
\begin{proof}
Using the pure-connection coordinates in $\mathcal{P}$, the diffeomorphism is locally given by
$$
\zeta^*x^\mu=x^\mu \quad ,\quad
\zeta^*\Gamma^\alpha_{\beta\gamma}=\Gamma^\alpha_{\beta\gamma} \quad ,\quad
\zeta^*p^{\gamma\mu}=p_\nu^{\nu\gamma,\mu} \ .
$$
Its inverse acting on the momenta is given by 
$$
{\zeta^{-1}}^*x^\mu=x^\mu \ ,\
{\zeta^{-1}}^*\Gamma^\alpha_{\beta\gamma}=\Gamma^\alpha_{\beta\gamma} \ ,\
{\zeta^{-1}}^*p_\alpha^{\beta\gamma,\mu}={\zeta^{-1}}^*\left(\frac13\delta_\alpha^\beta p_\nu^{\nu\mu,\gamma}- \frac13\delta_\alpha^\mu p_\nu^{\nu\beta,\gamma}\right)=
\frac13\delta_\alpha^\beta p^{\mu\gamma}- \frac13\delta_\alpha^\mu p^{\beta\gamma} \ ,
$$
and is an exhaustive map because ${\rm Im}(\zeta^{-1})={\cal P}_{\Gamma}$, 
as a consequence of the reasoning done before in this paragraph. 
The equality $\Omega_{H_\Gamma}^0=\zeta^*\Omega_H$ is obtained
straightforwardly from the local expressions of these forms.
\end{proof}

\section{Relation with the Einstein-Hilbert model}
\label{se:rel:he-ma}

The Einstein-Hilbert model can be recovered from the Einstein-Palatini (Metric-Affine) model by demanding the connection to be the Levi-Civita connection associated with the metric \cite{pons}. 
In this section we will show this equivalence geometrically.

\subsection{The Einstein-Hilbert model}

(See \cite{art:GR-2017} for more details and the proofs of the results).

The Lagrangian description of the Einstein-Hilbert model (without energy-matter sources) is developed in the bundle $\pi_\Sigma\colon\Sigma\rightarrow M$, where the fibres are spaces of 
Lorentz metrics on $M$;
that is, for every $x\in M$, the fiber $\pi_\Sigma^{-1}(x)$
is the set of metrics with signature $(-+++)$ acting on $\Tan_xM$.
The adapted fiber coordinates in $E$ are $(x^\mu,g_{\alpha\beta})$. The canonical projections of the jet bundles are $\overline{\pi}^k_\Sigma: J^k\pi_\Sigma\rightarrow M$.
The {\sl Hilbert-Einstein Lagrangian density} (in vacuum) is
$\mathcal{L}_{\mathfrak V}=L_{\mathfrak V}\,\d^4x$, 
being $L_{\mathfrak V}\in\Cinfty(J^2\pi_\Sigma)$ the {\sl Hilbert-Einstein Lagrangian function}, which is again
$L_{\mathfrak V}=\varrho R$,
where $\varrho=\sqrt{|det(g_{\alpha\beta})|}$ and $R$ 
is the {\sl scalar curvature}, but now
the connection is the {\sl Levi-Civita connection} of the metric $g$.

The Lagrangian formalism takes place in the higher-order bundle $J^3\pi_\Sigma$, with local coordinates
$(x^\mu,\,g_{\alpha\beta},\,g_{\alpha\beta,\mu},\,g_{\alpha\beta,\mu\nu},
\,g_{\alpha\beta,\mu\nu\lambda})$, which is endowed with
the {\sl Poincar\'{e}-Cartan $5$-form} associated with $L_{\mathfrak V}$,
denoted by $\Omega_{\mathcal{L}_{\mathfrak V}}\in\df^5(J^3\pi_\Sigma)$, 
and so we have the Lagrangian system $(J^3\pi_\Sigma,\Omega_{\mathcal{L}_{\mathfrak V}})$.
It is a premultisymplectic system since $L_{\mathfrak V}$ is singular
and then, the constraint algorithm leads to a final constraint submanifold 
$S_f\hookrightarrow J^3\pi_\Sigma$ where there are tangent holonomic multivector fields 
which are solutions to the Lagrangian field equations.

The Hamiltonian formalism takes place in the bundle ${\cal P}_\Sigma\to M$,
where $\mathcal{P}_\Sigma=\mathcal{FL}_{\mathfrak V}(J^3\pi_\Sigma)$.
In a similar way as in the Einstein-Palatini model, 
we can construct the Hamilton-Cartan form
$\Omega_{h_{\mathfrak V}}\in\Omega^5(\mathcal{P})$ which verifies that
$\Omega_{\mathcal{L}_{\mathfrak V}}=
{\mathcal{FL}^o_{\mathfrak V}}^*\Omega_{h_{\mathfrak V}}$;
where 
${\mathcal{FL}^o_{\mathfrak V}}\colon J^3\pi_\Sigma\to{\cal P}_\Sigma$
is the restricted Legendre map associated with $\mathcal{L}_{\mathfrak V}$.
So we have the Hamiltonian system 
$({\cal P},\Omega_{h_{\mathfrak V}})$.
The form $\Omega_{h_{\mathfrak V}}$ is multisymplectic and 
then $\mathcal{P}_\Sigma$ is the final constraint submanifold
for the Hamiltonian field equations.
The essential thing is that it can be proved that $\mathcal{P}_\Sigma$ is diffeomorphic to $J^1\pi_\Sigma$
(and hence to ${J^1\pi_\Sigma}^*$).

It is proved \cite{first,rosado2} that there are first-order (regular) Lagrangians in
$J^1\pi_\Sigma$ which are equivalent to the the Hilbert-Einstein Lagrangian 
and that allow us a description of the Einstein-Hilbert model in $J^1\pi_\Sigma$ (with coordinates $(x^\mu,g_{\alpha\beta},g_{\alpha\beta,\mu})$). 
The first-order Lagrangian density proposed in \cite{rosado} is 
$\overline\Lag=\overline{L}\,\d^4x$, where the Lagrangian function is
\beann
\overline{L}&=&L_0-\sum_{\substack{\alpha\leq\beta\\\lambda\leq\sigma}}g_{\alpha\beta,\mu}g_{\lambda \sigma,\nu}\derpar{L^{\alpha\beta,\mu\nu}}{g_{\lambda\sigma}}
\in\Cinfty(J^1\pi_\Sigma) \ ;
\\
\displaystyle L^{\alpha\beta,\mu\nu}&=&
\frac{n(\alpha\beta)}{2}\varrho(g^{\alpha\mu}g^{\beta\nu}+
g^{\alpha\nu}g^{\beta\mu}-2g^{\alpha\beta}g^{\mu\nu})\ ,
\\
L_0&=&\varrho g^{\alpha\beta}\{g^{\gamma\delta}(g_{\delta\mu,\beta}\tilde\Gamma^{\mu}_{\alpha\gamma}-g_{\delta\mu,\gamma}\tilde\Gamma^{\mu}_{\alpha\beta}) +\tilde\Gamma^{\delta}_{\alpha\beta}\tilde\Gamma^{\gamma}_{\gamma\delta}-\tilde\Gamma^{\delta}_{\alpha\gamma}\tilde\Gamma^{\gamma}_{\beta\delta}\} \ ,
\eeann
where $\tilde\Gamma^{\mu}_{\alpha\gamma}$ are the Christoffel symbols of the Levi-Civita connection associated with the metric $g_{\alpha\beta}$.
The corresponding Poincar\'e-Cartan form is
$$
\Omega_{\overline\Lag}=\d \overline{L}\wedge\d^4x-\sum_{\alpha\leq\beta}\d \frac{\partial \overline{L}}{\partial g_{\alpha\beta,\mu}}\wedge\d g_{\alpha\beta}\wedge\d^3x_\mu
\in\df^5(J^1\pi_\Sigma)\ .
$$
So we have the Lagrangian system $(J^1\pi_\Sigma,\Omega_{\overline\Lag})$ and,
as the Lagrangian $\overline{L}$ is regular, then
$\Omega_{\overline\Lag}$ is a multisymplectic form
and the Lagrangian field equations
have solutions everywhere in $J^1\pi_\Sigma$.

In addition, the corresponding Legendre map 
$\mathcal{F\overline L}\colon J^1\pi_\Sigma\to{J^1\pi_\Sigma}^*$
is a diffeomorphism. Then we have the Hamilton-Cartan form 
$\Omega_{\overline{h}}:=((\mathcal{F\overline L})^{-1})^*\Omega_{\overline\Lag}\in\df^5({J^1\pi_\Sigma}^*)$.
So we have the Hamiltonian system $({J^1\pi_\Sigma}^*,\Omega_{\overline{h}})$
and the corresponding Hamiltonian field equations
have solutions everywhere in ${J^1\pi_\Sigma}^*$. In addition, 
the solutions to the Lagrangian problem are in one-to-one correspondence
with thes solution to the Hamiltonian problem through the Legendre map.

\subsection{Relation between the Einstein-Hilbert and the Metric-Affine models}

The pre-metricity constraints determine the derivatives of the metric in function of the metric and the connection. 
The converse, which is a similar result to the existence of the Levi-Civita connection, can be formulated as follows:

\begin{prop}\label{pr:im}
Let $(M,g)$ be a (semi)-Riemmanian manifold of dimension $m>1$ and $C_\alpha\in C^\infty(U)$, $1\leq\alpha\leq m$, fixed functions defined on a open set $U\subset M$. Then there exists a unique linear connection $\Gamma$ defined on $U$ such that:
\begin{enumerate}
\item Pre-metricity: $(\nabla ^\Gamma g)_{\rho\sigma,\mu}=\displaystyle\frac{2}{m-1}g_{\rho\sigma}T^\lambda_{\lambda\mu}$.
\item Torsion: $T^\alpha_{\beta\gamma}=
\displaystyle\frac1{m-1}\,\delta^\alpha_\beta \,T^\lambda_{\lambda\gamma}-
\frac1{m-1}\,\delta^\alpha_\gamma\,T^\lambda_{\lambda\beta}$
\item Gauge fixing: $\Gamma^\lambda_{\alpha\lambda}=C_\alpha$.
\end{enumerate}
\end{prop}
\begin{proof}
From the pre-metricity conditions we have
\beann
\frac12g^{\mu\alpha}(g_{\rho\mu,\sigma}+g_{\sigma\mu,\rho}-g_{\rho\sigma,\mu})&=&
\Gamma^\alpha_{\rho\sigma}+\frac12(g^{\mu\alpha}g_{\rho\lambda}T^\lambda_{\sigma\mu}+g^{\mu\alpha}g_{\sigma\lambda}T^\lambda_{\rho\mu}-T^\alpha_{\rho\sigma})
\\ & &
+\frac1{m-1}(T^\lambda_{\lambda\sigma}\delta^\alpha_\rho+T^\lambda_{\lambda\rho}\delta^\alpha_\sigma-g^{\alpha\mu}g_{\rho\sigma}T^\lambda_{\lambda\mu}) \ .
\eeann
Using the torsion conditions and the gauge fixing we get
$$
\frac12\,g^{\mu\alpha}(g_{\rho\mu,\sigma}+g_{\sigma\mu,\rho}-g_{\rho\sigma,\mu})=\Gamma^\alpha_{\rho\sigma}+\frac{1}{m-1}\Gamma^\lambda_{\lambda\rho}\delta_\sigma^\alpha-\frac{1}{m-1}C_\rho\delta_\sigma^\alpha \ ,
$$
and contracting the indices $\alpha$ and $\rho$ and rearranging the terms:
$$
\frac{1}{m-1}\,\Gamma^\lambda_{\lambda\sigma}=\frac{1}{2m}g^{\mu\nu}g_{\mu\nu,\sigma}+\frac{1}{m(m-1)}C_\sigma \ .
$$
Finally, incorporating this result to the previous equation, we conclude that
$$
\Gamma^\alpha_{\rho\sigma}=\frac12g^{\mu\alpha}(g_{\rho\mu,\sigma}+g_{\sigma\mu,\rho}-g_{\rho\sigma,\mu})-\frac{1}{2m}g^{\mu\nu}g_{\mu\nu,\rho}\delta^\alpha_\sigma+\frac1mC_\rho\delta^\alpha_\sigma \ ,
$$
which determines uniquely the connection in $U$.
\end{proof}

\noindent {\bf Comment}: This proposition is invariant under diffeomorphism in the following sense: it has been shown in Section \ref{sec:lag:sym} that the pre-metricity and torsion conditions are invariant. For the gauge fixing condition 3, consider an infinitesimal Lagrangian symmetry $j^1Y_Z$, and compute de Lie derivative
$$
0=\Lie(j^1Y_Z)(C_\alpha-\Gamma^\lambda_{\alpha\lambda})=f^\mu\frac{\partial C_\alpha}{\partial x^\mu}+\frac{\partial f^\mu}{\partial x^\alpha}\Gamma^{\lambda}_{\mu\lambda}+\frac{\partial^2 f^{\lambda}}{\partial x^\alpha \partial x^\lambda}\ .
$$
Since $f$ is a diffeomorphism in $M$, its Jacobian matrix $J_f$ is invertible:

$$
\Gamma^\lambda_{\alpha\lambda}=-(J^{-1}_f)^\mu_\alpha\left(\frac{\partial^2 f^{\lambda}}{\partial x^\mu \partial x^\lambda}+f^\mu\frac{\partial C_\alpha}{\partial x^\mu}\right)\equiv C'_\alpha\in C^\infty (M) \ .
$$
Therefore, a diffeomorphism in the space-time manifold changes
only the functions $C_\alpha$; that is, the particular choice of a gauge.

In order to establish the relation between both models,
our standpoint is the Hamiltonian formalism of the
Einstein-Palatini model developed in Section \ref{nmc}.
So, let ${\cal P}_f\hookrightarrow{\cal P}$ be
the final constraint submanifold for this last model.
Then, consider the following local map:
\beann
\xi\colon& {\cal P} & \rightarrow  J^1\pi_\Sigma
\\
&(x^\mu,g_{\alpha\beta},\Gamma^\alpha_{\beta\gamma})&\mapsto (x^\mu,g_{\alpha\beta},\overline{g}_{\alpha\beta,\gamma})
\eeann
where $\overline{g}_{\alpha\beta,\gamma}=g_{\alpha\lambda}\Gamma^\lambda_{\mu\beta}+
g_{\beta\lambda}\Gamma^\lambda_{\mu\alpha}+\frac{2}{3}g_{\alpha\beta}T^\lambda_{\lambda\mu}$.
Notice that 
$\overline{\tau}_{\mathcal{P}}\circ \jmath_f=\overline{\pi}^1_\Sigma\circ\xi$.

\begin{lemma}\label{le:im:inj}
Denoting  by ${\cal G}$ the set of gauge variations obtained in Proposition \ref{gauge1},
we have that
$\ker\,\xi_*=~{\cal G}$.
\end{lemma}
\begin{proof}
Consider a generic vector field $X\in\vf({\cal P})$,
tangent to ${\cal P}_f$,
$$
X=f^\mu\frac{\partial}{\partial x^\mu}+\sum_{\alpha\leq\beta}f_{\alpha\beta}\frac{\partial}{\partial g_{\alpha\beta}}+f^\alpha_{\beta\gamma}\frac{\partial}{\partial \Gamma^\alpha_{\beta\gamma}} \ .
$$
If $X\in\ker\,\xi_*$,
then $f^\mu=0$ and $f_{\alpha\beta}=0$. 
For the last coefficients we have:
$$
0=\xi_*X
=g_{\alpha\lambda}f^\lambda_{\gamma\beta}+
g_{\beta\lambda}f^\lambda_{\gamma\alpha}+\frac{2}{3}g_{\alpha\beta}\left(f^\lambda_{\lambda\gamma}-f^\lambda_{\gamma\lambda}\right) \ .
$$
For the coefficients of the form $f^\alpha_{\beta\gamma}=C_\beta\delta^\alpha_\gamma$ for $C_\beta\in\Cinfty({\cal P})$, 
the condition holds. 
Now, for every solution $f^\alpha_{\beta\gamma}$ to these equations, consider $h^\alpha_{\beta\gamma}=f^\alpha_{\beta\gamma}-f^\lambda_{\lambda\beta}\delta^\alpha_{\gamma}$, 
which are also solutions because the equation is linear. Thus
\beq\label{eq:im:aux}
g_{\alpha\lambda}h^\lambda_{\gamma\beta}+
g_{\beta\lambda}h^\lambda_{\gamma\alpha}-\frac{2}{3}g_{\alpha\beta}h^\lambda_{\gamma\lambda}=0 \ .
\eeq
Notice that $h^\nu_{\nu\gamma}=0$. 
Now, contracting with $g^{\alpha\beta}$, we obtain that $h^\lambda_{\gamma\lambda}=0$.
Furthermore, as we are on the points of ${\cal P}_f$, 
where the torsion constraints hold, this implies that $h^\alpha_{\beta\gamma}-h^\alpha_{\gamma\beta}=0$, 
and therefore they are symmetric functions
(for the indices $\beta\gamma$). Now, if 
$S_{\alpha\gamma\beta}:=g_{\alpha\lambda}h^\lambda_{\gamma\beta}$;
taking into account the symmetry of
$h^\alpha_{\beta\gamma}$, we have that $S_{\alpha\gamma\beta}=S_{\alpha\beta\gamma}$, 
and from \eqref{eq:im:aux} we obtain $S_{\alpha\gamma\beta}=-S_{\beta\gamma\alpha}$. 
These two conditions hold simultaneously only if $S_{\alpha\gamma\beta}=0$. 
Therefore, $h^\alpha_{\beta\gamma}=0$, and hence
$\displaystyle
\ker\,\xi_*=~\left\langle C_\beta\delta^\alpha_\gamma\frac{\partial}{\partial \Gamma^\alpha_{\beta\gamma}}
\right\rangle={\cal G}$.
\end{proof}

Let ${\cal P}'_f$ be the manifold obtained making the quotient of ${\cal P}_f$ 
(which is defined by the torsion constraints) 
by the gauge vector fields, and let the natural projection
$\tau_f'\colon{\cal P}_f\rightarrow{\cal P}'_f$.
Then:

\begin{theo}
\label{difeoequiv}
${\cal P}'_f$ is locally diffeomorphic to $J^1\pi_\Sigma$
and hence to ${J^1\pi_\Sigma}^*$.
\end{theo}
\begin{proof}
Consider a smooth section $\varsigma$ of $\tau_f'$, and let 
$\xi':=\xi\circ\varsigma\colon {\cal P}'_f\rightarrow J^1\pi_\Sigma$. 
From lemma \ref{le:im:inj}, $\ker\,\xi_*\supset{\cal G}$; therefore $\xi'$ does not depend on the section chosen. 
Moreover, $\ker\,\xi_*\subset{\cal G}$ and it is injective. Finally, it is exhaustive because for every point of $J^1\pi_\Sigma$, its preimage contains the connection given by proposition \ref{pr:im}. In conclusion, $\xi'$ is a local diffeomorphism and then ${\cal P}'_f$ is (locally) diffeomorphic to $J^1\pi_\Sigma$.
\end{proof}

Then, a simple calculation in coordinates leads to the following result:

\begin{prop}
\label{PCforms}
$\Omega_H=\xi^*\Omega_{\overline\Lag}=(\mathcal{F\overline L}\circ\xi)^*\Omega_{\overline{h}}$.
\end{prop}

\noindent{\bf Comment}:
The comparison between the multiplicity of solutions 
of the Einstein-Hilbert and the Metric-Affine models 
can help us to interpret some of the conditions. 
The multiplicity of the semiholonomic solutions of the Einstein-Hilbert model 
appears in the second derivative of the components of the metric 
(in the Hamiltonian formalism using the non-momentum coordinates). 
They are of the form (see  \cite{art:GR-2017}) $F_{\alpha\beta;\mu,\nu}=\frac12g_{\lambda\sigma}(\Gamma_{\nu \alpha }^\lambda\Gamma_{\mu \beta}^\sigma+\Gamma_{\nu \beta}^\lambda\Gamma_{\mu \alpha }^\sigma)+F^{\mathfrak h}_{\alpha\beta;\mu,\nu}$, where
$$
F^{\mathfrak h}_{\alpha\beta;\mu,\nu}=F^{\mathfrak h}_{\beta\alpha;\mu,\nu}=F^{\mathfrak h}_{\alpha\beta;\nu,\mu}
\quad ,\quad
g^{\alpha\beta}\left(F^{\mathfrak h}_{\eta\tau;\alpha,\beta}+F^{\mathfrak h}_{\alpha\beta;\eta,\tau}-F^{\mathfrak h}_{\alpha\eta;\tau,\beta}-F^{\mathfrak h}_{\alpha\tau;\eta,\beta}\right)=0 \ .
$$
The map $\xi$ transforms any section $\psi$ solution of the Einstein-Palatini model into a solution $\xi^*\psi$ of the Einstein-Hilbert model. The functions $C^\alpha_{\beta\gamma,\mu}$ in \eqref{eq:sol:Ha}, corresponding to the gauge variation, get annihilated by the action of $\xi$. Therefore, we can say that the functions $K^\alpha_{\beta\gamma,\mu}$ (corresponding to $\psi$)  and $F^{\mathfrak h}_{\alpha\beta;\mu,\nu}$ (corresponding to $\xi^*\psi$) are related, as they are in one to one correspondence. Their conditions can be related using this equivalence as it is shown in the following table: supposing that $F^{\mathfrak h}_{\alpha\beta;\mu,\nu}$ and $K^\alpha_{\beta\gamma,\mu}$ are related, we have:
\begin{center}
\begin{tabular}{ c c c }
\text{Metric-Affine} &  & \text{Einstein-Hilbert} 
\\ 
 $K^\lambda_{(\eta\tau)\lambda}=0$& $\Leftrightarrow$ & $g^{\alpha\beta}(F^{\mathfrak h}_{\eta\tau;\alpha,\beta}+F^{\mathfrak h}_{\alpha\beta;\eta,\tau}-F^{\mathfrak h}_{\alpha\eta;\tau,\beta}-F^{\mathfrak h}_{\alpha\tau;\eta,\beta})=0$
 \\
$g_{\alpha\lambda}K^\lambda_{[\nu\beta\mu]}+g_{\beta\lambda}K^\lambda_{[\nu\alpha\mu]}+2g_{\alpha\beta}T^\lambda_{\mu\nu}\Gamma^\sigma_{\sigma\lambda}=0$ &$ \Leftrightarrow$ & $F^{\mathfrak{h}}_{\alpha\beta,[\mu\nu]}=0$  
\\
$K^\lambda_{\lambda\gamma,\mu}=0$&& For any $F^{\mathfrak{h}}_{\alpha\beta,\mu\nu}$
\\
$K^\alpha_{[\beta\gamma],\mu}+\frac13\delta^\alpha_{[\beta} K^\nu_{\gamma]\nu,\mu}+\Gamma^\lambda_{\mu[\gamma}\Gamma^\alpha_{\beta]\lambda}-\frac13\delta^\alpha_{[\beta}\Gamma^\lambda_{\mu\gamma]}\Gamma^\nu_{\nu\lambda}$
\\ $+\frac13\delta^\alpha_{[\beta}\Gamma^\lambda_{\mu\nu}\Gamma^\nu_{\gamma]\lambda}=0$&& For any $F^{\mathfrak{h}}_{\alpha\beta,\mu\nu}$
\\
For any $K^\alpha_{\beta\gamma,\mu}$&& $F^{\mathfrak{h}}_{[\alpha\beta],\mu\nu}=0$
\end{tabular}
\end{center}

\subsection{Integrability}

In the (first-order) Einstein-Hilbert model, every point $p\in J^1\pi_{\Sigma}$ 
is in the image of a section solution to the field equations, 
${\rm Im}(\varphi_p)$, since
$J^1\pi_{\Sigma}$ is the final manifold for this model. 
As a consequence of the equivalence between both models,
$\mathcal{P}_f$ must be also the final constraint submanifold 
for the Einstein-Palatini model; that is: 

\begin{prop}  
For every $q\in\mathcal{P}_f$, there exists a section $\psi_H$ solution to the Hamiltonian field equations of the Metric-Affine model such that $q\in{\rm Im}(\psi_H)$.
\label{par:int1}
\end{prop}
\begin{proof}
Consider the solution $\varphi_{\xi(q)}$ in the Einstein-Hilbert Hamiltonian formalism. 
Moreover, consider $\zeta\colon J^1\pi_{\Sigma}\rightarrow \mathcal{P}_f\subset{\cal P}$ a section of $\xi$ such that $\zeta(\xi(q))=q$ which exists because $\xi$ is exhaustive. 
Therefore $q\in 
{\rm Im}(\zeta\circ\varphi_{\xi(q)})$ and, in order to check that $\zeta\circ\phi_{\xi(q)}$ is a solution, consider an arbitrary $Y\in\mathfrak{X}({\cal P})$; then
\beann{}
(\zeta\circ\varphi_{\xi(q)})^*(i(Y)\Omega_H)&=&(\zeta\circ\varphi_{\xi(q)})^*(\inn(Y)\xi^*\Omega_{\Lag_{\mathfrak{V}}})
\\ &=&
(\xi\circ\zeta\circ\varphi_{\xi(q)})^*(\inn(\xi_*Y)\Omega_{\Lag_{\mathfrak{V}}})=\varphi_{\xi(q)}^*(\inn(\xi_*Y)\Omega_{\Lag_{\mathfrak{V}}})=0 \ ;
\eeann{}
where we have used that $(\xi\circ\zeta)(p)=p$ because it is a section, 
and that $\varphi_{\xi(q)}$ is a solution. Finally, 
$$
\overline{\tau}_{\mathcal{P}}\circ\jmath_f\circ \zeta\circ\varphi_{\xi{q}}=
\overline{\pi}^1_\Sigma\circ\xi\circ \zeta\circ\varphi_{\xi(q)}=
\overline{\pi}^1_\Sigma\circ\varphi_{\xi(q)}={\rm Id}_M \ ;
$$
thus $\psi_H=\zeta\circ\varphi_{\xi(q)}$ is a section of 
$\overline{\tau}_{\mathcal{P}}\circ\jmath_f=\overline{\tau}_f$, 
and hence it is a solution.
\end{proof}

The Lagrangian counterpart of this result also holds, 
although it is not straightforward because we are working with a singular field theory.

\begin{prop}  
For every $p\in\mathcal{S}_f$, there exists a holonomic section $\psi_\Lag$ 
solution to the Lagrangian field equations of the Metric-Affine model such that $p\in{\rm Im}(\psi_\Lag)$.
\label{par:int2}
\end{prop}
\begin{proof}
Consider the diffeomorphism $\tau_{\cal P}\colon \mathcal{P}\rightarrow E$ stated in Proposition \ref{PdifE}
(in particular, it relates the Lagrangian coordinates with the non-momenta coordinates). 
Then we have that $\tau_{\cal P}^{-1}(\pi_f^1(p))\in\mathcal{P}_f$. 
Furthermore there exists a solution to the Hamiltonian field equations $\psi_H$ 
such that $\tau_{\cal P}^{-1}(\pi_f^1(p))\in{\rm Im}(\psi_H)$, 
as it is shown in the above Proposition. 
Then, we are going to prove that the holonomic section 
$\psi_\Lag$ solution in the Lagrangian formalism is 
$\psi_\Lag=j^1(\tau_{\cal P}\circ\psi_H)$.

In fact, first observe that,
for the Metric-Affine model, the fibers of the Legendre map 
$\mathcal{FL}_{\rm EP}^o$ are the vertical fibers of 
$\pi^1\colon J^1\pi\to E$
(since ${\cal P}={\rm Im}\,\mathcal{FL}_{\rm EP}^o$ 
is diffeomorphic to $E$), and then, 
as $\psi_\Lag$ is a canonical lifting to $J^1\pi$ of a section in $E$,
we have that $\mathcal{FL}_{\rm EP}^o\circ\psi_\Lag=\psi_H$.
Furthermore, $\psi_\Lag$ is a solution to the Lagrangian field equations.
Indeed, as $\mathcal{FL}^o_{\rm EP}$ is a submersion, 
we can take a local basis of $\vf(J^1\pi)$
made by vector fields $\{ Y_A,Z_a\}$,
where $Y_A$ are $\mathcal{FL}^o_{\rm EP}$-projectable 
and $Z_a\in\ker\,(\mathcal{FL}^o_{\rm EP})_*$;
and then the vector fields 
$X_A=(\mathcal{FL}^o_{\rm EP})_*Y_A$ are a local basis
for $\vf({\cal P})$. Therefore, taking into account that
$\mathcal{FL}_{\rm EP}^o\circ\psi_\Lag=\psi_H$
and that $\psi_H$ is a solution to the Hamiltonian field equations,
\beann
\psi_\Lag^*\inn(Y_A)\Omega_{\Lag_{EP}}&=&
\psi_\Lag^*\inn(Y_A)(\mathcal{FL}_{\rm EP}^{o\ \ *}\Omega_H)=
\psi_\Lag^*\mathcal{FL}_{\rm EP}^{o\ \ *}\inn(X_A)\Omega_H
\\ &=&
(\mathcal{FL}_{\rm EP}^o\circ\psi_\Lag)^*\inn(X_A)\Omega_H=
\psi_H^*\inn(X)\Omega_H=0\ ;
\eeann
and $\psi_\Lag^*\inn(Z_a)\Omega_{\Lag_{EP}}=0$ trivially.
This allows us to conclude that $\psi_\Lag^*\inn(Y)\Omega_{\Lag_{EP}}=0$,
for every $Y\in\vf(J^1\pi)$,
and hence $\psi_\Lag$ is 
is a solution to the Lagrangian field equations.

Finally, ${\rm Im}\,\psi_\Lag\subset\mathcal{S}_f$. 
Indeed, equations \eqref{eq:ha1} and $\eqref{eq:ha2}$ for 
$\psi_H$ imply that all the points in ${\rm Im}\,\psi_\Lag$
verify the constraints $c^{\mu\nu}$ and $m_{\rho\sigma,\mu}$.
The constraints $r^\alpha_{\beta\gamma,\nu}$ and $i_{\rho\sigma,\mu\nu}$ 
are also satisfied because they arise from the tangency condition on the 
semiholonomic constraints (see Section \ref{semiholcons}) 
and the integrability condition respectively; 
and then they are satisfied for holonomic sections which are solutions to the Lagrangian field equations.

The following diagram summarizes the situation
(see also the diagram \eqref{diagLH}).
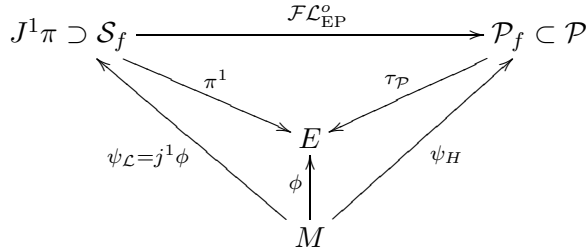
$$
\xymatrix{
J^1\pi\supset{\cal S}_f  
\ar[rrrr]^<(0.45){\mathcal{FL}^o_{\rm EP}} \ar[drr]^<(0.45){\pi^1}
 & \ & \ & \ & {\cal P}_f\subset{\cal P} \ar[dll]_<(0.45){\tau_{\cal P}} \\
\ & \ & E & \ & \ \\
 & \ & M \ar[uull]^{\psi_{\cal L}=j^1\phi} \ar[uurr]_{\psi_H} \ar[u]^<(0.45){\phi}
& \ &}
$$
\end{proof}

\section{Conclusions and outlook}

We have presented a multisymplectic covariant description
of the Lagrangian and Hamiltonian formalisms of the Einstein-Palatini model of General Relativity
(without energy-matter sources).
It is described by a first-order ``metric-affine'' Lagrangian which is (highly) degenerate and hence it originates a theory with constraints and gauge content.
 
The Lagrangian field equations are expressed in terms of holonomic multivector fields 
which are associated with distributions whose integral sections are the solutions to the theory.
Then, we use a constraint algorithm 
to determine a submanifold of the jet bundle $J^1\pi$ where, first,
there exist semi-holonomic multivector fields which are solution to these equations and are tangent to this submanifold, and second, these multivector field are integrable (i.e., holonomic). 
The constraints arising from the algorithm
determine where the image of the sections may lay.

In coordinates, the Lagrangian field equations split into two kinds: 
the {\sl metric} and the {\sl connection equations}
(equations \eqref{eq:fun3}, \eqref{eq:fun4},
\eqref{eq:fun5}).
In the same way, the Lagrangian constraints can be classified into three different types. First there are the {\sl torsion constraints}, which impose strict limitations on the torsion of the connection. Then we have the constraints which  appear as a consequence of demanding the semi-holonomy condition for the multivector field solutions (Theorem \ref{finalconstraints}). In particular, the {\sl Euler-Lagrange equations} themselves
(which appear as constraints of the theory as a consequence of the fact that the Poincar\'e-Cartan form
is $\pi^1$-projectable and the equations are first-order PDE's), and specially the so-called {\sl pre-metricity constraints}, 
which are closely related to the metricity condition for the Levi-Civita connection. Only the tangency condition on the torsion constraint lead also to new constraints. Finally, a family of additional {\sl integrability constraints} 
appear as a consequence of demanding the integrability of the multivector fields 
which are solutions. 
Only the initial torsion constraints are projectable under the Legendre map ${\cal F}\Lag_{\rm EP}$ (because the other ones appear as a consequence of demanding the (semi)holonomy of the solutions), and thus they are the only ones that also appear in the Hamiltonian formalism
(see \cite{GPR-91} for an analysis of this subject for higher-order dynamical theories). We have obtained explicitly all semiholonomic multivector fields solutions to the field equations
(Proposition \ref{ELmvf}).

It is interesting to point out that, although there are regular Lagrangians that are 
equivalent to the Hilbert-Einstein and the Einstein-Palatini Lagrangians
(after a gauge reduction procedure),
and which are then defined in a shorter fiber bundle,
these regular Lagrangians have not a clear physical and/or mathematical interpretation,
as it is the case of those of Hilbert-Einstein and Einstein-Palatini 
where the Lagrangian function is essentially the scalar curvature.

We have done also a brief discussion about symmetries and conserved quantities, 
giving the expression of the natural Lagrangian symmetries, their conserved quantities and the corresponding flows.

The (covariant) multimomentum Hamiltonian formalism for the
Einstein-Palatini model has been also developed.
The final constraint submanifold is also obtained in this formalism, 
and it is defined by the ${\cal F}\Lag_{\rm EP}$-projection of the torsion constraints
(Propositions \ref{subham} and  \ref{subham2}).
The explicit expression of the multivector field solutions is obtained
(Proposition \ref{solham}) and their integrability
is briefly analysed.
The local description is given using two different kinds of coordinates: 
the {\sl non-momenta coordinates} which, as a consequence of the Legendre map, 
are the same as in the Lagrangian case, and the 
{\sl pure-connection coordinates}, where the momenta associated to the connection replace the metric, resulting in metric-free coordinates.
An intrinsic interpretation of these last coordinates
is also given.

Analyzing the gauge content of the model, 
we have obtained the local expression of the natural gauge vector fields, both in the Lagrangian and the Hamiltonian formalisms (Propositions \ref{pr:gv:la} and \ref{gauge2}). 
We have recovered the gauge symmetries discussed in \cite{pons}, showing that there are no more.
As it is known \cite{art:Capriotti2,pons}, it is possible to recover the Einstein-Hilbert model by a gauge fixing in the Einstein-Palatini model, which consists in imposing the trace of the torsion to vanish. This particular gauge fixing transforms the torsion and the pre-metricity constraints, which are a consequence of the constraint algorithm, to  the torsionless and the metricity conditions respectively (Proposition \ref{pr:pm}).
 This equivalence has been studied in detail if a gauge quotient is used instead of a particular gauge fixing (Theorem \ref{difeoequiv} and Propositions  \ref{pr:im} and \ref{PCforms}). We have used this analysis to establish the geometric relation between the Einstein-Palatini and the Einstein-Hilbert models, including  the relation between the holonomic solutions in both formalisms.

Finally, using this equivalence, we have been able to prove that 
the constraint submanifolds ${\cal S}_f$ and ${\cal P}_f$ obtained from
the Lagrangian and Hamiltonian constraint algorithms, respectively
(where there exist multivector fields tangent to them,
satisfying the geometric Lagrangian and Hamiltonian field equations on them)
are the (maximal) {\sl final constraint submanifolds} 
where these multivector fields are integrable; i.e., there are sections solutions to the field equations passing through every point on them
(Propositions \ref{par:int1} and \ref{par:int2}).

In a next paper we will study the Einstein-Palatini model with energy-matter sources,
analyzing how the type of source influences the constraints, the gauge freedom and the symmetries of the theory.

\appendix

\section{Appendix: Symmetries and gauge symmetries of a Lagrangian system}
\label{appends}

In this appendix we state geometrically the basic definitions and results about 
symmetries of Lagrangian field theories (see, for instance, \cite{EMR-96,GPR-2017} for details).
 
Thus, consider a singular Lagrangian system $(J^1\pi,\Omega_\Lag)$, ($\Omega_\Lag\in\df^4({J^1\pi})$),
with final constraint submanifold ${\rm j}_f\colon{\cal S}_f\hookrightarrow J^1\pi$,
and the natural submersions
$\pi^1_f=\pi^1\circ {\rm j}_f\colon{\cal S}_f\to E$, 
$\overline{\pi}^1_f=\overline{\pi}^1\circ {\rm j}_f\colon{\cal S}_f\to M$.
Let  $\Omega_f={\rm j}_f^*\Omega_\Lag$ be the 
{\sl restricted Poincar{\'e}-Cartan form}.

The most relevant kinds of symmetries are the following:

 \begin{definition}
A {\rm Cartan} or {\rm Noether symmetry} of $(J^1\pi,\Omega_\Lag)$
is a diffeomorphism $\Phi\colon J^1\pi\to J^1\pi$
such that $\Phi({\cal S}_f)={\cal S}_f$ and $\Phi^*\Omega_\Lag=\Omega_\Lag$ (on ${\cal S}_f$).
In addition, if $\Phi^*\Theta_\Lag=\Theta_\Lag$ (on ${\cal S}_f$), then 
$\Phi$ is an {\rm exact Cartan symmetry}.
Furthermore, if $\Phi=j^1\varphi$ \ for a diffeormorphism $\varphi\colon E\to E$, 
the Cartan symmetry is said to be {\rm natural}.

An {\rm infinitesimal Cartan} or {\rm Noether symmetry} of $(J^1\pi,\Omega_\Lag)$
is a vector field $X\in\vf (J^1\pi)$ tangent to ${\cal S}_f$ satisfying that $\Lie(X)\Omega_\Lag=0$ (on ${\cal S}_f$).
In addition, if $\Lie(X)\Theta_\Lag=0$ (on ${\cal S}_f$), then $Y$ is an
{\rm infinitesimal exact Cartan symmetry}.
Furthermore, if $X=j^1Y$ for some 
$Y\in\vf (E)$, then the infinitesimal Cartan symmetry is said to be {\rm natural}.
 \end{definition}

Symmetries transform solutions to the field equations into solutions.
In particular, for natural symmetries we have:

\begin{prop}
If $\Phi=j^1\phi\colon J^1\pi\to J^1\pi$,
for a diffeormorphism $\varphi\colon E\to E$,
is a natural Cartan symmetry, 
and ${\bf X}\in\ker^4\Omega_\Lag$ is holonomic,
then $\Phi$ transforms the holonomic sections of ${\bf X}$
into holonomic sections, and hence
$\Phi_*{\bf X}\in\ker^4\Omega_\Lag$ is also holonomic.

As a consequence, if $X=j^1Y\in\vf(J^1\pi)$ is a natural infinitesimal Cartan symmetry,
and $\Phi_t$ is a local flow of $X$, then
$\Phi_t$ transforms the holonomic sections of ${\bf X}$
into holonomic sections.
\label{conserveqs}
\end{prop}
\proof
Let $j^1\varphi\colon M\to J^1\pi$ be an holonomic section of ${\bf X}$,
for $\varphi\colon M\to E$; then it is a solution to the field equations
and  then $(j^1\varphi)^*\inn(X')\Omega_\Lag=0$, for every $X'\in\vf(J^1\pi)$.
Therefore, on the points of ${\cal S}_f$,
\beann
(j^1(\phi\circ\varphi))^*\inn(X')\Omega_\Lag &=&
((j^1\varphi)^*(j^1\phi)^*\inn(X')\Omega_\Lag)=
(j^1\varphi)^*\inn((j^1\phi)_*^{-1}X')(j^1\phi)^*\Omega_\Lag
\\ &=&
(j^1\varphi)^*\inn((j^1\phi)_*^{-1}X')\Omega_\Lag=0 \ ,
\eeann
since  $(j^1\varphi)$ is a solution to the field equations.
Then $j^1(\phi_t\circ\varphi)$ is also a solution to the field equation.
The last statement is immediate since, by definition,
the local flows $\Phi_t\colon J^1\pi\to J^1\pi$ of $j^1Y$
are canonical liftings of the local flows $\phi_t\colon E\to E$ of $Y$.
\qed

In particular, we are specially interested in symmetries of the Lagrangian:

\begin{definition}
A {\rm Lagrangian symmetry} of $(J^1\pi,\Omega_\Lag)$ is a diffeomorphism
$j^1\phi\colon J^1\pi\to J^1\pi$, for some $\phi\in{\rm Diff}(E)$,
 such that $(j^1\phi)({\cal S}_f)=S_f$ and
$(j^1\phi)(\Lag)=\Lag$ (on ${\cal S}_f$).

An {\rm infinitesimal Lagrangian symmetry} of
$(J^1\pi,\Omega_\Lag)$ is a
vector field $j^1Y\in\vf(J^1\pi)$, for some $Y\in\vf(E)$,
such that $j^1Y$ is tangent to ${\cal S}_f$ and
$\Lie (j^1Y)(\Lag)=0$ (on ${\cal S}_f$).
\end{definition}

\noindent {\bf Comment}:
It is well known that canonical liftings of diffeomorphisms and vector
fields preserve the canonical structures of $J^1\pi$.
Therefore, if $j^1\phi\colon J^1\pi\to J^1\pi$ is a Lagrangian symmetry, 
as the Lagrangian density $\Lag$ is invariant, then
$(j^1\phi)^*\Theta_{\Lag} =\Theta_{\Lag}$, and hence
it is an exact Cartan symmetry.
As a consequence, if \ $j^1Y\in\vf(J^1\pi)$ is an infinitesimal Lagrangian symmetry, then
$\Lie (j^1Y)\Theta_{\Lag} =0$, and hence it is
an infinitesimal exact Cartan symmetry.

Symmetries are associated to the existence of
conserved quantities or conservation laws:

\begin{definition}
A {\rm conserved quantity}
of  the Lagrangian system $(J^1\pi,\Omega_\Lag)$
is a form $\xi\in\df^{m-1}(J^1\pi)$ such that
$\Lie({\bf X})\xi=0$ (on ${\cal S}_f$), for every
 ${\bf X}\in\ker^m_{\overline{\pi}^1}\Omega_\Lag$.
 \end{definition}
 
If $\xi\in\df^{m-1}(J^1\pi)$ is a
conserved quantity and ${\bf X}\in\ker^m\Omega_\Lag$ is integrable,
then $\xi$ is closed on the integral submanifolds of ${\bf X}$; that is,
if ${\rm j}_S\colon S\hookrightarrow J^1\pi$ is an integral submanifold, 
then $\d {\rm j}_S^*\xi=0$.
Therefore, for every integral section $\psi\colon M\to J^1\pi$
of ${\bf X}$, in a bounded domain $W\subset M$, Stokes theorem allows to write
$$
\int_{\partial W}{\psi^*\xi}=\int_W\d{\psi^*\xi}=0 \ ;
$$
and the form $\psi^*\xi$ is called the {\sl current} associated with the conserved quantity $\xi$.

Furthermore, {\sl Noether's theorem} in this context states that
if $X\in\vf (J^1\pi)$ is an infinitesimal Cartan symmetry, with $\inn(X)\Omega_\Lag=\d\xi_X$ (on $U\subset J^3\pi$), 
then $\xi_X$ is a conserved quantity.
 As a particular case, if $X$ is an exact infinitesimal Cartan symmetry then
 $\xi_X=\inn(X)\Theta_\Lag$.
For every integral submanifold $\psi$ of ${\bf X}$,
the form $\psi^*\xi_X$ is then called a {\rm Noether current}.

The standard use of the term {\sl gauge} in Physics is for describing 
certain kinds of symmetries which arise as a consequence of 
the non-regularity  of the system (i.e. the Lagrangian function)
and lead to the existence of states (i.e., sections solution to the field equations) that are physically equivalent.
This characteristic is known as {\sl gauge freedom}.
Next we introduce and discuss the geometric concept of these {\sl gauge symmetries}
for Lagrangian field theories, inspired by the geometric treatment given in
\cite{BK-86,GN-79} about {\sl gauge freedom} and {\sl gauge vector fields}
for non-regular dynamical systems.

When a Lagrangian system has gauge symmetries, 
a relevant problem consists in removing the unphysical redundant information 
introduced by the existence of gauge equivalent states.
This is achieved implementing the well-known procedures of reduction.
This procedure rules as follows:
their local generators, which are called `gauge vector fields',
generate an involutive distribution in $\Tan{\cal S}_f$ and hence
we can quotient the manifold ${\cal S}_f$ by this distribution
in order to obtain a quotient set which is made of the true physical degrees of freedom of the theory
and is assumed to be a differentiable manifold $\widetilde{\cal S}_f$.
Furthermore, $\widetilde{\cal S}_f$ is a fiber bundle over $M$, 
with projection $\widetilde\pi_{{\cal S}_f}\colon\widetilde{\cal S}_f\to M$.
and the {\sl real} physical states of the field are the sections of this projection.
This is known as the {\sl gauge reduction procedure} for removing the (unphysical)
gauge degrees of freedom of the theory.
An alternative way to remove the gauge freedom consists in taking 
a (local) section of the projection $\tilde\pi_f$,
and this is called a {\sl gauge fixing}.

Gauge vector fields must have the following properties:

\noindent - Denote $\underline{\vf({\cal S}_f)}:=\{ X\in \vf(J^1\pi)\ \vert\ \mbox{\rm $X$ is tangent to ${\cal S}_f$}\}$.
As the flux of gauge vector fields connect equivalent physical states,
they must be elements of $\underline{\vf({\cal S}_f)}$.

\noindent - As we have said, the existence of gauge symmetries and of
gauge freedom is related to the non-regularity of the Lagrangian $\Lag$ (and conversely).
As a consequence of this, in general the restricted Poincar\'e-Cartan form $\Omega_f$
is degenerated and then it is a pre-multisymplectic form.
Therefore, it is reasonable to think that the gauge reduction procedure, 
which removes  the (unphysical) gauge degrees of freedom, 
must remove also the degeneracy of the form.
Hence, gauge vector fields should be the elements of the set
$$
\underline{\ker\,\Omega_f}:=\{ X\in\underline{\vf({\cal S}_f)}\,\vert\, {\rm j}_f^*\inn(X)\Omega_\Lag=0\} \ ,
$$
or, what is equivalent, if $X^{{\cal S}_f}\in\vf({\cal S}_f)$
is such that ${\rm j}_{f*}X^{{\cal S}_f}=X\vert_{{\cal S}_f}$,
for every $X\in\underline{\vf({\cal S}_f)}$, then
$$
0={\rm j}_f^*\inn(X)\Omega_\Lag=\inn(X^{{\cal S}_f}){\rm j}_f^*\Omega_\Lag=
 \inn(X^{{\cal S}_f})\Omega_f \ ,
$$
and then $X^{{\cal S}_f}\in\ker\,\Omega_f$.
The flux of these vector fields transform solutions to the field equations into solutions, but without preserving the holonomy necessarily.

\noindent - Gauge vector fields must be $\overline{\pi}^1$-vertical
(we denote by $\vf^{V(\overline{\pi}^1)}(J^1\pi)$ the set of $\overline{\pi}^1$-vertical
vector fields).
In this way, we assure that the base manifold $M$ 
does not contain gauge equivalent points and then all the 
gauge degrees of freedom are in the fibres of $J^1\pi$.
Therefore, after doing the reduction procedure
or a gauge fixing in order to remove the gauge multiplicity,
the base manifold $M$ remains unchanged.

\noindent - Furthermore, it is usual to demand that physical symmetries are
{\sl natural}. This means that they are canonical liftings to the 
bundle of phase states of symmetries in the configuration space $E$; 
that is, canonical lifting to $J^1\pi$ of vector fields in $E$.
This condition assures that gauge symmetries transform holonomic solutions to the field equations into holonomic solutions (see Prop. \ref{conserveqs}).

As a consequence of all of this, we define:

\begin{definition} 
$X\in\vf(J^1\pi) $ is a {\rm geometric gauge vector field}
(or a {\rm gauge variation}) of $(J^1\pi,\Omega_\Lag)$ if
\,$X\in\underline{\ker\,\Omega_f}$.
The elements
\,$X\in\underline{\ker\,\Omega_f}\cap\vf^{V(\overline{\pi}^1)}(J^1\pi))$
are the {\rm vertical gauge vector fields}
(or {\rm vertical gauge variations}).
Finally, if \,$X\in\underline{\ker\,\Omega_f}\cap\vf^{V(\overline{\pi}^1)}(J^1\pi))$
and is a natural vector field, it is said to be a {\rm natural gauge vector field}
(or a {\rm natural gauge symmetry}).
\end{definition}

In this paper we are interested only in natural gauge vector fields.

All these definitions and properties can be stated 
in an analogous way for the Hamiltonian system
$({\cal P},\Omega_H)$ associated with $(J^1\pi,\Omega_\Lag)$.


\section*{Acknowledgments}

 We acknowledge the financial support of the
{\sl Ministerio de Ciencia e Innovaci\'on} (Spain), projects
MTM2014--54855--P and
MTM2015-69124--REDT.
the Ministerio de Ciencia, Innovaci\'on y Universidades project
PGC2018-098265-B-C33,
and of {\sl Generalitat de Catalunya}, project 2017--SGR--932.



\begin{thebibliography}{9}
{\small

\bibitem{AA-80}
V. Aldaya, J.A. de Azc\'arraga,
``Geometric formulation of classical mechanics and f\/ield theory'',
{\sl Riv. Nuovo Cimento} {\bf 3}(10) (1980) 1--66.
(doi: 10.1063/1.523904).

\bibitem{BK-86}
M.J. Bergvelt, E.A. de Kerf,
``The Hamiltonian structure of Yang-Mills theories and instantons'' (Part I), 
{\sl Physica} {\bf 139A} (1986) 101--124.
(doi: 10.1088/1361-6382/aa924a).

\bibitem{BMS-2017}
J. Berra-Montiel, A. Molgado, D. Serrano-Blanco,
``De Donder-Weyl Hamiltonian formalism of MacDowell-Mansouri gravity'',
{\sl Class. Quant. Grav.} {\bf 34}(23) (2017) 235002.
(doi: 10.1088/1361-6382/aa924a)

\bibitem{art:Capriotti}
S. Capriotti,
``Differential geometry, Palatini gravity and reduction'',
{\sl J. Math. Phys.} {\bf 55}(1) (2014) 012902.
(doi: 10.1063/1.4862855).

\bibitem{art:Capriotti2}
S. Capriotti,
``Unified formalism for Palatini gravity'',
{\sl Int. J. Geom. Meth. Mod. Phys.} {\bf 15}(3) (2018) 1850044.
(doi:10.1142/S0219887818500445).

\bibitem{CCI-91}
J.F. Cari\~nena, M. Crampin, L.A. Ibort,
``On the multisymplectic formalism for first order field theories'',
{\sl Diff. Geom. Appl.} {\bf 1}(4) (1991) 345--374.
(doi: 10.1016/0926-2245(91)90013-Y).

\bibitem{first}
  M. Castrill\'on, J. Mu\~noz-Masqu\'e, M.E. Rosado,
 ``First-order equivalent to Einstein-Hilbert Lagrangian'', 
 {\sl J. Math. Phys.} {\bf 55}(8) (2014) 082501.
(doi: 10.1063/1.4890555).

\bibitem{CVB-2006}
R. Cianci, S. Vignolo, D. Bruno,
``General Relativity as a constrained Gauge Theory''
{\sl Int. J. Geom. Meth. Mod. Phys.} {\bf 3}(8) (2006) 1493-1500.
(doi: 10.1142/S0219887806001818).

\bibitem{CreTe-2016}
C. Cremaschini, M. Tessarotto,
``Manifest Covariant Hamiltonian Theory of General Relativity'',
{\sl App. Phys. Research} {\bf 8}(2) (2016) 60-81.
(doi: 10.5539/apr.v8n2p60).

\bibitem{CreTe-2016b}
C. Cremaschini, M. Tessarotto,
``Hamiltonian approach to GR-Part 1: covariant theory of classical gravity'',
{\sl Eur. Phys. Journal C} (2017) {\bf 77}:329. 
(doi: 10.1140/epjc/s10052-017-4854-1).

\bibitem{pons} 
N. Dadhich, J.M. Pons,
``On the equivalence of the Einstein–Hilbert and the Einstein–Palatini formulations 
of general relativity for an arbitrary connection”, 
{\sl Gen. Rel. Grav.} {\bf 44}(9) (2012) 2337-2352.
(doi: 10.1007/s10714-012-1393-9).

\bibitem{LMM-96}
M. de Le\'on, J. Mar\'\i n-Solano, J.C. Marrero,
``A geometrical approach to classical field theories: a constraint algorithm for singular theories'',
in New Developments in Dif\/ferential Geometry (Debrecen, 1994), Editors L.~Tamassi and J.~Szenthe,
{\sl Math. Appl.} {\bf 350}, Kluwer Acad. Publ., Dordrecht, 1996, 291--312.
(doi: 10.1007/978-94-009-0149-0\_22).

\bibitem{LMMMR-2002}
M. de Le\'on, J. Mar\'\i n-Solano, J.C. Marrero, M.C. Mu\~noz-Lecanda, N. Rom\'an-Roy, 
``Singular Lagrangian systems on jet bundles'',
\textsl{Fortsch. Phys.} \textbf{50}(2) (2002) 105-169.
(doi: 10.1002/1521-3978(200203)50:2<105::AID-PROP105>3.0.CO;2-N).

\bibitem{LMMMR-2005}
{\rm M. de Le\'on, J. Mar\'\i n-Solano, J.C. Marrero,
 M.C. Mu\~noz-Lecanda, N. Rom\'an-Roy},
``Pre-multisymplectic constraint algorithm for field theories''.
{\sl Int. J. Geom. Meth. Mod. Phys.} {\bf 2}(5) (2005) 839--871.
(doi: 10.1142/S0219887805000880).

\bibitem{EMR-96}
A. Echeverr\'\i a-Enr\'\i quez, M.C. Mu\~noz-Lecanda, N. Rom\'an-Roy,
``Geometry of Lagrangian first-order classical field theories''.
{\sl Forts. Phys.} {\bf 44}(3) (1996) 235-280.
(doi: 10.1002/prop.2190440304).

\bibitem{art:Echeverria_Munoz_Roman98}
A. Echeverr\'\i a-Enr\'\i quez, M.C. Mu\~noz-Lecanda, N. Rom\'an-Roy,
``Multivector fields and connections: Setting Lagrangian equations in field theories'',
{\sl J. Math. Phys.} \textbf{39}(9) (1998) 4578-–4603.
(doi: 10.1063/1.532525).

\bibitem{EMR-99b}
A. Echeverr\'\i a-Enr\'\i quez, M.C. Mu\~noz-Lecanda, N.
 Rom\'an-Roy, ``Multivector Field Formulation of Hamiltonian
 Field Theories: Equations and Symmetries'',
 {\sl J. Phys. A: Math. Gen.} {\bf 32}(48) (1999) 8461-8484.
 (doi: 10.1088/0305-4470/32/48/309).

\bibitem{EMR-00a}
A. Echeverr\'\i a-Enr\'\i quez, M.C. Mu\~noz-Lecanda, N. Rom\'an-Roy,
``On the multimomentum bundles and the Legendre maps in field theories'',
{\sl Rep. Math. Phys.} {\bf 45}(1) (2000), 85-105,
(doi: 10.1016/S0034-4877(00)88873-4).

\bibitem{EMR-2018}
A. Echeverr\'\i a-Enr\'\i quez, M.C. Mu\~noz-Lecanda, N. Rom\'an-Roy,
``Connections and jet fields'',
arXiv:1803.10451 [math.DG] (2018).

\bibitem{Einstein}
A. Einstein, 
``Einheitliche Fieldtheorie von Gravitation und Elektrizität'', 
\textsl{Pruess. Akad.Wiss.} {\bf 414}, (1925); A. Unzicker and T. Case, ``Translation of Einstein’s attempt of a unified field theory with teleparallelism'',
arXiv:physics/0503046[11].

\bibitem{ESG-1995}
G. Esposito, C. Stornaiolo, G. Gionti,
``Spacetime Covariant Form of Ashtekar's Constraints''
{\sl Nuovo Cim.B} {\bf 110}(10) (1995) 1137-1152.
(doi: 10.1007/BF02724605).

\bibitem{Gc-73}
P.L. Garc\'ia,
``The Poincar\'e-Cartan invariant in the calculus of variations'',
{\sl Symp. Math.} {\bf 14} (1973) 219-246.

\bibitem{GPR-2017}
J. Gaset, P.D. Prieto-Mart\'inez, N. Rom\'an-Roy,
``Variational principles and symmetries on fibered multisymplectic manifolds'',
{\sl Comm. in Maths.} {\bf 24}(2) 137-152.
(doi: 10.1515/cm-2016-0010).

\bibitem{art:GR-2016}
J. Gaset, N. Rom\'an-Roy,
``Order reduction, projectability and constraints of second-order field theories and higher-order mechanics'', 
{\sl Rep. Math. Phys.} {\bf 78}(3) (2016) 327-337.
(doi: 10.1063/1.4940047).

\bibitem{art:GR-2017}
J. Gaset, N. Rom\'an-Roy,
``Multisymplectic unified formalism for Einstein-Hilbert Gravity'',
{\sl J. Math. Phys.} {\bf 59}(3) (2018) 032502.
(doi: 10.1063/1.4998526).

\bibitem{GMS-97}
G. Giachetta, L. Mangiarotti, G. Sardanashvily,
{\sl New Lagrangian and Hamiltonian methods in f\/ield theory},
 World Scientif\/ic Publishing Co., Inc., River Edge, NJ, 1997.
(ISBN: 981-02-1587-8.).

\bibitem{GIMMSY-mm}
M.J. Gotay, J. Isenberg, J.E. Marsden, R. Montgomery,
``Momentum maps and classical relativistic fields. I.~Covariant theory'',
arXiv:physics/9801019 [math-ph] (2004).

\bibitem{GN-79}
M.J. Gotay, J-M. Nester,
``Presymplectic Hamilton and Lagrange systems, gauge transformations and the
 Dirac theory of constraints'',
in {\sl Group Theoretical Methods in Physics}; W. Beigelbock, A. B\"ohm, E. Takasugi eds.
{\sl Lect. Notes in Phys.} {\bf 94} 272-279; Springer, Berlin (1979).
(doi: 10.1007/3-540-09238-2\underline{ }74).

\bibitem{GS-73}
H. Goldschmidt, S. Sternberg,
``The Hamilton-Cartan formalism in the calculus of variations'',
{\sl Ann. Inst. Fourier Grenoble} {\bf 23}(1) (1973) 203-267.

\bibitem{GPR-91}
X. Gr\`acia, J.M. Pons, N. Rom\'an-Roy,
``Higher order conditions for singular Lagrangian dynamics'',
{\sl J. Phys. A: Math. Gen.} {\bf 25}(7) (1992) 1989-2004. (doi: 10.1088/0305-4470/25/7/037).

\bibitem{IS-2016}
A. Ibort, A. Spivak,
``On a covariant Hamiltonian description of Palatini's gravity on manifolds with boundary'',
arXiv:1605.03492 [math-ph] (2016).

\bibitem{Ka-q1}
I.V. Kanatchikov,
``Precanonical quantum gravity: quantization without the space-time decomposition'',
{\sl Int. J. Theor. Phys.} {\bf 40}(6) (2001), 1121--1149.
(doi: 10.1023/A:1017557603606).

\bibitem{Ka-q2}
I.V. Kanatchikov,
``On precanonical quantization of gravity'',
{\sl Nonlin. Phenom. Complex Sys.} (NPCS) {\bf 17} (2014) 372-376.

\bibitem{Ka-2016}
I.V. Kanatchikov,
``On the `spin connection foam' picture of quantum gravity from precanonical quantization'',
{\sl Procs. 14th Marcel Grossmann Meeting on General Relativity:
``Recent Developments in Theoretical and Experimental General Relativity, Astrophysics, and Relativistic Field Theories''},
U. Rome “La Sapienza”, Italy 2015, (2017)
 3907-3915.
(doi: 10.1142/9789813226609\_0519).

\bibitem{Krupka}
D. Krupka,
{\it Introduction to Global Variational Geometry}, Atlantis Studies
in Variational Geometry, Atlantis Press 2015, 
(ISBN: 978-94-6239-073-7).

\bibitem{KrupkaStepanova}
D. Krupka, O. Stepankova,
``On the Hamilton form in second order calculus of variations'', 
{\sl Procs. Int. Meeting on Geometry and
Physics}, 85-101. Florence 1982, Pitagora, Bologna, 1983.

\bibitem{MGCD-2017}
M. Montesinos, D. Gonz\'alez, M. Celada, B. D\'iaz,
``Reformulation of the symmetries of first-order general relativity'',
{\sl Class. Quant. Grav.} {\bf 34}(20) (2017) 205002.
(doi: 10.1088/1361-6382/aa89f3)

\bibitem{art:Munoz-rosado2013}
J. Mu\~noz-Masqu\'e, M.E. Rosado.
``Diffeomorphism-invariant covariant Hamiltonians of a pseudo-Riemannian metric and a linear connection'',
{\sl Adv. Theor. Math. Phys.} {\bf 16}(3) (2012) 851–886.
(doi: 10.4310/ATMP.2012.v16.n3.a3).

\bibitem{pere}
  P.D. Prieto Mart\'inez, N. Rom\'an-Roy,
  ``A new multisymplectic unified formalism for second-order classical field theories'',
{\sl J. Geom. Mech.} {\bf 7}(2) (2015) 203-253.
(doi: 10.3934/jgm.2015.7.203).

\bibitem{art:Roman09}
N. Rom\'{a}n-Roy,
``Multisymplectic {Lagrangian} and {Hamiltonian} formalisms of classical field theories'',
\textsl{Symm. Integ. Geom. Methods Appl. (SIGMA)} \textbf{5} (2009) 100, 25pp.
(doi: 10.3842/SIGMA.2009.100).

\bibitem{rosado}
   M.E. Rosado, J. Mu\~noz-Masqu\'e,
``Integrability of second-order Lagrangians admitting a first-order Hamiltonian formalism'',
 {\sl Diff. Geom. and Apps.}
{\bf 35} (Sup. September 2014) (2014) 164-177.
(doi: 10.1016/j.difgeo.2014.04.006).

 \bibitem{rosado2}
  M.E. Rosado, J. Mu\~noz-Masqu\'e,
 ``Second-order Lagrangians admitting
a first-order Hamiltonian formalism'', {\sl J. Annali di Matematica} {\bf 197}(2) (2018) 357-397. 
(doi: 10.1007/s10231-017-0683-y).

\bibitem{rovelli}
C. Rovelli,
``A note on the foundation of relativistic mechanics. II: Covariant Hamiltonian General Relativity'',
 in {\sl Topics in Mathematical Physics, General Relativity and Cosmology}, 
H. Garcia-Compean, B. Mielnik, M. Montesinos, M. Przanowski eds, 
 397, (World Scientific, Singapore) (2006).

\bibitem{Sd-95}
G. Sardanashvily,
{\sl Generalized Hamiltonian formalism for f\/ield theory. Constraint systems},
World Scientif\/ic Publishing Co., Inc., River Edge, NJ, 1995.
(ISBN: 981-02-2045-6).

\bibitem{book:Saunders89}
D.J. {Saunders}, {\it The geometry of jet bundles}, London Mathematical
  Society, Lecture notes series {\bf142}, Cambridge University Press,
  Cambridge, New York 1989.
  (ISBN-13: 978-0521369480).

\bibitem{To-97}
C.G. Torre,
``Local cohomology in field theory (with applications to the Einstein equations)'',
arXiv:hep-th/9706092 (1997).

\bibitem{vey1}
D. Vey,
``Multisymplectic formulation of vielbein gravity. De Donder-Weyl formulation, Hamiltonian $(n-1)$-forms'',
{\sl Class. Quantum Grav.} {\bf 32}(9) (2015) 095005.
(doi: 10.1088/0264-9381/32/ 9/095005).
}
\end{thebibliography}
\end{document}